\documentclass[10pt]{amsart}
\usepackage[T1]{fontenc}
\usepackage{a4wide}
\usepackage[latin1]{inputenc}
\usepackage{amsmath}
\usepackage{amsfonts}
\usepackage{amssymb}
\usepackage{amsthm}
\usepackage{subfigure}
\newtheorem{theo}{Theorem}[section]
\newtheorem{prop}[theo]{Proposition}
\newtheorem{coro}[theo]{Corollary}
\newtheorem{lemm}[theo]{Lemma}
\theoremstyle{definition}
\newtheorem{def1}[theo]{Definition}
\theoremstyle{remark}
\newtheorem*{rema}{Remark}
\newcommand{\Op}{\text{Op}}

\title{Entropy of semiclassical measures in dimension 2}

\author[G. Rivi\`ere]{Gabriel Rivi\`ere}
\address{Centre de Math\'ematiques Laurent Schwartz, \'Ecole Polytechnique, 91128 Palaiseau Cedex, France}
\email{gabriel.riviere@math.polytechnique.fr}

\begin{document}

\maketitle

\begin{abstract}
We study the asymptotic properties of eigenfunctions of the Laplacian in the case of a compact Riemannian surface of Anosov type. We show that the Kolmogorov-Sinai entropy of a semiclassical measure $\mu$ for the geodesic flow $g^t$ is bounded from below by half of the Ruelle upper bound, i.e.
$$h_{KS}(\mu,g)\geq \frac{1}{2}\int_{S^*M} \chi^+(\rho) d\mu(\rho).$$
\end{abstract}

\section{Introduction}

In quantum mechanics, the semiclassical principle asserts that in the high energy limit, one should observe classical phenomena. Our main concern will be the study of this property when the classical system is said to be chaotic.\\
Let $M$ be a compact $\mathcal{C}^{\infty}$ Riemannian surface. For all $x\in M$, $T^*_xM$ is endowed with a norm $\|.\|_x$ given by the metric over $M$. The geodesic flow $g^t$ over $T^*M$ is defined as the Hamiltonian flow corresponding to the Hamiltonian $H(x,\xi):=\frac{\|\xi\|_x^2}{2}$. This last quantity corresponds to the classical kinetic energy in the case of the absence of potential. As any observable, this quantity can be quantized via pseudodifferential calculus and the quantum operator corresponding to $H$ is $-\frac{\hbar^2\Delta}{2}$ where $\hbar$ is proportional to the Planck constant and $\Delta$ is the Laplace Beltrami operator acting on $L^2(M)$.\\
Our main result concerns the influence of the classical Hamiltonian behavior on the spectral asymptotic properties of $\Delta$. More precisely, our main interest is the study of the measure $|\psi_{\hbar}(x)|^2dx$ where $\psi_{\hbar}$ is an eigenfunction of $-\frac{\hbar^2\Delta}{2}$ associated to the eigenvalue $\frac{1}{2}$, i.e.
$$-\hbar^2\Delta\psi_{\hbar}=\psi_{\hbar}.$$
This is equivalent to the study of large eigenvalues of $\Delta$. As $M$ is a compact Riemannian manifold, the family $-\hbar^{-2}$ forms a discrete subsequence that tends to infinity. One natural question is to study the (weak) limits of the probability measure $|\psi_{\hbar}(x)|^2dx$ as $\hbar$ tends to $0$. This means studying the asymptotic behavior of the probability to find a particle in $x$ when the system is in the state~$\psi_{\hbar}$. In order to study the influence of the Hamiltonian flow, we first need to lift this measure to the cotangent bundle. This can be achieved thanks to pseudodifferential calculus. In fact there exists a procedure of quantization that gives us an operator $\Op_{\hbar}(a)$ on the phase space $L^2(M)$ for any observable $a(x,\xi)$ in a certain class of symbols. Then a natural way to lift the previous measure is to define the following quantity:
$$\mu_{\hbar}(a)=\int_{T^*M}a(x,\xi)d\mu_{\hbar}(x,\xi):=\langle \psi_{\hbar},\Op_{\hbar}(a)\psi_{\hbar}\rangle_{L^2(M)}.$$
This formula gives a distribution $\mu_{\hbar}$ on the space $T^{*}M$ and describes now the distribution in position and velocity.\\
Let $(\psi_{\hbar_k})$ be a sequence of orthonormal eigenfunctions of the Laplacian corresponding to the eigenvalues $-\hbar_k^{-2}$ such that the corresponding sequence of distributions $\mu_k$ on $T^{*}M$ converges as $k$ tends to infinity to a limit $\mu$. Such a limit is called a semiclassical measure. Using standard facts of pseudodifferential calculus, it can be shown that $\mu$ is a probability measure that does not depend on the choice of the quantization $\Op_{\hbar}$ and that is carried on the unit energy layer
$$S^*M:=\left\{(x,\xi):H(x,\xi)=\frac{1}{2}\right\}.$$
Moreover, another result from semiclassical analysis, known as the Egorov property, states that for any fixed $t$,
\begin{equation}\label{Egorov}
\forall a\in\mathcal{C}^{\infty}_c(T^*M),\ U^{-t}\Op_{\hbar}(a)U^{t}=\Op_{\hbar}(a\circ g^t)+\mathcal{O}_t(\hbar),
\end{equation}
where $U^t$ denotes the quantum propagator $e^{\frac{\imath t\hbar\Delta}{2}}$. Precisely, it says that for fixed times, the quantum evolution is related to the classical evolution under the geodesic flow. From this, it can be deduced that $\mu$ is invariant under the geodesic flow. One natural question to ask is what measures supported on $S^{*}M$ are in fact semiclassical measures. The corresponding question in quantum chaos is: when the classical behavior is said to be chaotic, what is the set of semiclassical measures? A first result in this direction has been found by Shnirelman~\cite{Sc}, Zelditch~\cite{Ze}, Colin de Verdi\`ere~\cite{CdV}:
\begin{theo}
Let $(\psi_k)$ be an orthonormal basis of $L^2(M)$ composed of eigenfunctions of the Laplacian. Moreover, suppose the geodesic flow on $S^*M$ is ergodic with respect to Liouville measure. Then, there exists a subsequence $(\mu_{k_p})_p$ of density one that converges to the Liouville measure on $S^*M$ as $p$ tends to infinity.
\end{theo}
By 'density one', we mean that $\frac{1}{n}\sharp\{p:1\leq k_p\leq n\}$ tends to one as $n$ tends to infinity. This theorem states that, in the case of an ergodic geodesic flow, almost all eigenfunctions concentrate on the Liouville measure in the high energy limit. This phenomenon is called quantum ergodicity and has many extensions. The Quantum Unique Ergodicity Conjecture states that the set of semiclassical measures should be reduced to the Liouville measure in the case of Anosov geodesic flow~\cite{RS}. This question still remains widely open. In fact, in the case of negative curvature, there are many measures invariant under the geodesic flow: for example, there exists an infinity of closed geodesics (each of them carrying naturally an invariant measure). In recent papers, Lindenstrauss proved a particular form of the conjecture, the Arithmetic Quantum Unique Ergodicity~\cite{Li}. Precisely, he proved that for a sequence of Hecke eigenfunctions of the Laplacian on an arithmetic surface, $|\psi|^2dx$ converges to the Lebesgue measure on the surface. This result is actually the best-known positive result towards the conjecture.\\
In order to understand the phenomenon of quantum chaos, many people started to study toy models as the cat map (a typical hyperbolic automorphism of $\mathbb{T}^2$). These dynamical systems provide systems with similar dynamical properties to the geodesic flow on a manifold of negative curvature. Moreover, they can be quantized using Weyl formalism and the question of Quantum Ergodicity naturally arises. For example, Bouzouina and de Bi\`evre proved the Quantum Ergodicity property for the quantized cat map~\cite{BdB}. However, de Bi\`evre, Faure and Nonnenmacher proved that in this case, the Quantum Unique Ergodicity is too optimistic~\cite{FNdB}. In fact, they constructed a sequence of eigenfunctions that converges to $\frac{1}{2}(\delta_0+\text{Leb})$, where $\delta_0$ is the Dirac measure on $0$ and $\text{Leb}$ is the Lebesgue measure on $\mathbb{T}^2$. Faure and Nonnenmacher also proved that if we split the semiclassical measure into its pure point, Lebesgue and singular continuous components, $\mu=\mu_{\text{pp}}+\mu_{\text{Leb}}+\mu_{\text{sc}}$, then $\mu_{\text{pp}}(\mathbb{T}^2)\leq\mu_{\text{Leb}}(\mathbb{T}^2)$ and in particular $\mu_{\text{pp}}(\mathbb{T}^2)\leq 1/2$~\cite{FN}. As in the case of geodesic flow, there is an arithmetic point of view on this problem. Recently, Kelmer proved that in the case of $\mathbb{T}^{2d}$ ($d\geq 2$, for a generic family of symplectic matrices), either there exists isotropic submanifold invariant under the $2d$ cat map or one has Arithmetic Quantum Unique Ergodicity~\cite{Ke}. Moreover, in the first case, he showed that we can construct semiclassical measure equal to Lebesgue on the isotropic submanifold.

\subsection{Statement of the main result}

In recent papers~\cite{An},~\cite{AN2}, Anantharaman and Nonnenmacher got concerned with the study of the localization of eigenfunctions on $M$ as in the case of the toy models. They tried to understand it via the Kolmogorov-Sinai entropy. This paper is in the same spirit and our main result gives an information on the set of semiclassical measures in the case of a surface $M$ of Anosov type. More precisely, we give an information on the localization (or complexity) of a semiclassical measure:
\begin{theo}\label{maintheo} Let $M$ be a $\mathcal{C}^{\infty}$ Riemannian surface and $\mu$ a semiclassical measure. Suppose the geodesic flow $(g^t)_t$ has the Anosov property. Then,
\begin{equation}\label{mainineq}h_{KS}(\mu,g)\geq \frac{1}{2}\left|\int_{S^*M} \log J^u(\rho) d\mu(\rho)\right|,\end{equation}
where $J^u(\rho)$ is the unstable Jacobian at the point $\rho$.
\end{theo}
We recall that the lower bound can be expressed in term of the Lyapunov exponent~\cite{BP} as \begin{equation}\label{mainineq3}h_{KS}(\mu,g)\geq \frac{1}{2}\int_{S^*M}\chi^+(\rho)d\mu(\rho),\end{equation}
where $\chi^+(\rho)$ is the upper Lyapunov exponent at the point $\rho$~\cite{BP}. In order to comment this result, let us recall a few facts about the Kolmogorov-Sinai (also called metric) entropy. It is a nonnegative number associated to a flow $g$ and a $g$-invariant measure $\mu$, that estimates the complexity of $\mu$ with respect to this flow. For example, a measure carried by a closed geodesic will have entropy zero. In particular, this theorem shows that the support of a semiclassical measure cannot be reduced to a closed geodesic. Moreover, this lower bound seems to be the optimal result we can prove using this method and only the dynamical properties of $M$. In fact, in the case of the toy models some of the counterexamples that have been constructed (see~\cite{FNdB},~\cite{Ke},~\cite{G}) have entropy equal to $\displaystyle\frac{1}{2}\int_{S^*M} \chi^+(\rho) d\mu(\rho).$ Recall also that a standard theorem of dynamical systems due to Ruelle~\cite{R} asserts that, for any invariant measure $\mu$ under the geodesic flow,
\begin{equation}\label{ruelle}h_{KS}(\mu,g)\leq \int_{S^*M} \chi^+(\rho) d\mu(\rho)\end{equation}
with equality if and only if $\mu$ is the Liouville measure in the case of an Anosov flow~\cite{LY}.
\\The lower bound of theorem~\ref{maintheo} was conjectured to hold for any semiclassical measure for an Anosov manifold in any dimension by Anantharaman~\cite{An}. In fact, Anantharaman proved that in any dimension, the entropy of a semiclassical measure should be bounded from below by a (not really explicit) positive constant~\cite{An}. Then, Anantharaman and Nonnenmacher showed that inequality~(\ref{mainineq3}) holds in the case of the Walsh Baker's map~\cite{AN1} and in the case of constant negative curvature in all dimension~\cite{AN2}. In the general case of an Anosov flow on a manifold of dimension $d$, Anantharaman, Koch and Nonnenmacher~\cite{AKN} proved a lower bound using the same method:
$$h_{KS}(\mu,g)\geq \int_{S^*M} \sum_{j=1}^{d-1}\chi^+_j(\rho) d\mu(\rho)-\frac{(d-1)\lambda_{\max}}{2}.$$
where $\lambda_{\max}:=\lim_{t\rightarrow\pm\infty}\frac{1}{t}\log\sup_{\rho\in S^*M}|d_{\rho}g^t|$ is the maximal expansion rate of the geodesic flow and the $\chi^+_j$'s are the positive Lyapunov exponents~\cite{BP}. In particular if $\lambda_{\max}$ is very large, the previous inequality can be trivial. However, they conjectured inequality~(\ref{mainineq3}) should hold in the general case of manifolds of Anosov type by replacing $\chi^{+}$ by the sum of nonnegative Lyapunov exponents~\cite{AN2},~\cite{AKN}. Our main result answers this conjecture in the particular case of surfaces of Anosov type and our proof is really specific to the case of dimension $2$. Now let us discuss briefly the main ideas of our proof of theorem~\ref{maintheo}.

\subsection{Heuristic of the proof}
\label{heuristic}

The procedure developed in~\cite{AKN} uses a result known as the entropic uncertainty principle~\cite{MU}. To use this principle in the semiclassical limit, we need to understand the precise link between the classical evolution and the quantum one for large times. Typically, we have to understand Egorov theorem~(\ref{Egorov}) for large range of times of order $t\sim|\log\hbar|$ (i.e. have a uniform remainder term of~(\ref{Egorov}) for a large range of times). For a general symbol $a$ in $\mathcal{C}^{\infty}_c(T^*M)$, we can only expect to have a uniform Egorov property for times $t$ in the range of times $[-\frac{1}{2}|\log\hbar|/\lambda_{\max},\frac{1}{2}|\log\hbar|/\lambda_{\max}]$~\cite{BR}. However, if we only consider this range of times, we do not take into account that the unstable jacobian can be very different between two points of $S^*M$. In this paper, we would like to say that the range of times for which the Egorov property holds depends also on the support of the symbol $a(x,\xi)$ we consider. For particular families of symbol of small support (that depends on $\hbar$), we show that we have a 'local' Egorov theorem with an allowed range of times that depends on our symbol (see~(\ref{mainegorov}) for example). To make this heuristic idea work, we first try to reparametrize the flow~\cite{CFS} in order to have a uniform expansion rate on the manifold. We define $\overline{g}^{\tau}(\rho):=g^t(\rho)$ where
\begin{equation}\label{parameter}\tau:=-\int_0^t\log J^u(g^s\rho)ds.\end{equation}
This new flow $\overline{g}$ has the same trajectories as $g$. However, the 'velocity of motion' along the trajectory at $\rho$ is $|\log J^u(\rho)|$-greater for $\overline{g}$ than for $g$. We underline here that the unstable direction is of dimension $1$ (as $M$ is a surface) and it is crucial because it implies that $\log J^u$ exactly measures the expansion rate in the unstable direction at each point\footnote{In fact, for the Anosov case, the crucial point is that at each point $\rho$ of $S^*M$, the expansion rate is the same in any direction, i.e. $dg^{-1}_{|E^u(g^1\rho)}$ is of the form $J^{u}(\rho)^{\frac{1}{d-1}}v_{\rho}$ where $d$ is the dimension of the manifold $M$ and $v_{\rho}$ is an isometry. The proof of theorem~\ref{maintheo} can be immediately adapted to Anosov manifolds of higher dimensions satisfying this isotropic expansion property (for example manifolds of constant negative curvature).}. As a consequence, this new flow $\overline{g}$ has a uniform expansion rate. Once this reparametrization is done, we use the following formula to recover $t$ knowing $\tau$:
\begin{equation}\label{inverse}t_{\tau}(\rho)=\inf\left\{s>0:-\int_0^s\log J^u(g^{s'}\rho)ds'\geq\tau\right\}.\end{equation}
The number $t_{\tau}(\rho)$ can be thought of as a stopping time corresponding to $\rho$. We consider now $\tau=\frac{1}{2}|\log\hbar|$. For a given symbol $a(x,\xi)$ localized near a point $\rho$, $t_{\frac{1}{2}|\log\hbar|}(\rho)$ is exactly the range of times for which we can expect Egorov to hold. This new flow seems in a way more adapted to our problem. Moreover, we can define a $\overline{g}$-invariant measure $\overline{\mu}$ corresponding to $\mu$~\cite{CFS}. The measure $\overline{\mu}$ is absolutely continuous with respect to $\mu$ and verifies $\frac{d\overline{\mu}}{d\mu}(\rho)=\log J^u(\rho)/\int_{S^*M}\log J^u(\rho)d\mu(\rho)$. We can apply the classical result of Abramov
$$h_{KS}(\mu,g)=\left|\int_{S^*M}\log J^u(\rho)d\mu(\rho)\right|h_{KS}(\overline{\mu},\overline{g}).$$
To prove theorem~\ref{maintheo}, we would have to show that $h_{KS}(\overline{\mu},\overline{g})\geq1/2$. However, the flow $\overline{g}$ has no reason to be a Hamiltonian flow to which corresponds a quantum propagator $\overline{U}$. As a consequence, there is no particular reason that this inequality should be a consequence of~\cite{AN2}. In the quantum case, there is also no obvious reparametrization we can make as in the classical case. However, we will reparametrize the quantum propagator starting from a discrete reparametrization of the geodesic flow and by introducing a small parameter of time $\eta$. To have an artificial discrete reparametrization of the geodesic flow, we will introduce a suspension set~\cite{CFS}. Then, in this setting, we will define discrete analogues of the previous quantities~(\ref{parameter}) and~(\ref{inverse}) that will be precised in the paper. It will allow us to prove a lower bound on the entropy of a certain reparametrized flow and then using Abramov theorem~\cite{Ab} deduce the expected lower bound on the entropy of a semiclassical measure.\\
Finally, we would like to underline that in a recent paper~\cite{G}, Gutkin also used a version of the Abramov theorem to prove an analogue of theorem~\ref{maintheo} in the case of toy models with an unstable direction of dimension $1$.

\subsection{Extension of theorem~\ref{maintheo}}

Finally, we would like to discuss other classes of dynamical systems for which it could be interesting to get an analogue of theorem~\ref{maintheo}. For instance, regarding the counterexamples in~\cite{Has}, it would be important to derive an extension of theorem~\ref{maintheo} to ergodic billiards. A first step in this direction should be to study the case of surfaces of nonpositive curvature. For the sake of simplicity, we will not discuss the details of this extension in this article and refer the reader to~\cite{GR2} for a more detailled discussion. However, we would like to point out that surfaces of nonpositive curvature share enough properties with Anosov manifolds so that this extension should not be so surprising. First, one can introduce a new quantity that replaces the unstable Jacobian in our proof. This quantity comes from the study of Jacobi fields and is called the unstable Riccati solution $U^u(\rho)$~\cite{Ru},~\cite{Eb}. In this setting, it has been shown that the Ruelle inequality can be rewritten as follows~\cite{FM}:
$$h_{KS}(\mu,g)\leq \int_{S^*M} U^u(\rho) d\mu(\rho).$$
So, a natural extension of theorem~\ref{maintheo} would be to prove that, for a smooth Riemannian surface $M$ of nonpositive sectional curvature and a semiclassical measure $\mu$,
\begin{equation}\label{mainineq2}h_{KS}(\mu,g)\geq \frac{1}{2}\int_{S^*M} U^u(\rho) d\mu(\rho).\end{equation}
In particular, this result would show that the support of a semiclassical measure cannot be reduced to a closed unstable geodesic. We underline that this inequality is also coherent with the quasimodes constructed by Donnelly~\cite{Do}. In fact, his quasimodes are supported on closed stable geodesics (included in flat parts of a surface of nonpositive curvature) and have zero entropy. We can make a last observation on the assumptions on the manifold: it is not known whether the geodesic flow is ergodic or not for the Liouville measure on a surface of nonpositive curvature. The best result in this direction is that there exists an open invariant subset $U$ of positive Liouville measure such that the restriction $g_{|U}$ is ergodic with respect to Liouville~\cite{BP}. So, the entropic properties of semiclassical measures still seem to hold for weakly chaotic systems.\\
We would like to highlight what are the specific properties of surfaces of nonpositive curvature that can be exploited to get inequality~(\ref{mainineq2}). A crucial property that is used in the proof of theorem~\ref{maintheo} is that there exist \emph{continuous stable and unstable foliations}. This property was already at the heart of~\cite{An},~\cite{AN2} and~\cite{AKN}. Another property that is crucially used is the fact that Anosov manifolds have \emph{no conjugate points}. A nice fact about manifolds of nonpositive curvature is that these two properties remain true with the notable difference that the stable and unstable manifolds are not anymore uniformly transverse. Our main affirmation is that these two properties are the \emph{crucial dynamical properties} that make the different proofs from~\cite{AN2},~\cite{AKN} and this article work. In particular, one can use results about \emph{uniform divergence of vanishing Jacobi fields}~\cite{Ru} to derive the main inequality from~\cite{AN2} (section $3$ of this reference). We do not give the points that need to be modified and refer the reader to~\cite{GR2} for a more detailed discussion. Another notable difference with the present article relies on the introduction of a thermodynamical setting at the quantum level as in~\cite{AN2} and~\cite{AKN} to get optimal estimates with the uncertainty principle~\cite{GR2}.
\begin{rema} One could also ask whether it would be possible to extend this result to surfaces without conjugate points. In fact, these surfaces also have a stable and unstable foliations (and of course no conjugate points). Moreover, according to Green~\cite{Gr} and Eberlein~\cite{Eb0}, the Jacobi fields also satisfy a property of uniform divergence (at least in dimension $2$). The main difficulty is that the continuity of the stable and unstable foliations is not true anymore~\cite{BBB} and at this point, we do not see any way of escaping this difficulty.
\end{rema}

\subsection{Organization of the paper}

In section~\ref{entropysection}, we briefly recall properties we will need about entropy in the classical and quantum settings. In particular, we recall the version of Abramov theorem we will need. In section~\ref{Anosov}, we describe the assumptions we make on the manifold $M$ and introduce some notations. In section~\ref{proof}, we draw a precise outline of the proof of theorem~\ref{maintheo} and state some results that we will prove in the following sections. Sections~\ref{partition} and~\ref{commutativity} are devoted to the detailed proofs of the results we admitted in section~\ref{proof}. Sections~\ref{bigpdotheo} and~appendix~\ref{appendix} are devoted to results of semiclassical analysis that are quite technical and that we will use at different points of the paper (in particular in section~\ref{commutativity}).

\subsection*{Acknowledgments} First of all, I am very grateful to my advisor Nalini Anantharaman for her time and her patience spent to teach me so many things about the subject. I also thank her for having read carefully preliminary versions of this work and for her support. I would also like to thank warmly St\'ephane Nonnenmacher for enlightening explanations about semiclassical analysis and more generally for his encouragement. I am grateful to Herbert Koch for helpful and stimulating suggestions about the application of the entropic uncertainty principle. Finally, I would like to thank the anonymous referees for precious comments and suggestions to improve the presentation of this article.

\section{Classical and quantum entropy}
\label{entropysection}

\subsection{Kolmogorov-Sinai entropy}
\label{KSentropy}
Let us recall a few facts about Kolmogorov-Sinai (or metric) entropy that can be found for example in~\cite{Wa}. Let $(X,\mathcal{B},\mu)$ be a measurable probability space, $I$ a finite set and $P:=(P_{\alpha})_{\alpha\in I}$ a finite measurable partition of $X$, i.e. a finite collection of measurable subsets that forms a partition. Each $P_{\alpha}$ is called an atom of the partition. Assuming $0\log 0=0$, one defines the entropy of the partition as
\begin{equation}\label{defent0}H(\mu,P):=-\sum_{\alpha\in I}\mu(P_{\alpha})\log\mu(P_{\alpha})\geq 0.\end{equation}
Given two measurable partitions $P:=(P_{\alpha})_{\alpha\in I}$ and $Q:=(Q_{\beta})_{\beta\in K}$, one says that $P$ is a refinement of $Q$ if every element of $Q$ can be written as the union of elements of $P$ and it can be shown that $H(\mu,Q)\leq H(\mu,P)$. Otherwise, one denotes $P\vee Q:=(P_{\alpha}\cap Q_{\beta})_{\alpha\in I,\beta\in K}$ their join (which is still a partition) and one has $H(\mu,P\vee Q)\leq H(\mu,P)+H(\mu,Q)$ (subadditivity property). Let $T$ be a measure preserving transformation of $X$. The $n$-refined partition $\vee_{i=0}^{n-1}T^{-i}P$ of $P$ with respect to $T$ is then the partition made of the atoms $(P_{\alpha_0}\cap\cdots\cap T^{-(n-1)}P_{\alpha_{n-1}})_{\alpha\in I^n}$. We define the entropy with respect to this refined partition
\begin{equation}\label{nKSentropy}H_{n}(\mu, T, P)=-\sum_{|\alpha|=n}\mu(P_{\alpha_0}\cap\cdots\cap T^{-(n-1)}P_{\alpha_{n-1}})\log\mu(P_{\alpha_0}\cap\cdots\cap T^{-(n-1)}P_{\alpha_{n-1}}).\end{equation}
Using the subadditivity property of entropy, we have for any integers $n$ and $m$,
\begin{equation}\label{classicalsubad}
H_{n+m}(\mu, T, P)\leq H_{n}(\mu, T, P)+H_{m}(T^n\sharp\mu, T, P)=H_{n}(\mu, T, P)+H_{m}(\mu, T, P).
\end{equation}
For the last equality, it is important to underline that we really use the $T$-invariance of the measure $\mu$. A classical argument for subadditive sequences allows us to define the following quantity:
\begin{equation}\label{defentropy}h_{KS}(\mu,T,P):=\lim_{n\rightarrow\infty}\frac{H_n\left(\mu,T,P\right)}{n}.\end{equation}
It is called the Kolmogorov Sinai entropy of $(T,\mu)$ with respect to the partition $P$. The Kolmogorov Sinai entropy $h_{KS}(\mu,T)$ of $(\mu,T)$ is then defined as the supremum of $h_{KS}(\mu,T,P)$ over all partitions $P$ of $X$. Finally, it should be noted that this quantity can be infinite (not in our case thanks to Ruelle inequality~(\ref{ruelle}) for instance). Note also that if, for all index $(\alpha_0,\cdots,\alpha_{n-1})$, $\mu(P_{\alpha_0}\cap\cdots\cap T^{-(n-1)}P_{\alpha_{n-1}})\leq Ce^{-\beta n}$ with $C$ positive constant, then $h_{KS}(\mu,T)\geq \beta$: the metric entropy measures the exponential decrease of the atoms of the refined partition.

\subsection{Quantum entropy}

One can defined a quantum counterpart to the metric entropy. Let $\mathcal{H}$ be an Hilbert space. We call a partition of identity $(\tau_{\alpha})_{\alpha\in I}$ a finite family of operators that satisfies the following relation:
\begin{equation}\label{identity}\sum_{\alpha\in I} \tau_{\alpha}^*\tau_{\alpha}=\text{Id}_{\mathcal{H}}.\end{equation}
Then, one defines the quantum entropy of a normalized vector $\psi$ as
\begin{equation}\label{quantumentropy}h_{\tau}(\psi):=-\sum_{\alpha\in I}\|\tau_{\alpha}\psi\|^2\log\|\tau_{\alpha}\psi\|^2.\end{equation}
Finally, one has the following generalization of a theorem from~\cite{AN2} (the proof immediately generalizes to this case), known as the entropic uncertainty principle~\cite{MU}:
\begin{theo}\label{uncertainty} Let $O_{\beta}$ be a family of bounded operators and $U$ a unitary operator of an Hilbert space $(\mathcal{H},\|.\|)$. Let $\delta'$ be a positive number. Given $(\tau_{\alpha})_{\alpha\in I}$ and $(\pi_{\beta})_{\beta\in K}$ two partitions of identity and $\psi$ a vector in $\mathcal{H}$ of norm $1$ such that
$$\|(Id-O_{\beta})\pi_{\beta}\psi\|\leq\delta'.$$
Suppose both partitions are of cardinal less than $\mathcal{N}$, then
$$h_{\tau}(U\psi)+h_{\pi}(\psi)\geq-2\log \left(c_O(U)+\mathcal{N}\delta'\right),$$
where  $\displaystyle c_O(U)=\max_{\alpha\in I,\beta\in K}\left(\|\tau_{\alpha}U\pi_{\beta}^*O_{\beta}\|\right)$, with $\|\tau_{\alpha}U\pi_{\beta}^*O_{\beta}\|$ the operator norm in $\mathcal{H}$.
\end{theo}

\subsection{Entropy of a special flow}
\label{sectionAbr}

In the previous papers of Anantharaman, Koch and Nonnenmacher (see~\cite{AKN} for example), the main difficulty that was faced to prove main inequality~(\ref{mainineq}) was that the value of $\log J^u(\rho)$ could change a lot depending on the point of the energy layer they looked at. As was mentioned (see section~\ref{heuristic}), we will try to adapt their proof and take into account the changes of the value of $\log J^u(\rho)$. To do this, we will, in a certain way, reparametrize the geodesic flow. Before explaining precisely this strategy, let us recall a classical fact of dynamical system for reparametrization of measure preserving transformations known as the Abramov theorem.\\
First, let us define a special flow (see~\cite{Ab},~\cite{CFS}). Let $(X,\mathcal{B},\mu)$ be a probability space, $T$ an automorphism of $X$ and $f$ a measurable function such that $f(x)>a>0$ for all $x$ in $X$. The function $f$ is called a roof function. We are interested in the set
\begin{equation}\label{suspension}\overline{X}:=\{\left(x,s\right): x\in X,0\leq s<f\left(x\right)\}.\end{equation}
$\overline{X}$ is equipped with the $\sigma$-algebra by restriction of the $\sigma$-algebra on the cartesian product $X\times\mathbb{R}$. For $A$ measurable, one defines $\overline{\mu}(A):=\frac{1}{\int_X fd\mu}\int\int_A d\mu(x)ds$ and $\overline{\mu}(\overline{X})=1$.
\begin{def1}
The special flow under the automorphism $T$, constructed by the function $f$ is the flow $(\overline{T}^t)$ that acts on $\overline{X}$ in the following way, for $t\geq 0$,
\begin{equation}\label{special}\overline{T}^t\left(x,s\right):=\left(T^n x,s+t-\sum_{k=0}^{n-1}f\left(T^k x\right)\right),\end{equation}
where $n$ is the only integer such that $\displaystyle\sum_{k=0}^{n-1}f\left(T^k x\right)\leq s+t<\sum_{k=0}^{n}f\left(T^k x\right)$.\\
For $t<0$, one puts, if $s+t>0$,
$$\overline{T}^t\left(x,s\right):=\left(x,s+t\right),$$
and otherwise,
$$\overline{T}^t\left(x,s\right):=\left(T^{-n} x,s+t+\sum_{k=-n}^{-1}f\left(T^k x\right)\right),$$
where $n$ is the only integer such that $\displaystyle-\sum_{k=-n}^{-1}f\left(T^k x\right)\leq s+t<-\sum_{k=-n+1}^{-1}f\left(T^k x\right)$.
\end{def1}
\begin{rema} A suspension semi-flow can also be defined from an endomorphism.\end{rema}
It can be shown that this special flow preserves the measure $\overline{\mu}$ if $T$ preserves $\mu$ \cite{CFS}. Finally, we can state Abramov theorem for special flows~\cite{Ab}:
\begin{theo}\label{theoabr}
With the previous notations, one has, for all $t\in\mathbb{R}$:
\begin{equation}\label{Abramov}h_{KS}\left(\overline{T}^t,\overline{\mu}\right)=\frac{|t|}{\int_X fd\mu}h_{KS}\left(T,\mu\right).\end{equation}
\end{theo}

\section{Classical setting of the paper}

\label{Anosov}

Before starting the main lines of the proof, we want to describe the classical setting for our surface $M$ and introduce notations that will be useful in the paper. We suppose the geodesic flow over $T^*M$ to have the Anosov property for the first part of the paper. This means that for any $\lambda>0$, the geodesic flow $g^t$ is Anosov on the energy layer $\mathcal{E}(\lambda):=H^{-1}(\lambda)\subset T^*M$ and in particular, the following decomposition holds for all $\rho\in\mathcal{E}(\lambda)$:
$$T_{\rho}\mathcal{E}(\lambda)=E^u(\rho)\oplus E^s(\rho)\oplus \mathbb{R}X_{H}(\rho),$$
where $X_H$ is the Hamiltonian vector field associated to $H$, $E^u$ the unstable space and $E^s$ the stable space~\cite{BS}. It can be denoted that in the setting of this article, they are all one dimensional spaces. The unstable Jacobian $J^u(\rho)$ at the point $\rho$ is defined as the Jacobian of the restriction of $g^{-1}$ to the unstable subspace $E^u(g^1\rho)$:
$$J^u(\rho):=\det\left(dg^{-1}_{|E^u(g^{1}\rho)}\right).$$
For $\theta$ small positive number ($\theta$ will be fixed all along the paper), one defines $$\mathcal{E}^{\theta}:=H^{-1}(]1/2-\theta,1/2+\theta[).$$ As the geodesic flow is Anosov, we can suppose there exist $0<a_0<b_0$ such that
$$\forall\rho\in \mathcal{E}^{\theta},\ a_0\leq -\log J^u(\rho)\leq b_{0}.$$
\begin{rema}In fact, in the general setting of an Anosov flow, we can only suppose that there exists $k_{0}\in\mathbb{N}$ such that $\displaystyle\det\left(dg^{-k_{0}}_{|E^u(g^{k_0}\rho)}\right)<1$ for all $\rho\in\mathcal{E}^{\theta}$. So, to be in the correct setting, we should take $g^{k_0}$ instead of $g$ in the paper. In fact, as $h_{KS}(\mu,g^{k_0})=k_0h_{KS}(\mu,g)$ and $$-\int_{S^*M}\log\det\left(dg^{-k_{0}}_{|E^u(g^{k_0}\rho)}\right)d\mu(\rho)=-k_0\int_{S^*M}\log\det\left(dg^{-1}_{|E^u(g^{1}\rho)}\right)d\mu(\rho),$$
theorem~\ref{maintheo} follows for $k_0=1$ from the case $k_0$ large. However, in order to avoid too many notations, we will suppose $k_0=1$.\end{rema}
We also fix $\epsilon$ and $\eta$ two small positive constants lower than the injectivity radius of the manifold. We choose $\eta$ small enough to have $(2+\frac{b_0}{a_0})b_0\eta\leq\frac{\epsilon}{2}$ (this property will only be used in the proof of lemma~\ref{adapted}). We underline that there exists $\varepsilon>0$ such that if
$$\forall\ (\rho,\rho')\in\mathcal{E}^{\theta}\times\mathcal{E}^{\theta},\ d(\rho,\rho')\leq\varepsilon\Rightarrow |\log J^u(\rho)-\log J^u(\rho')|\leq a_0\epsilon.$$

\subsection*{Discretization of the unstable Jacobian} As was already mentioned, our strategy to prove theorem~\ref{maintheo} will be introduce a discrete reparametrization of the geodesic flow. Regarding this goal, we cut the manifold $M$ and precisely,
we consider a partition $M=\bigsqcup_{i=1}^K O_i$ of diameter smaller than some positive $\delta$. Let $(\Omega_i)_{i=1}^K$ be a finite open cover of $M$ such that for all $1\leq i\leq K$, $O_i\subsetneq\Omega_i$. For $\gamma\in\{1,\cdots,K\}^2$, define an open subset of $T^*M$:
$$U_{\gamma}:=(T^*\Omega_{\gamma_0}\cap g^{-\eta}T^*\Omega_{\gamma_1})\cap\mathcal{E}^{\theta}.$$
We choose the partition $(O_i)_{i=1}^K$ and the open cover $(\Omega_i)_{i=1}^K$ of $M$ such that  $(U_{\gamma})_{\gamma\in\{1,\cdots,K\}^2}$ is a finite open cover of diameter smaller\footnote{In particular, the diameter of the partition $\delta$ depends on $\theta$ and $\epsilon$.} than $\varepsilon$ of $\mathcal{E}^{\theta}$. Then, we define the following quantity, called the discrete Jacobian in time $\eta$:
\begin{equation}\label{jac0}J_{\eta}^u\left(\gamma\right):=\sup\left\{ J^u(\rho):\rho\in U_{\gamma}\right\},\end{equation}
if the previous set is non empty, $e^{-b_0}$ otherwise. Outline that $J^u_{\eta}(\gamma)$ depends on $\eta$ as $U_{\gamma}$ depends on $\eta$. The definition can seem quite asymmetric as we consider the supremum of $J^u(\rho)$ and not of $J^u_{\eta}(\rho)$. However, this choice makes things easier for our analysis.\\
Finally, let $\alpha=(\alpha_0,\alpha_1,\cdots)$ be a sequence (finite or infinite) of elements of $\{1,\cdots,K\}$ whose length is larger than~$1$ and define
\begin{equation}\label{jac}f_+(\alpha):=-\eta\log J_{\eta}^u\left(\alpha_0,\alpha_1\right)\leq\eta b_0\leq\frac{\epsilon}{2},\end{equation}
where the upper bound follows from the previous hypothesis. We underline that, for $\gamma=(\gamma_0,\gamma_1)$, we have
\begin{equation}\label{continuity} \forall\ \rho\in U_{\gamma},\ |f_+(\gamma)+\eta\log J^u(\rho)|\leq a_0\eta\epsilon.\end{equation}
\begin{rema} This last inequality shows that even if our choice for $J^u_{\eta}(\gamma)$ seems quite asymmetric, it allows to have an explicit bound in $\eta$ for quantity~(\ref{continuity}) and it will be quite useful. With a more symmetric choice, we would not have been able to get an explicit bound in $\eta$ for~(\ref{continuity}).
\end{rema}
In the following, we will also have to consider negative times. To do this, we define the analogous functions, for $\beta:=(\cdots,\beta_{-1},\beta_0)$ of finite (or infinite) length,
$$f_-(\beta):=f(\beta_{-1},\beta_0).$$
\begin{rema} Let $\alpha$ and $\beta$ be as previously (finite or infinite). For the sake of simplicity, we will use the notation
$$\beta.\alpha:=(\cdots,\beta_{-1},\beta_0,\alpha_0,\alpha_1,\cdots).$$
The same obviously works for any sequences of the form $(\cdots,\beta_{p-1},\beta_p)$ and $(\alpha_{q},\alpha_{q+1},\cdots)$.
\end{rema}

\section{Outline of the proof}
\label{proof}

Let $(\psi_{\hbar_k})$ be a sequence of orthonormal eigenfunctions of the Laplacian corresponding to the eigenvalues $-1/\hbar_k^{-2}$ such that the corresponding sequence of distributions $\mu_k$ on $T^{*}M$ converges as $k$ tends to infinity to the semiclassical measure $\mu$. For simplicity of notations and to fit semiclassical analysis notations, we will denote $\hbar$ tends to $0$ the fact that $k$ tends to infinity and $\psi_{\hbar}$ and $\hbar^{-2}$ the corresponding eigenvector and eigenvalue. To prove theorem~\ref{maintheo}, we will in particular give a symbolic interpretation of a semiclassical measure and apply the previous results on special flows to this measure.\\
Let $\epsilon'>4\epsilon$ be a positive number, where $\epsilon$ was defined in section~\ref{Anosov}. The link between the two quantities $\epsilon$ and $\epsilon'$ will only be used in section~\ref{bigpdotheo} to define $\nu$. In the following of the paper, the Ehrenfest time $n_E(\hbar)$ will be the quantity
\begin{equation}\label{ehrenfest}n_E(\hbar):=[(1-\epsilon')|\log\hbar|].\end{equation}
We underline that it is an integer time and that, compared with usual definitions of the Ehrenfest time, there is no dependence on the Lyapunov exponent. We also consider a smaller non integer time
\begin{equation}\label{ehrint}T_E(\hbar):=(1-\epsilon)n_E(\hbar).\end{equation}

\subsection{Quantum partitions of identity}

In order to find a lower bound on the metric entropy of the semiclassical measure $\mu$, we would like to apply the entropic uncertainty principle (theorem~\ref{uncertainty}) and see what informations it will give (when $\hbar$ tends to $0$) on the metric entropy of the semiclassical measure $\mu$. To do this, we define quantum partitions of identity corresponding to a given partition of the manifold.

\subsubsection{Partitions of identity}

In section~\ref{Anosov}, we considered a partition of small diameter $(O_i)_{i=1}^K$ of $M$. We also defined $(\Omega_i)_{i=1}^K$ a corresponding finite open cover of small diameter of $M$. By convolution of the characteristic functions $\mathbf{1}_{O_i}$, we obtain $\displaystyle\mathcal{P}=\left(P_i\right)_{i=1,..K}$ a smooth partition of unity on $M$ i.e. for all $x\in M$,
$$\sum_{i=1}^{K}P_i^2(x)=1.$$
We assume that for all $1\leq i\leq K$, $P_i$ is an element of $\mathcal{C}^{\infty}_c(\Omega_i)$. To this classical partition corresponds a quantum partition of identity of $L^2(M)$. In fact, if $P_i$ denotes the multiplication operator by $P_i(x)$ on $L^2(M)$, then one has
\begin{equation}\label{time0}\sum_{i=1}^{K}P_i^{*}P_i=\text{Id}_{L^2(M)}.\end{equation}

\subsubsection{Refinement of the quantum partition under the Schrödinger flow}
\label{refinementbis}
Like in the classical setting of entropy~(\ref{nKSentropy}), we would like to make a refinement of the quantum partition. To do this refinement, we use the Schrödinger propagation operator $U^t=e^{\frac{\imath t\hbar\Delta}{2}}$. We define  $A(t):=U^{-t}AU^t$, where $A$ is an operator on $L^2(M)$. To fit as much as possible with the metric entropy (see definition~(\ref{nKSentropy}) and Egorov property~(\ref{Egorov})), we define the following operators:
\begin{equation}\label{tau}\tau_{\alpha}=P_{\alpha_{k}}(k\eta)\cdots P_{\alpha_{1}}(\eta)P_{\alpha_0}\end{equation}
and
\begin{equation}\label{pi}\pi_{\beta}=P_{\beta_{-k}}(-k\eta)\cdots P_{\beta_{-2}}(-2\eta)P_{\beta_0}P_{\beta_{-1}}(-\eta),\end{equation}
where $\alpha=(\alpha_0,\cdots,\alpha_k)$ and $\beta=(\beta_{-k},\cdots,\beta_0)$ are finite sequences of symbols such that $\alpha_j\in [1,K]$ and $\beta_{-j}\in [1,K]$. We can remark that the definition of $\pi_{\beta}$ is the analogue for negative times of the definition of $\tau_{\alpha}$. The only difference is that we switch the two first terms $\beta_0$ and $\beta_{-1}$. The reason of this choice will appear later in the application of the quantum uncertainty principle (see equality~(\ref{switch}) in section~\ref{uncert}). One can see that for fixed $k$, using the Egorov property~(\ref{Egorov}), \begin{equation}\label{applegorov}\|P_{\alpha_{k}}(k\eta)\cdots P_{\alpha_{1}}(\eta)P_{\alpha_0}\psi_{\hbar}\|^2\rightarrow\mu(P_{\alpha_{k}}^2\circ g^{k\eta}\times\cdots P_{\alpha_{1}}^2\circ g^{\eta}\times P_{\alpha_0}^2)\ \text{as}\ \hbar\ \text{tends}\ \text{to}\ 0.\end{equation} This last quantity is the one used to compute $h_{KS}(\mu,g^{\eta})$ (with the notable difference that the $P_j$ are here smooth functions instead of characteristic functions: see~(\ref{nKSentropy})). As was discussed in the heuristic of the proof~\ref{heuristic}, we will have to understand for which range of times $k\eta$, the Egorov property can be be applied. In particular, we will study for which range of times, the operator $\tau_{\alpha}$ is a pseudodifferential operator of symbol $P_{\alpha_{k}}\circ g^{k\eta}\times\cdots P_{\alpha_{1}}\circ g^{\eta}\times P_{\alpha_0}$ (see~(\ref{applegorov})). In~\cite{AN2} and~\cite{AKN}, they only considered $k\eta\leq|\log\hbar|/\lambda_{\max}$ where $\lambda_{\max}:=\lim_{t\rightarrow\pm\infty}\frac{1}{t}\log\sup_{\rho\in S^*M}|d_{\rho}g^t|$. This choice was not optimal and in the following, we try to define sequences $\alpha$ for which we can say that $\tau_{\alpha}$ is a pseudodifferential operator.

\subsubsection{Index family adapted to the variation of the unstable Jacobian}
\label{roof}
Let $\alpha=(\alpha_0,\alpha_1,\cdots)$ be a sequence (finite or infinite) of elements of $\{1,\cdots,K\}$ whose length is larger than~$1$. We define a natural shift on these sequences
$$\sigma_+((\alpha_0,\alpha_1,\cdots)):=(\alpha_1,\cdots).$$
For negative times and for $\beta:=(\cdots,\beta_{-1},\beta_0)$, we define the backward shift
$$\sigma_-((\cdots,\beta_{-1},\beta_0)):=(\cdots,\beta_{-1}).$$
In the paper, we will mostly use the symbol $x$ for infinite sequences and reserve $\alpha$ and $\beta$ for finite ones. Then, using notations of section~\ref{Anosov} and as described in section~\ref{partition}, index families depending on the value of the unstable Jacobian can be defined as follows:
\begin{equation}\label{Ih}I^{\eta}(\hbar):=I^{\eta}(T_E(\hbar))=\left\{\left(\alpha_0,\cdots,\alpha_{k}\right):k\geq3,\sum_{i=1}^{k-2}f_+\left(\sigma^i_+\alpha\right)\leq T_E(\hbar)<\sum_{i=1}^{k-1}f_+\left(\sigma^i_+\alpha\right)\right\},\end{equation}
\begin{equation}\label{Kh}K^{\eta}(\hbar):=K^{\eta}(T_E(\hbar))=\left\{\left(\beta_{-k},\cdots,\beta_{0}\right): k\geq3,\sum_{i=1}^{k-2}f_-\left(\sigma_-^{i}\beta\right)\leq T_E(\hbar)<\sum_{i=1}^{k-1}f_-\left(\sigma_-^{i}\beta\right)\right\}.\end{equation}
We underline that we will consider any sequence of the previous type and not only sequences for which $U_{\alpha}$ is not empty. These sets define the maximal sequences for which we can expect to have Egorov property for the corresponding $\tau_{\alpha}$. The sums used to define these sets are in a way a discrete analogue of the integral in the inversion formula~(\ref{inverse}) defined in the introduction\footnote{In the higher dimension case mentioned in the footnote of section~\ref{heuristic}, we should take $(d-1)T_{E}(\hbar)$ (where $d$ is the dimension of $M$) instead of $T_{E}(\hbar)$ in the definition of $I^{\eta}(\hbar)$ and $K^{\eta}(\hbar)$.}. The sums used to define the allowed sequences are in fact Riemann sums (with small parameter $\eta$) corresponding to the integral~(\ref{parameter}). We can think of the time $|\alpha|\eta$ as a stopping time for which property~(\ref{applegorov}) will hold (for a symbol $\tau_{\alpha}$ corresponding to $\alpha$).\\
A good way of thinking of these families of words is by introducing the sets
$$\Sigma_+:=\{1,\cdots,K\}^{\mathbb{N}}\ \text{and}\ \Sigma_-:=\{1,\cdots,K\}^{-\mathbb{N}}.$$
We will see that the sets $I^{\eta}(\hbar)$ (resp. $K^{\eta}(\hbar)$) lead to natural partitions of $\Sigma$ (resp. $\Sigma_-$). In the following, it can be helpful to keep in mind picture~\ref{square}. On this figure, we draw the case $K=4$. The biggest square has sides of length $1$. Each square represents an element of $I^{\eta}(\hbar)$ and each square with sides of length $1/2^{k}$ represents a sequence of length $k+1$ (for $k\geq 0$). If we denote $C(\alpha)$ the square that represents $\alpha$, then we can represent the sequences $\alpha.\gamma$ for each $\gamma$ in $\{1,\cdots,4\}$ by subdividing the square $C(\alpha)$ in $4$ squares of same size. Finally, by definition of $I^{\eta}(\hbar)$, we can remark that if $\alpha.\gamma$ is represented in the subdivision (for $\gamma$ in $\{1,\cdots,4\}$), then $\alpha.\gamma'$ is represented in the subdivision for each $\gamma'$ in $\{1,\cdots,4\}$.
\begin{figure}[htb]
\vspace{4cm}
\centerline{
\setlength{\unitlength}{0.5cm}
\begin{picture}(8,0)
\thicklines
\put(0,0){\line(1,0){8}}
\put(0,2){\line(1,0){8}}
\put(0,4){\line(1,0){8}}
\put(0,6){\line(1,0){8}}
\put(0,8){\line(1,0){8}}
\put(0,0){\line(0,1){8}}
\put(2,0){\line(0,1){8}}
\put(4,0){\line(0,1){8}}
\put(6,0){\line(0,1){8}}
\put(8,0){\line(0,1){8}}
\put(1,0){\line(0,1){2}}
\put(3,0){\line(0,1){2}}
\put(3,4){\line(0,1){2}}
\put(5,6){\line(0,1){2}}
\put(7,2){\line(0,1){2}}
\put(0,1){\line(1,0){4}}
\put(6,3){\line(1,0){2}}
\put(2,5){\line(1,0){2}}
\put(4,7){\line(1,0){2}}
\put(3.5,1){\line(0,1){1}}
\put(3.5,4){\line(0,1){1}}
\put(4.5,7){\line(0,1){1}}
\put(3,1.5){\line(1,0){1}}
\put(3,4.5){\line(1,0){1}}
\put(4,7.5){\line(1,0){1}}
\put(0.1,7.6){\tiny{$C(11)$}}
\put(2.1,7.6){\tiny{$C(12)$}}
\put(0.1,3.6){\tiny{$C(31)$}}
\put(6.1,3.6){\tiny{$C(421)$}}
\end{picture}}
\caption{\label{square}\textit{Refinement of variable size}}
\end{figure}
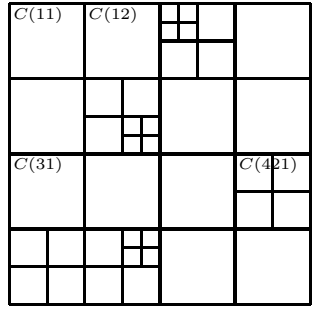
Families of operators can be associated to these families of index: $(\tau_{\alpha})_{\alpha\in I^{\eta}(\hbar)}$ and $(\pi_{\beta})_{\beta\in K^{\eta}(\hbar)}$. One can show that these partitions form quantum partitions of identity (see section~\ref{partition}), i.e.
$$\sum_{\alpha\in I^{\eta}(\hbar)}\tau_{\alpha}^*\tau_{\alpha}=\text{Id}_{L^{2}(M)}\ \text{and}\ \sum_{\beta\in K^{\eta}(\hbar)}\pi_{\beta}^*\pi_{\beta}=\text{Id}_{L^{2}(M)}.$$

\subsection{Symbolic interpretation of semiclassical measures}
\label{symbolic}
Now that we have defined these partitions of variable size, we want to show that they are adapted to compute the entropy of a certain measure with respect to some reparametrized flow associated to the geodesic flow. To do this, we start by giving a symbolic interpretation of the quantum partitions. Recall that we have denoted $\Sigma_+:=\{1,\cdots,K\}^{\mathbb{N}}$. We will also denote $\mathcal{C}_i$ the subset of sequences $(x_n)_{n\in\mathbb{N}}$ such that $x_0=i$. Define also
$$[\alpha_{0},\cdots,\alpha_{k}]:=\mathcal{C}_{\alpha_{0}}\cap\cdots\cap\sigma^{-k}_+\mathcal{C}_{\alpha_{k}},$$
where $\sigma_+$ is the shift $\sigma_+((x_n)_{n\in\mathbb{N}})=(x_{n+1})_{n\in\mathbb{N}}$ (it fits the notations of the previous section). The set $\Sigma_+$ is then endowed with the probability measure (not necessarily $\sigma$-invariant):
$$\mu_{\hbar}^{\Sigma_+}\left(\left[\alpha_{0},\cdots,\alpha_{k}\right]\right)=\mu_{\hbar}^{\Sigma_+}\left(\mathcal{C}_{\alpha_{0}}\cap\cdots\cap\sigma^{-k}_+\mathcal{C}_{\alpha_{k}}\right)=\|P_{\alpha_{k}}(k\eta)\cdots P_{\alpha_{0}}\psi_{\hbar}\|^2.$$
Using the property~(\ref{identity}), it is clear that this definition assures the compatibility conditions to define a probability measure
$$\sum_{\alpha_{k+1}}\mu_{\hbar}^{\Sigma_+}\left(\left[\alpha_{0},\cdots,\alpha_{k+1}\right]\right)=\mu_{\hbar}^{\Sigma_+}\left(\left[\alpha_{0},\cdots,\alpha_{k}\right]\right).$$
Then, we can define the suspension flow, in the sense of Abramov~(section~\ref{sectionAbr}), associated to this probability measure. To do this, the suspension set~(\ref{suspension}) is defined as
\begin{equation}\label{suspset}\overline{\Sigma}_+:=\{\left(x,s\right)\in\Sigma_+\times\mathbb{R}_+:0\leq s<f_+\left(x\right)\}.\end{equation}
Recall that the roof function $f_+$ is defined as $f_+(x):=f_+(x_0,x_1).$ We define a probability measure $\overline{\mu}_{\hbar}^{\overline{\Sigma}_+}$ on $\overline{\Sigma}_+$:
\begin{equation}\label{suspmeas}\overline{\mu}_{\hbar}^{\overline{\Sigma}_+}=\mu_{\hbar}^{\Sigma_+}\times \frac{dt}{\sum_{\alpha\in\{1,\cdots,K\}^2}f_+(\alpha)\|P_{\alpha}\psi_{\hbar}\|^2}=\mu_{\hbar}^{\Sigma_+}\times \frac{dt}{\sum_{\alpha\in\{1,\cdots,K\}^2}f_+(\alpha)\mu_{\hbar}^{\Sigma_+}\left(\left[\alpha\right]\right)}.\end{equation}
The semi-flow~(\ref{special}) associated to $\sigma_+$ is for time $s$:
\begin{equation}\label{suspflow}\overline{\sigma}^s_+\left(x,t\right):=\left(\sigma^{n-1}_+ (x),s+t-\sum_{j=0}^{n-2}f_+\left(\sigma^j_+x\right)\right),\end{equation}
where $n$ is the only integer such that $\displaystyle\sum_{j=0}^{n-2}f_+\left(\sigma^j_+x\right)\leq s+t<\sum_{j=0}^{n-1}f_+\left(\sigma^j_+x\right)$. In the following, we will only consider time $1$ of the flow and its iterates and we will denote $\overline{\sigma}_+:=\overline{\sigma}^1_{+}$.
\begin{rema}It can be underlined that the same procedure holds for the partition $(\pi_{\beta})$. The only differences are that we have to consider $\Sigma_{-}:=\{1,\cdots,K\}^{-\mathbb{N}}$, $\sigma_-((x_n)_{n\leq 0})=(x_{n-1})_{n\leq 0}$ and that the corresponding measure is, for $k\geq 1$,
$$\mu_{\hbar}^{\Sigma_{-}}\left(\left[\beta_{-k},\cdots,\beta_{0}\right]\right)=\mu_{\hbar}^{\Sigma_{-}}\left(\overline{\sigma}_-^{-k}\mathcal{C}_{\beta_{-k}}\cap\cdots\cap\mathcal{C}_{\beta_{0}}\right)=\|P_{\beta_{-k}}(-k\eta)\cdots P_{\beta_{0}}P_{\beta_{-1}}(-\eta)\psi_{\hbar}\|^2.$$
For $k=0$, one should take the only possibility to assure the compatibility condition
$$\mu_{\hbar}^{\Sigma_{-}}\left(\left[\beta_{0}\right]\right)=\sum_{j=1}^K\mu_{\hbar}^{\Sigma_{-}}\left(\left[\beta_{-j},\beta_{0}\right]\right).$$
The definition is quite different from the positive case but in the semiclassical limit, it will not change anything as $P_{\beta_0}$ and $P_{\beta_-1}(-\eta)$ commute. Finally, the `past' suspension set can be defined as
$$\overline{\Sigma}_-:=\{(x,s)\in\Sigma_-\times\mathbb{R}_+:0\leq s<f_-(x)\}.$$
\end{rema}
Now let $\alpha$ be an element of $I^{\eta}(\hbar)$. Define:
\begin{equation}\label{symbpart}\tilde{\mathcal{C}}_{\alpha}:=\mathcal{C}_{\alpha_0}\cap\cdots\cap\sigma^{-k}_+\mathcal{C}_{\alpha_{k}}.\end{equation}
This new family of subsets forms a partition of $\Sigma_+$ (see picture~\ref{square}). Then, a partition $\overline{\mathcal{C}}_{\hbar}^+$ of $\overline{\Sigma}_+$ can be defined starting from the partition $\tilde{\mathcal{C}}$ and $[0,f_+(\alpha)[$. An atom of this suspension partition is an element of the form $\overline{\mathcal{C}}_{\alpha}=\tilde{\mathcal{C}}_{\alpha}\times [0,f_+(\alpha)[$ (see figure $(a)$ of~\ref{f:towers}). For $\overline{\Sigma}^-$ (the suspension set corresponding to $\Sigma_-$), we define an analogous partition $\overline{\mathcal{C}}^{-}_{\hbar}=([\beta]\times[0,f_-(\beta)[)_{\beta\in K^{\eta}(\hbar)}$. Finally, with this interpretation, equality~(\ref{estQentr}) from section~\ref{uncert} (which is just a careful adaptation of the uncertainty principle) can be read as follows:
\begin{equation}\label{Qentrcl}H\left(\overline{\mu}_{\hbar}^{\overline{\Sigma}_+},\overline{\mathcal{C}}_{\hbar}^+\right)+H\left(\overline{\mu}_{\hbar}^{\overline{\Sigma}_-},\overline{\mathcal{C}}^-_{\hbar}\right)\geq\left((1-\epsilon')(1-\epsilon)-c\delta_0\right)|\log\hbar|+C,\end{equation}
where $H$ is defined by~(\ref{defent0}) and $\delta_0$ is some small fixed parameter. To fit as much as possible with the setting of the classical metric entropy, we would like $\overline{\mathcal{C}}_{\hbar}^+$ to be the refinement (under the special flow) of an $\hbar$-independent partition. It is not exactly the case but we can prove the following lemma~(see section~\ref{welladapted} and figure~\ref{f:towers}):
\begin{lemm}\label{adapted} There exists an explicit partition $\overline{\mathcal{C}}_+$ of $\overline{\Sigma}_+$, independent of $\hbar$ such that $\vee_{i=0}^{n_E(\hbar)-1}\overline{\sigma}^{-i}_+\overline{\mathcal{C}}_+$ is a refinement of the partition $\overline{\mathcal{C}}_{\hbar}^+$. Moreover, let $n$ be a fixed positive integer. Then, an atom of the refined partition $\displaystyle\vee_{i=0}^{n-1}\overline{\sigma}^{-i}_+\overline{\mathcal{C}}_+$ is of the form $[\alpha]\times B(\alpha)$, where $\alpha=(\alpha_0,\cdots,\alpha_k)$ is a $k+1$-uple such that $(\alpha_0,\cdots,\alpha_k)$ verifies $n(1-\epsilon)\leq\displaystyle\sum_{j=0}^{k-1}f_+\left(\sigma^j_+\alpha\right)\leq n(1+\epsilon)$ and $B(\alpha)$ is a subinterval of $[0,f_+(\alpha)[$.
\end{lemm}
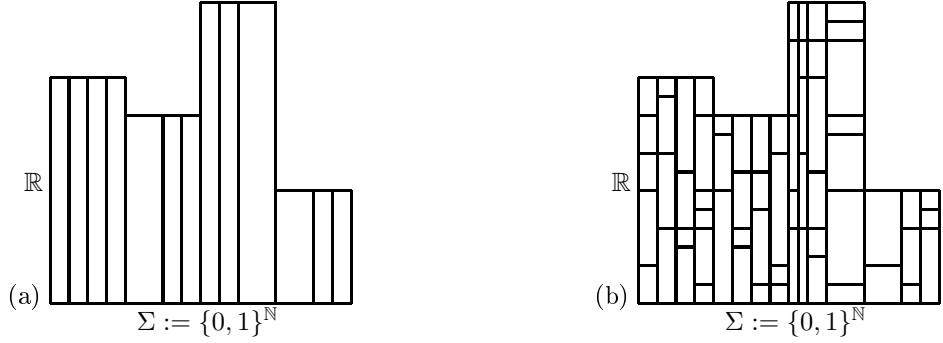
\begin{figure}[htb]
\vspace{4cm}
\centerline{
\subfigure(a){\setlength{\unitlength}{0.5cm}
\begin{picture}(8,0)
\thicklines
\put(0,0){\line(1,0){8}}
\put(0,6){\line(1,0){2}}
\put(2,5){\line(1,0){2}}
\put(4,8){\line(1,0){2}}
\put(6,3){\line(1,0){2}}
\put(0,0){\line(0,1){6}}
\put(0.5,0){\line(0,1){6}}
\put(1,0){\line(0,1){6}}
\put(1.5,0){\line(0,1){6}}
\put(2,0){\line(0,1){6}}
\put(3,0){\line(0,1){5}}
\put(3.5,0){\line(0,1){5}}
\put(4,0){\line(0,1){8}}
\put(4.5,0){\line(0,1){8}}
\put(5,0){\line(0,1){8}}
\put(6,0){\line(0,1){8}}
\put(7,0){\line(0,1){3}}
\put(7.5,0){\line(0,1){3}}
\put(8,0){\line(0,1){3}}
\put(2.3,-0.7){$\Sigma:=\{0,1\}^{\mathbb{N}}$}
\put(-0.7,3){$\mathbb{R}$}
\end{picture}}
\hspace{3cm}
\subfigure(b){\setlength{\unitlength}{0.5cm}
\begin{picture}(8,0)
\thicklines
\put(0,0){\line(1,0){8}}
\put(0,1){\line(1,0){0.5}}
\put(0,3){\line(1,0){0.5}}
\put(0,4){\line(1,0){1}}
\put(0,5){\line(1,0){0.5}}
\put(0.5,2){\line(1,0){1}}
\put(0.5,5.5){\line(1,0){0.5}}
\put(1,2){\line(1,0){1}}
\put(1,1.5){\line(1,0){0.5}}
\put(1,3.5){\line(1,0){0.5}}
\put(1.5,2.5){\line(1,0){0.5}}
\put(1.5,3){\line(1,0){0.5}}
\put(1.5,0.5){\line(1,0){0.5}}
\put(1.5,5){\line(1,0){0.5}}
\put(0,6){\line(1,0){2}}
\put(2,5){\line(1,0){2}}
\put(2,3){\line(1,0){0.5}}
\put(2,4.5){\line(1,0){0.5}}
\put(2.5,3.5){\line(1,0){0.5}}
\put(2.5,1.5){\line(1,0){0.5}}
\put(2.5,2){\line(1,0){0.5}}
\put(3,2.5){\line(1,0){0.5}}
\put(3,0.5){\line(1,0){1}}
\put(3,3.5){\line(1,0){0.5}}
\put(3.5,4){\line(1,0){0.5}}
\put(3.5,1){\line(1,0){0.5}}
\put(4,8){\line(1,0){2}}
\put(4,5){\line(1,0){0.25}}
\put(4,3){\line(1,0){0.25}}
\put(4.25,4){\line(1,0){0.25}}
\put(4,2){\line(1,0){1}}
\put(4,7){\line(1,0){2}}
\put(4.25,6){\line(1,0){0.75}}
\put(4.5,1.25){\line(1,0){0.5}}
\put(4.5,3.5){\line(1,0){0.5}}
\put(5,0.5){\line(1,0){1}}
\put(5,3){\line(1,0){1}}
\put(5,4.5){\line(1,0){1}}
\put(5,7.5){\line(1,0){1}}
\put(5,5){\line(1,0){1}}
\put(6,3){\line(1,0){2}}
\put(6,1){\line(1,0){1}}
\put(7,2){\line(1,0){1}}
\put(7,0.5){\line(1,0){0.5}}
\put(7.5,2.5){\line(1,0){0.5}}
\put(0,0){\line(0,1){6}}
\put(0.5,0){\line(0,1){6}}
\put(1,0){\line(0,1){6}}
\put(1.5,0){\line(0,1){6}}
\put(2,0){\line(0,1){6}}
\put(2.5,0){\line(0,1){5}}
\put(3,0){\line(0,1){5}}
\put(3.5,0){\line(0,1){5}}
\put(4,0){\line(0,1){8}}
\put(4.25,0){\line(0,1){8}}
\put(4.5,0){\line(0,1){8}}
\put(5,0){\line(0,1){8}}
\put(6,0){\line(0,1){8}}
\put(7,0){\line(0,1){3}}
\put(7.5,0){\line(0,1){3}}
\put(8,0){\line(0,1){3}}
\put(2.3,-0.7){$\Sigma:=\{0,1\}^{\mathbb{N}}$}
\put(-0.7,3){$\mathbb{R}$}
\end{picture}}}
\caption{\label{f:towers}\textit{The basis of each tower corresponds to the set of sequences starting with the letters $(\alpha_0,\alpha_1)$, where $\alpha_0$ and $\alpha_1$ are in $\{0,1\}$ and each tower corresponds to the set $\mathcal{C}_{\alpha_0,\alpha_1}\times [0,f_+(\alpha_0,\alpha_1))$. The set $\overline{\Sigma}_+$ admits several partitions. The figure on the left corresponds to the partition $\overline{\mathcal{C}}_{\hbar}^+$ of $\overline{\Sigma}_+$. The figure on the right corresponds to the refinement of the fixed partition $\overline{\mathcal{C}}_+$ under $\overline{\sigma}_+$, i.e. $\displaystyle\vee_{i=0}^{n_{E}(\hbar)-1}\overline{\sigma}^{-i}_+\overline{\mathcal{C}}_+$.
}}
\end{figure}
This lemma is crucial as it allows to interpret an inequality on the quantum entropy as an inequality on classical entropy. In fact, applying basic properties of $H$ between two partitions (see section~\ref{KSentropy} and figure~\ref{f:towers}), one finds that
\begin{equation}\label{refinement}H\left(\overline{\mu}_{\hbar}^{\overline{\Sigma}_+},\overline{\mathcal{C}}_{\hbar}^+\right)\leq H\left(\overline{\mu}_{\hbar}^{\overline{\Sigma}_+},\vee_{i=0}^{n_E(\hbar)-1}\overline{\sigma}^{-i}_+\overline{\mathcal{C}}_+\right)= H_{n_E(\hbar)}\left(\overline{\mu}_{\hbar}^{\overline{\Sigma}_+},\overline{\sigma}_+,\overline{\mathcal{C}}_+\right).\end{equation}
One can obtain the same lemma for the `past' shift and in particular, it gives an $\hbar$-independent partition $\overline{\mathcal{C}}_-$. To conclude this symbolic interpretation of quantum entropy, with natural notations, inequality~(\ref{Qentrcl}) together with~(\ref{refinement}) gives the following proposition
\begin{prop}\label{mainUP} With the previous notations, one has the following inequality:
\begin{equation}\label{inequality}\frac{1}{n_E(\hbar)}\left(H_{n_E(\hbar)}\left(\overline{\mu}_{\hbar}^{\overline{\Sigma}_+},\overline{\sigma}_+,\overline{\mathcal{C}}_+\right)+H_{n_E(\hbar)}\left(\overline{\mu}_{\hbar}^{\overline{\Sigma}_{-}},\overline{\sigma}_-,\overline{\mathcal{C}}_{-}\right)\right)\geq\left(1-\epsilon-c\delta_0\right)+\frac{C}{n_{E}(\hbar)}.\end{equation}
\end{prop}
The quantum entropic uncertainty principle gives an information on the entropy of a special flow. Now, we would like to let $\hbar$ tends to $0$ to find a lower on the metric entropy of a limit measure (that we will precise in section~\ref{subad}) with respect to $\overline{\sigma}_+$. However, both $n_E(\hbar)$ and $\mu_{\hbar}$ depend on $\hbar$ and we have to be careful before passing to the semiclassical limit.

\subsection{Subadditivity of the entropy}

\label{subad}

The Egorov property~(\ref{Egorov}) implies that $\mu_{\hbar}^{\Sigma_+}$ tends to a measure $\mu^{\Sigma_+}$ on $\Sigma_+$ (as $\hbar$ tends to $0$) defined as follows:
\begin{equation}\label{limit}\mu^{\Sigma_+}\left(\left[\alpha_{0},\cdots,\alpha_{k}\right]\right)=\mu\left(P_{\alpha_k}^2\circ g^{k\eta}\times\cdots\times P_{\alpha_0}^2\right),\end{equation}
where $k$ is a fixed integer. Using the property of partition, this defines a probability measure on $\Sigma_+$. To this probability measure corresponds a probability measure $\overline{\mu}^{\overline{\Sigma}_+}$ on the suspension set $\overline{\Sigma}_+$. It is an immediate corollary that $\overline{\mu}^{\overline{\Sigma}_+}$ is the limit of the probability measure $\overline{\mu}_{\hbar}^{\overline{\Sigma}_+}$. Moreover, using Egorov one more time, one can check that the measure $\mu^{\Sigma_+}$ is $\sigma_+$-invariant and using results about special flows~\cite{CFS}, $\overline{\mu}^{\overline{\Sigma}_+}$ is $\overline{\sigma}_+$-invariant. The same works for $\mu_{\hbar}^{\Sigma_{-}}$ and $\overline{\mu}_{\hbar}^{\overline{\Sigma}_{-}}$.
\begin{rema} In the following, we will often prove properties in the case of $\Sigma_+$. The proofs are the same in the case of $\Sigma_{-}$.\end{rema}
As $n_E(\hbar)$ and $\mu_{\hbar}$ depend both on $\hbar$, we cannot let $\hbar$ tend to $0$ if we want to keep an information about the metric entropy. In fact, the left quantity in~(\ref{inequality}) does not tend a priori to the Kolmogorov-Sinai entropy. We want to proceed as in the classical case (see~(\ref{classicalsubad})) and prove a subadditivity property. This will allow to replace $n_E(\hbar)$ by a fixed $n_0$ (see below) in the left hand side of~(\ref{inequality}). This is done with the following theorem that will be proved in section~\ref{commutativity}:
\begin{theo} \label{subadditivity}Let $\overline{\mathcal{C}}$ be the partition of lemma~(\ref{adapted}). There exists a function $R(n_0,\hbar)$ on $\mathbb{N}\times(0,1]$ such that
$$\forall n_0\in\mathbb{N},\ \ \ \ \lim_{\hbar\rightarrow 0}|R(n_0,\hbar)|=0.$$
Moreover, for any $\hbar\in(0,1]$ and any $n_0,m\in\mathbb{N}$ such that $n_0+m\leq n_E(\hbar)$, one has
$$H_{n_0+m}\left(\overline{\mu}_{\hbar}^{\overline{\Sigma}_+},\overline{\sigma}_+,\overline{\mathcal{C}}_+\right)\leq H_{n_0}\left(\overline{\mu}_{\hbar}^{\overline{\Sigma}_+},\overline{\sigma}_+,\overline{\mathcal{C}}_+\right)+H_{m}\left(\overline{\mu}_{\hbar}^{\overline{\Sigma}_+},\overline{\sigma}_+,\overline{\mathcal{C}}_+\right)+R(n_0,\hbar).$$
The same holds for $\Sigma_{-}$.
\end{theo}
This theorem says that the entropy satisfies almost the subadditivity property (see~(\ref{classicalsubad})) for time lower than the Ehrenfest time. It is an analogue of a theorem from~\cite{AN2} (proposition $2.8$) except that we have taken into account the fact that the unstable jacobian varies on the surface and that we can make our semiclassical analysis for larger time than in~\cite{AN2}. The proof of this theorem is the object of section~\ref{commutativity} and~\ref{bigpdotheo} (where semiclassical analysis for 'local Ehrenfest time' is performed). Then, one can apply the standard argument for subadditive sequences. Let $n_0$ be a fixed integer in $\mathbb{N}$ and write the euclidian division $n_E(\hbar)=qn_0+r$ with $r<n_0$. The previous theorem then implies
$$\frac{H_{n_E(\hbar)}\left(\overline{\mu}_{\hbar}^{\overline{\Sigma}_+},\overline{\sigma}_+,\overline{\mathcal{C}}_+\right)}{n_E(\hbar)}\leq\frac{H_{n_0}\left(\overline{\mu}_{\hbar}^{\overline{\Sigma}_+},\overline{\sigma}_+,\overline{\mathcal{C}}_+\right)}{n_0}+\frac{H_{r}\left(\overline{\mu}_{\hbar}^{\overline{\Sigma}_+},\overline{\sigma}_+,\overline{\mathcal{C}}_+\right)}{n_E(\hbar)}+\frac{R(n_0,\hbar)}{n_0}.$$
As $r$ stays uniformly bounded in $n_0$, the inequality~(\ref{inequality}) becomes
\begin{equation}\label{inequality2}
\frac{1}{n_0}\left(H_{n_0}\left(\overline{\mu}_{\hbar}^{\overline{\Sigma}_+},\overline{\sigma}_+,\overline{\mathcal{C}}_+\right)+H_{n_0}\left(\overline{\mu}_{\hbar}^{\overline{\Sigma}_{-}},\overline{\sigma}_{-},\overline{\mathcal{C}}_{-}\right)\right)\geq\left(1-\epsilon-c\delta_0\right)+\frac{C(n_0)}{n_{E}(\hbar)}-2\frac{R(n_0,\hbar)}{n_0}.
\end{equation}

\subsection{Application of the Abramov theorem}
\label{applabramov}
Using inequality~(\ref{inequality2}), we can conclude using Abramov theorem~(\ref{Abramov}). Making $\hbar$ tend to $0$, one finds that (as was mentioned at the beginning of~\ref{subad})
$$\frac{1}{n_0}\left(H_{n_0}\left(\overline{\mu}^{\overline{\Sigma}_+},\overline{\sigma}_+,\overline{\mathcal{C}}_+\right)+H_{n_0}\left(\overline{\mu}^{\overline{\Sigma}_{-}},\overline{\sigma}_-,\overline{\mathcal{C}}_{-}\right)\right)\geq\left(1-\epsilon-c\delta_0\right).$$
The Abramov theorem holds for automorphisms so one can look at the natural extension of $(\Sigma_+,\sigma_+)$ and $(\Sigma_-,\sigma_-)$. To do this, we introduce $\Sigma'=\{1,\cdots,K\}^{\mathbb{Z}}$ and $\sigma'((x_n)_{n\in\mathbb{Z}}):=(x_{n+1})_{n\in\mathbb{Z}}$. With these notations, the natural extension of $(\Sigma_+,\sigma_+)$ is $(\Sigma',\sigma')$ and the one of $(\Sigma_-,\sigma_-)$ is $(\Sigma',\sigma'^{-1})$. We define then two associated suspension sets
$$\overline{\Sigma}'_+:=\{(x,s)\in\Sigma\times\mathbb{R}_+:0\leq s<f(x_0,x_1)\}\ \text{and}\ \overline{\Sigma}'_-:=\{(x,s)\in\Sigma\times\mathbb{R}_+:0\leq s<f(x_{-1},x_0)\}.$$ 
We also denote $\overline{\sigma}'_+$ (resp. $\overline{\sigma}'_-$) the suspension flow on $\overline{\Sigma}'_+$ (resp. $\overline{\Sigma}'_-$) associated to the automorphism $\sigma'$ (resp. $\sigma'^{-1}$). Finally, we underline that $\overline{\mathcal{C}}_+$ (resp. $\overline{\mathcal{C}}_-$) can be viewed as partitions of the set $\overline{\Sigma}'_+$ (resp. $\overline{\Sigma}'_-$). This discussion allows us to derive that
\begin{equation}\label{lowerboundbis}\frac{1}{n_0}\left(H_{n_0}\left(\overline{\mu}^{\overline{\Sigma}'_+},\overline{\sigma}'_+,\overline{\mathcal{C}}_+\right)+H_{n_0}\left(\overline{\mu}^{\overline{\Sigma}'_-},\overline{\sigma}'_-,\overline{\mathcal{C}}_{-}\right)\right)\geq\left(1-\epsilon-c\delta_0\right).\end{equation}
In view of section~\ref{partition}, we have an exact expression for $\mathcal{C}$ in terms of the functions $(P_i)_i$ (see proof of lemma~\ref{adapted}). The measure $\overline{\mu}^{\overline{\Sigma}'_+}$ (resp. $\overline{\mu}^{\overline{\Sigma}'_-}$) is $\overline{\sigma}'_+$-invariant (resp. $\overline{\sigma}'_-$-invariant) as $\mu^{\Sigma}$ is $\sigma$-invariant (resp. $\sigma^{-1}$-invariant)~\cite{CFS}. In the previous inequality, there is still one notable difference with the metric entropy: we consider smooth partitions of identity $(P_i)_i$ (as it was necessary to make the semiclassical analysis). To return to the classical case, the procedure of~\cite{AN2} can be adapted using the exact form of the partition $\overline{\mathcal{C}}$ (see lemma~\ref{adapted}). Recall that each $P_i$ is an element of $\mathcal{C}^{\infty}_c(\Omega_i)$ and that we considered a partition $M=\bigsqcup_i O_i$ of small diameter $\delta$, where each $O_i\subsetneq\Omega_i$ (see section~\ref{Anosov}). One can slightly move the boundaries of the $O_i$ such that they are not charged by $\mu$ (see~appendix of~\cite{An}). By convolution of the $\mathbf{1}_{O_i}$, we obtained the smooth partition $(P_i)_i$ of identity of diameter smaller than $2\delta$. The previous inequality does not depend on the derivatives of the $P_i$. Regarding also the form of the partition $\overline{\mathcal{C}}$ (see lemma~\ref{adapted}), we can replace the smooth functions $P_i$ by the characteristic functions $\mathbf{1}_{O_i}$ in inequality~(\ref{lowerboundbis}). One can let $n_0$ tend to infinity and find
$$ h_{KS}\left(\overline{\mu}^{\overline{\Sigma}'_+},\overline{\sigma}'_+\right)+h_{KS}\left(\overline{\mu}^{\overline{\Sigma}'_-},\overline{\sigma}'_-\right) \geq h_{KS}\left(\overline{\mu}^{\overline{\Sigma}'_+},\overline{\sigma}'_+,\overline{\mathcal{C}}_+\right)+h_{KS}\left(\overline{\mu}^{\overline{\Sigma}'_-},\overline{\sigma}'_-,\overline{\mathcal{C}}_{-}\right)\geq\left(1-\epsilon-c\delta_0\right).$$
Then, using Abramov theorem~(\ref{Abramov}), the previous inequality implies that
$$h_{KS}(\mu,g^{\eta})+h_{KS}(\mu,g^{-\eta})\geq h_{KS}\left(\overline{\mu}^{\overline{\Sigma}'_+},\overline{\sigma}'_+\right)+h_{KS}\left(\overline{\mu}^{\overline{\Sigma}'_-},\overline{\sigma}'_-\right)\geq\left(1-\epsilon-c\delta_0\right)\sum_{\gamma\in\{1,\cdots,K\}^2}f\left(\gamma\right)\mu^{\Sigma'}\left([\gamma]\right).$$
After division by $\eta$ and letting the diameter of the partition $\delta$ tends to $0$, then $\epsilon$ tends to $0$ and finally $\delta_0$ to $0$, one gets
$$h_{KS}(\mu,g)\geq\frac{1}{2}\left|\int_{S^* M}\log J^u(\rho) d\mu(\rho)\right|.\square$$

\subsection*{Notations} In the following, we have to prove the various results for both $\Sigma_+$ and $\Sigma_-$. We will always treat the case of $\Sigma_+$ and the case of $\Sigma_-$ can always be deduced using the same methods. For the sake of simplicity, we will forget the notation $+$ for $(\Sigma_+,\sigma_+,f_+)$ when there will be no ambiguity and we will use the notation $(\Sigma,\sigma,f)$.

\section{Partitions of variable size}
\label{partition}

In this section, we define precisely the index families $I^{\eta}$ and $K^{\eta}$ depending on the unstable jacobian used in section~\ref{proof}. These families are used to construct quantum partitions of identity and partitions adapted to the special flow~(see section~\ref{welladapted}). In the last section, we apply the uncertainty principle to eigenfunctions of the Laplacian for these quantum partitions of variable size.

\subsection{Stopping time}
\label{stoppingtime}

Let $t$ be a real positive number that will be greater than $2b_0\eta$. Define index families as follows (see section~\ref{roof} for definitions of $f_+$, $\sigma_+$, $f_-$ and $\sigma_-$):
$$I^{\eta}(t):=\left\{\alpha=\left(\alpha_0,\cdots,\alpha_{k}\right):k\geq3,\sum_{i=1}^{k-2}f_+\left(\sigma^i_+\alpha\right)\leq t<\sum_{i=1}^{k-1}f_+\left(\sigma^i_+\alpha\right)\right\},$$
$$K^{\eta}(t):=\left\{\beta=\left(\beta_{-k},\cdots,\beta_{0}\right): k\geq3,\sum_{i=1}^{k-2} f_-\left(\sigma_-^{i}\beta\right)\leq t<\sum_{i=1}^{k-1}f_-\left(\sigma_-^{i}\beta\right)\right\}.$$
Let $x$ be an element of $\{1,\cdots,K\}^{\mathbb{N}}$. We denote $k_{t}(x)$ the unique integer $k$ such that
$$\sum_{i=1}^{k-2}f_+\left(\sigma^i_+x\right)\leq t<\sum_{i=1}^{k-1}f_+\left(\sigma^i_+x\right).$$
In the probability language, $k_t$ is a stopping time in the sense that the property $\{k_t(x)\leq k\}$ depends only on the $k+1$ first letters of $x$. For a finite word $\alpha=(\alpha_0,\cdots,\alpha_{k})$, we say that $k=k_t(\alpha)$ if $\alpha$ satisfies the previous inequality. With these notations, $I^{\eta}(t):=\{\alpha:|\alpha|=k_t(\alpha)+1\}.$ The same holds for $K^{\eta}(t)$.
\begin{rema} This stopping time $k_t(\alpha)$ for $t\sim \frac{n_E(\hbar)}{2}$ will be the time for which we will later try to make the Egorov property work. Precisely, we will prove an Egorov property for some symbols corresponding to the sequence $\alpha$ (see~(\ref{mainegorov}) for example).\end{rema}
\begin{rema} We underline that our choice of defining the sets $I^{\eta}$ and $K^{\eta}$ with sums starting at $i=1$ (and not $0$) will simplify our construction in paragraph~\ref{ss:refined-part}.
\end{rema}

\subsection{Partitions associated}
\label{welladapted}
\subsubsection{Partitions of identity}

Let $\alpha=(\alpha_0,\cdots,\alpha_{k})$ be a finite sequence. Recall that we denoted $\tau_{\alpha}:=P_{\alpha_{k}}(k\eta)\cdots P_{\alpha_0}$, where $A(s):=U^{-s}AU^s$. In~\cite{AN2} and~\cite{AKN}, they used quantum partitions of identity by considering $(\tau_{\alpha})_{|\alpha|=k}$. In our paper, we consider a slightly different partition that is more adapted to the variations of the unstable jacobian:
\begin{lemm} Let $t$ be in $[2b_0\eta,+\infty[$. The family $(\tau_{\alpha})_{\alpha\in I^{\eta}(t)}$ is a partition of identity: $$\sum_{\alpha\in I^{\eta}(t)}\tau_{\alpha}^*\tau_{\alpha}=\text{Id}_{L^{2}(M)}.$$
\end{lemm}
\begin{proof} We define, for each $1\leq l\leq N$ (where $N+1$ is the size of the longest word of $I^{\eta}(t)$),
$$I^{\eta}_l(t):=\left\{\alpha=(\alpha_0,\cdots,\alpha_l):\exists \gamma=(\gamma_{l+1},\cdots,\gamma_k),\ N\geq k>l\ \text{s.t.}\ \alpha.\gamma\in I^{\eta}(t)\right\}.$$
We recall that we defined $\alpha.\gamma:=(\alpha_0,\cdots,\alpha_l,\gamma_{l+1},\cdots,\gamma_k))$. For $l=N$, this set is empty. We want to to show that for each $2\leq l\leq N$, we have:
\begin{equation}\label{interstep}\sum_{\alpha\in I^{\eta}(t),|\alpha|=l+1}\tau_{\alpha}^*\tau_{\alpha}+\sum_{\alpha\in I^{\eta}_l(t)}\tau_{\alpha}^*\tau_{\alpha}=\sum_{\alpha\in I_{l-1}^{\eta}(t)}\tau_{\alpha}^*\tau_{\alpha}.\end{equation}
To prove this equality we use the fact that $\sum_{\gamma=1}^KP_{\gamma}(l)^*P_{\gamma}(l)=\text{Id}_{L^2(M)}$ to write:
\begin{equation}\label{e:step-induc-figure}\sum_{\alpha\in I_{l-1}^{\eta}(t)}\tau_{\alpha}^*\tau_{\alpha}=\sum_{\gamma=1}^K\sum_{\alpha\in I_{l-1}^{\eta}(t)}\tau_{\alpha.\gamma}^*\tau_{\alpha.\gamma}.\end{equation}
We split then this sum in two parts to find equality~(\ref{interstep}). To conclude the proof, we write
$$\sum_{\alpha\in I^{\eta}(t)}\tau_{\alpha}^*\tau_{\alpha}=\sum_{k=2}^N\sum_{\alpha\in I^{\eta}(t),|\alpha|=k+1}\tau_{\alpha}^*\tau_{\alpha}$$
As $t>2b_0\eta\geq\max_{\gamma}f(\gamma)$, the set $I^{\eta}_1(t)$ is equal to $\{1,\cdots,K\}^2$. By induction from $N$ to $1$ using equality~(\ref{interstep}) at each step, we find then:
$$\sum_{\alpha\in I^{\eta}(t)}\tau_{\alpha}^*\tau_{\alpha}=\text{Id}_{L^{2}(M)}.$$\end{proof}
\begin{rema} A step of the induction can be easily understood by looking at figure~\ref{recurrence} where each square represents an index over which the sum is made (as it was explained for figure~\ref{square}). In fact, at each step of the induction $l$, we consider the smallest squares (which correspond to the longest words of length $l+1$) and use the property of partition of identity to reduce them to a larger square of size $2^{-l}$ (i.e. a word of smaller length $l$). Doing this exactly corresponds to step~(\ref{e:step-induc-figure}) of the induction.\end{rema}
Following the same procedure, we denote $\pi_{\beta}=P_{\beta_{-k}}(-k\eta)\cdots P_{\beta_0}P_{\beta_{-1}}(-\eta)$ for $\beta$ in $K^{\eta}(t)$. These operators follow the relation: $\displaystyle \sum_{\beta\in K^{\eta}(t)}\pi_{\beta}^*\pi_{\beta}=\text{Id}_{L^2(M)}$. As was mentioned in section~\ref{refinementbis}, because of a technical reason that will appear in the application of the entropic uncertainty principle (see~(\ref{switch})), the two definitions are slightly different.\\
\begin{figure}[htb]
\vspace{4cm}
\centerline{
\subfigure(a){\setlength{\unitlength}{0.5cm}
\begin{picture}(8,0)
\thicklines
\put(0,0){\line(1,0){8}}
\put(0,2){\line(1,0){8}}
\put(0,4){\line(1,0){8}}
\put(0,6){\line(1,0){8}}
\put(0,8){\line(1,0){8}}
\put(0,0){\line(0,1){8}}
\put(2,0){\line(0,1){8}}
\put(4,0){\line(0,1){8}}
\put(6,0){\line(0,1){8}}
\put(8,0){\line(0,1){8}}
\put(1,0){\line(0,1){2}}
\put(3,0){\line(0,1){2}}
\put(3,4){\line(0,1){2}}
\put(5,6){\line(0,1){2}}
\put(7,2){\line(0,1){2}}
\put(0,1){\line(1,0){4}}
\put(6,3){\line(1,0){2}}
\put(2,5){\line(1,0){2}}
\put(4,7){\line(1,0){2}}
\put(3.5,1){\line(0,1){1}}
\put(3.5,4){\line(0,1){1}}
\put(4.5,7){\line(0,1){1}}
\put(3,1.5){\line(1,0){1}}
\put(3,4.5){\line(1,0){1}}
\put(4,7.5){\line(1,0){1}}
\end{picture}}
\hspace{3cm}
\subfigure(b){\setlength{\unitlength}{0.5cm}
\begin{picture}(8,0)
\thicklines
\put(0,0){\line(1,0){8}}
\put(0,2){\line(1,0){8}}
\put(0,4){\line(1,0){8}}
\put(0,6){\line(1,0){8}}
\put(0,8){\line(1,0){8}}
\put(0,0){\line(0,1){8}}
\put(2,0){\line(0,1){8}}
\put(4,0){\line(0,1){8}}
\put(6,0){\line(0,1){8}}
\put(8,0){\line(0,1){8}}
\put(1,0){\line(0,1){2}}
\put(3,0){\line(0,1){2}}
\put(3,4){\line(0,1){2}}
\put(5,6){\line(0,1){2}}
\put(7,2){\line(0,1){2}}
\put(0,1){\line(1,0){4}}
\put(6,3){\line(1,0){2}}
\put(2,5){\line(1,0){2}}
\put(4,7){\line(1,0){2}}
\end{picture}}}
\caption{\label{recurrence}\textit{A step of the induction}}
\end{figure}
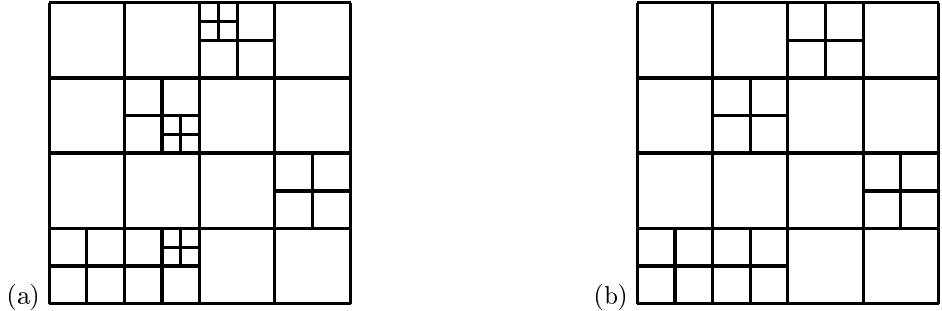

\subsubsection{Partitions of $\{1,\cdots,K\}^{\mathbb{N}}$ associated to $I^{\eta}(1)$}\label{ss:refined-part}

In this section, we would like to consider some partitions of $\Sigma:=\{1,\cdots,K\}^{\mathbb{N}}$ and of $\overline{\Sigma}$~(see~(\ref{suspset})) associated to the family $I^{\eta}(1)$. Precisely, we will construct an explicit partition $\overline{\mathcal{C}}$ of $\overline{\Sigma}$ such that its refinement at time $n$ under $\Sigma$ is linked with the partitions $([\alpha]\times[0,f(\alpha)[)_{\alpha\in I^{\eta}(n)}$ (see lemma~\ref{adapted}).\\
In this paragraph, we give an explicit expression for $\overline{\mathcal{C}}$ and in the next one, we prove lemma~\ref{adapted} that gives a link between the partition $\vee_{i=0}^{n-1}\overline{\sigma}^{-i}\overline{\mathcal{C}}$ and  $([\alpha]\times[0,f(\alpha)[)_{\alpha\in I^{\eta}(n)}$. Recall that
$$I^{\eta}(1):=\left\{\alpha=\left(\alpha_0,\cdots,\alpha_{k}\right):k\geq3,\sum_{i=1}^{k-2}f\left(\sigma^i\alpha\right)\leq 1<\sum_{i=1}^{k-1}f\left(\sigma^i\alpha\right)\right\}.$$
For $\alpha\in I^{\eta}(1)$, it can be easily remarked that $\displaystyle  \sum_{j=0}^{k-1}f\left(\sigma^j\alpha\right)> 1$. It means that there exists a unique integer $k'\leq k$ such that
$$\sum_{j=0}^{k'-2}f\left(\sigma^j\alpha\right)\leq 1<\sum_{j=0}^{k'-1}f\left(\sigma^j\alpha\right).$$
In the following, $k$ and $k'$ will be often denoted $k(\alpha)=k_1(\alpha)$ and $k'(\alpha)$ to remember the dependence in $\alpha$. The following lemma can be easily shown:
\begin{lemm}\label{tinylemma}
Let $\alpha\in I^{\eta}(1)$. One has $|k(\alpha)-k'(\alpha)|\leq\frac{b_0}{a_0}+1$.
\end{lemm}
\begin{proof} Suppose $k'+1<k$ (otherwise it is trivial). Write:
$$\sum_{j=1}^{k-2}f\left(\sigma^j\alpha\right)-\sum_{j=0}^{k'-1}f\left(\sigma^j\alpha\right)\leq 1 -1\ \text{implies}\ \sum_{j=k'}^{k-2}f\left(\sigma^j\alpha\right)\leq f\left(\alpha\right).$$
And finally, one finds $(k-2-k'+1)a_0\eta\leq b_0\eta.$\end{proof}
Let $\alpha$ be an element of $I^{\eta}(1)$. We make a partition of the interval $[0,f(\alpha)[$ under a form that will be useful (as it is adapted to the dynamics of the special flow). Motivated by the definition of a special flow, let us divide it as follows for $k=k(\alpha)$ and $k'=k'(\alpha)$:
$$I_{k'-2}(\alpha)=[0,\sum_{j=0}^{k'-1}f\left(\sigma^j\alpha\right)-1[,\cdots I_{p-2}(\alpha)=[\sum_{j=0}^{p-2}f\left(\sigma^j\alpha\right)-1,\sum_{j=0}^{p-1}f\left(\sigma^j\alpha\right)-1[,\cdots$$
$$I_{k-2}(\alpha)=[\sum_{j=0}^{k-2}f\left(\sigma^j\alpha\right)-1,f\left(\alpha\right)[,$$
where $k'(\alpha)\leq p\leq k(\alpha)$. If $k(\alpha)=k'(\alpha)$, one puts $I_{k'-2}(\alpha)=I_{k-2}(\alpha)=[0,f(\alpha)[$.\\
A partition $\tilde{\mathcal{C}}$ of $\Sigma$ can be defined. It is composed of the following atoms:
$$\tilde{\mathcal{C}}_{\gamma}:=\mathcal{C}_{\gamma_0}\cap\cdots\cap\sigma^{-k}\mathcal{C}_{\gamma_{k}},$$
where $\gamma$ be an element of $I^{\eta}(1)$. A partition $\overline{\mathcal{C}}$ of $\overline{\Sigma}$ can be constructed starting from the partition $\tilde{\mathcal{C}}$ and the partition of $[0,f(\gamma)[$. An atom of this partition $\overline{\mathcal{C}}$ is defined as
$$\overline{\mathcal{C}}:=\left\{\overline{\mathcal{C}}_{\gamma,p}=\tilde{\mathcal{C}}_{\gamma}\times I_{p-2}(\gamma):\gamma\in I^{\eta}(1),\ \text{and}\ k'(\gamma)\leq p\leq k(\gamma)\right\}.$$
We will verify in next paragraph that this partition satisfies the properties of lemma~\ref{adapted}. The choice of these specific intervals can seem quite artificial but it allows to know the exact action of $\overline{\sigma}$ on each atom of the partition
$$\forall (x,t)\in\overline{\mathcal{C}}_{\gamma,p},\ \overline{\sigma}(x,t)=(\sigma^{p-1}(x),1+t-\sum_{j=0}^{p-2}f(\sigma^jx)).$$
If we had only considered the partition made of the atoms $\tilde{C}_{\gamma}\times[0,f(\gamma)[$, we would not have a precise definition for $\overline{\sigma}(x,t)$.

\subsubsection{Proof of the crucial lemma~\ref{adapted}}

In this paragraph, lemma~\ref{adapted} is shown and proves in particular that the previous partition $\overline{\mathcal{C}}$ is well adapted to the special flow on $\overline{\Sigma}$. Let $(\gamma_i,p_i)_{0\leq i\leq n-1}$ be a family of couples such that $\gamma_i\in I^{\eta}(1)$ and $k'(\gamma_i)\leq p_i\leq k(\gamma_i)$. Suppose the considered atom is a non empty atom of $\displaystyle\vee_{i=0}^{n-1}\overline{\sigma}^{-i}\overline{\mathcal{C}}$ (otherwise the result is trivial by taking $B(\alpha)$ empty).\\
\\
We begin by proving the second part of lemma~\ref{adapted}. Let $(x,t)$ be an element of
$\overline{\mathcal{C}}_{\gamma_0,p_0}\cap\cdots\cap\overline{\sigma}^{-(n-1)}\overline{\mathcal{C}}_{\gamma_{n-1},p_{n-1}}$. We denote $k_j=k(\gamma_j).$ The sequence $x$ is of the form $(\gamma_0^0,\cdots,\gamma_0^{k_0},x')$ and $t$ belongs to $I_{p_0-2}(\gamma_0)$.
We recall that for $(x,t)\in\overline{\mathcal{C}}_{\gamma_0,p_0}$:
$$\overline{\sigma}(x,t)=\left(\sigma^{p_0-1} (x),1+t-\sum_{j=0}^{p_0-2}f\left(\sigma^jx\right)\right).$$ Necessarily, one has $\gamma_1=(\gamma_{0}^{p_0-1},\cdots,\gamma_0^{k_0},\gamma_1^{k_0-p_0+2},\cdots,\gamma_1^{k_1})$. Proceeding by induction, one finds that $x=(\gamma_0^0,\cdots,\gamma_0^{k_0},\gamma_1^{k_0-p_0+2},\cdots,\gamma_{n-1}^{k_{n-1}},x")$. Define then $\alpha=(\gamma_0^0,\cdots,\gamma_0^{k_0},\gamma_1^{k_0-p_0+2},\cdots,\gamma_{n-1}^{k_{n-1}})$ and
$$B(\gamma):=\left\{t\in[0,f(\gamma_0)[:\exists x\ st\ (x,t)\in \overline{\mathcal{C}}_{\gamma_0,p_0}\cap\cdots\cap\overline{\sigma}^{-(n-1)}\overline{\mathcal{C}}_{\gamma_{n-1},p_{n-1}}\right\}.$$
The first inclusion $\overline{\mathcal{C}}_{\gamma_0,p_0}\cap\cdots\cap\overline{\sigma}^{-(n-1)}\overline{\mathcal{C}}_{\gamma_{n-1},p_{n-1}}\subset\tilde{\mathcal{C}}_{\alpha}\times B(\gamma)$ is clear.\\
Now we will prove the converse inclusion. Consider $(x,t)$ an element of $\overline{\mathcal{C}}_{\gamma_0,p_0}\cap\cdots\overline{\sigma}^{-(n-1)}\overline{\mathcal{C}}_{\gamma_{n-1},p_{n-1}}$. The only thing to prove is that $(X,t)=((\gamma_0^0,\cdots,\gamma_0^{k_0},\gamma_1^{k_0-p_0+2},\cdots,\gamma_{n-1}^{k_{n-1}},x'),t)$ is still an element of $\overline{\mathcal{C}}_{\gamma_0,p_0}\cap\cdots\overline{\sigma}^{-(n-1)}\overline{\mathcal{C}}_{\gamma_{n-1},p_{n-1}}$, for every $x'$ in $\{1,\cdots,K\}^{\mathbb{N}}$. We proceed by induction and suppose $(X,t)$ belongs to $\overline{\mathcal{C}}_{\gamma_0,p_0}\cap\cdots\overline{\sigma}^{-(j-1)}\overline{\mathcal{C}}_{\gamma_{j-1},p_{j-1}}$ for some $j<n$. We have to verify that $\overline{\sigma}^j(X,t)$ belongs to $\overline{\mathcal{C}}_{\gamma_{j},p_{j}}$. As $(X,t)$ belongs to $\overline{\mathcal{C}}_{\gamma_0,p_0}\cap\cdots\overline{\sigma}^{-(j-1)}\overline{\mathcal{C}}_{\gamma_{j-1},p_{j-1}}$, we have
$$\overline{\sigma}^{j}(X,t)=\left(\sigma^{\tiny{p_0+\cdots+p_{j-1}-j}}(X),j+t-\sum_{i=0}^{\tiny{p_0+\cdots+p_{j-1}-j-1}}f(\sigma^iX)\right).$$
It has already been mentioned that for all $i$, $(\gamma_i^0,\cdots,\gamma_i^{k_i-p_i+1})=(\gamma_{i-1}^{p_{i-1}-1},\cdots,\gamma_{i-1}^{k_i})$ (as the considered atom is not empty). It follows that $\sigma^{p_0+\cdots+p_{j-1}-j}(X)$ belongs to $\tilde{\mathcal{C}}_{\gamma_j}$. We know that $\overline{\sigma}^j(x,t)$ is an element of $\overline{\mathcal{C}}_{\gamma_j,p_j}$ and as a consequence,
$$j+t-\sum_{i=0}^{\tiny{p_0+\cdots+p_{j-1}-j-1}}f(\sigma^iX)=j+t-\sum_{i=0}^{\tiny{p_0+\cdots+p_{j-1}-j-1}}f(\sigma^ix)\in I_{p_j-2}(\gamma_j).$$
By induction, we find that $\overline{\mathcal{C}}_{\gamma_0,p_0}\cap\cdots\cap\overline{\sigma}^{-(n-1)}\overline{\mathcal{C}}_{\gamma_{n-1},p_{n-1}}=\tilde{\mathcal{C}}_{\alpha}\times B(\gamma)$. For each $0\leq j\leq n-1$, $t$ belongs to $B(\gamma)$ implies that:
$$t\in I_{p_j-2}(\gamma_j)-j+\sum_{i=0}^{\tiny{p_0+\cdots+p_{j-1}-j-1}}f(\sigma^i\alpha).$$
The set $B(\gamma)$ is then defined as the intersection of $n$ subintervals of $[0,f(\gamma_0)[$ and is in fact a subinterval of $[0,f(\gamma_0)[$.\\
It remains now to prove upper and lower bounds on $\displaystyle\sum_{j=0}^{k-1}f\left(\sigma^j\alpha\right)$. Recall that:
$$\alpha=(\gamma_0^0,\cdots,\gamma_0^{k_0},\gamma_1^{k_0-p_0+2},\cdots,\gamma_{1}^{k_{1}},\cdots,\gamma_{n-1}^{k_{n-1}}).$$ As $0\leq f(\gamma)\leq\frac{\epsilon}{2}$ for all $\gamma$ (finite or infinite subsequence: see inequality~(\ref{jac})), we have then
$$\sum_{j=0}^{k-1}f\left(\sigma^j\alpha\right)\leq \sum_{l=0}^{n-2}\sum_{j=0}^{k_l-2}f\left(\sigma^j\gamma_l\right)+\sum_{j=0}^{k_{n-1}-1}f\left(\sigma^j\gamma_{n-1}\right)\leq n(1+\epsilon).$$
For the lower bound, the same kind of procedure works with a little more care. For $\gamma_0$,
$$\sum_{j=1}^{k_0-1}f(\sigma^j\alpha)=\sum_{j=1}^{k_0-1}f(\sigma^j\gamma_0)>1>1-\epsilon.$$
and for $1\leq l\leq n-1$, one has, using lemma~\ref{tinylemma},
$$\sum_{j=k_{l-1}-p_{l-1}+1}^{k_l-1}f(\sigma^j\gamma_l)>1-(k_{l-1}-p_{l-1}+1)b_0\eta>1-(2+\frac{b_0}{a_0})b_0\eta>1-\epsilon,$$
where the relations between $\epsilon$, $\eta$, $a_0$ and $b_0$ are defined in section~\ref{Anosov}. A lower bound on $\displaystyle\sum_{j=1}^{k-1}f(\sigma^j\alpha)$ is $n(1-\epsilon)$. This achieved the proof of the second part of lemma~\ref{adapted}.\\
Recall that we have defined
$$I^{\eta}(n(1-\epsilon)):=\left\{\left(\alpha_0',\cdots,\alpha_{k}'\right):k\geq2, \sum_{j=1}^{k-2}f\left(\sigma^j\alpha'\right)\leq n(1-\epsilon)<\sum_{j=1}^{k-1}f\left(\sigma^j\alpha'\right)\right\}.$$
So we have also proved that there exists $\alpha'$ in $I^{\eta}(n(1-\epsilon))$ such that
$$\overline{\mathcal{C}}_{\gamma_0,p_0}\cap\cdots\cap\overline{\sigma}^{-(n-1)}\overline{\mathcal{C}}_{\gamma_{n-1},p_{n-1}}\subset\tilde{\mathcal{C}}_{\alpha'}\times[0,f(\gamma_0)[.$$
In other words, $\displaystyle\vee_{i=0}^{n-1}\overline{\sigma}^{-i}\overline{\mathcal{C}}$ is a refinement of the partition
$\displaystyle\left(\tilde{\mathcal{C}}_{\alpha'}\times[0,f(\alpha')[\right)_{\alpha'\in I^{\eta}(n(1-\epsilon))}$ for any integer $n.$ It is slightly stronger than the first part of lemma~\ref{adapted} and it concludes the proof of lemma~\ref{adapted}.$\square$
\begin{rema} As a final comment on this section, we underline again that all the proofs have been written in the case of $\{1,\cdots,K\}^{\mathbb{N}}$ but can be adapted to the case of $\{1,\cdots, K\}^{-\mathbb{N}}$.\end{rema}

\subsection{Uncertainty principle for eigenfunctions of the Laplacian}
\label{uncert}

In the previous section~\ref{welladapted}, we have seen that the partitions of variable size are well adapted to the reparametrized flow (used in the Abramov theorem). Moreover, we have given a proof of lemma~\ref{adapted} that gives a link between the different partitions introduced. In this section, we will use the entropic uncertainty principle~(theorem~\ref{uncertainty}) to derive a lower bound on the classical entropy of $\overline{\mu}_{\hbar}^{\overline{\Sigma}}$ with respect to the partition~$\overline{\mathcal{C}}_{\hbar}:=([\alpha]\times[0,f(\alpha)[)_{\alpha\in I^{\eta}(\hbar)}$. Precisely, we will prove:
\begin{prop}\label{lowerbound} With the notations of section~\ref{proof}, one has:
\begin{equation}\label{estQentr}
H\left(\overline{\mu}_{\hbar}^{\overline{\Sigma}_+},\overline{\mathcal{C}}_{\hbar}^+\right)+H\left(\overline{\mu}_{\hbar}^{\overline{\Sigma}_-},\overline{\mathcal{C}}^-_{\hbar}\right)\geq(1-\epsilon')(1-\epsilon)|\log\hbar|-c\delta_0|\log\hbar|+C,\end{equation}
where $H$ is defined by~(\ref{defent0}) and where $C,c\in\mathbb{R}$ does not depend on $\hbar$.
\end{prop}
To prove this result, we will proceed in three steps. First, we will introduce an energy cutoff in order to get the sharpest bound as possible in the entropic uncertainty principle. Then, we will apply the entropic uncertainty principle and derive a lower bound on $H\left(\overline{\mu}_{\hbar}^{\overline{\Sigma}_+},\overline{\mathcal{C}}_{\hbar}^+\right)+H\left(\overline{\mu}_{\hbar}^{\overline{\Sigma}_-},\overline{\mathcal{C}}^-_{\hbar}\right)$. Finally, we will use sharp estimates from~\cite{AKN} to conclude.

\subsubsection{Energy cutoff}

\label{cut}

Before applying the uncertainty principle, we proceed to sharp energy cutoffs so as to get precise lower bounds on the quantum entropy (as it was done in \cite{An}, \cite{AN2} and \cite{AKN}). These cutoffs are made in our microlocal analysis in order to get as good exponential decrease as possible of the norm of the refined quantum partition. This cutoff in energy is possible because even if the distributions $\mu_{\hbar}$ are defined on $T^*M$, they concentrate on the energy layer $S^*M$. The following energy localization is made in a way to compactify the phase space and in order to preserve the semiclassical measure.\\
Let $\delta_0$ be a positive number less than $1$ and $\chi_{\delta_0}(t)$ in $\mathcal{C}^{\infty}(\mathbb{R},[0,1])$. Moreover, $\chi_{\delta_0}(t)=1$ for $|t|\leq e^{-\delta_0 /2}$ and $\chi_{\delta_0}(t)=0$ for $|t|\geq 1$. As in~\cite{AN2}, the sharp $\hbar$-dependent cutoffs are then defined in the following way:
$$\forall\hbar\in(0,1),\ \forall n\in\mathbb{N},\ \forall\rho\in T^*M,\ \ \ \ \chi^{(n)}(\rho,\hbar):=\chi_{\delta_0}(e^{-n\delta_0}\hbar^{-1+\delta_0}(H(\rho)-1/2)).$$
For $n$ fixed, the cutoff $\chi^{(n)}$ is localized in an energy interval of length $2e^{n\delta_0}\hbar^{1-\delta_0}$ centered around the energy layer $\mathcal{E}$. In this paper, indices $n$ will satisfy $2e^{n\delta_0}\hbar^{1-\delta_0}<<1$. It implies that the widest cutoff is supported in an energy interval of microscopic length and that $n\leq K_{\delta_0}|\log\hbar|$, where $K_{\delta_0}\leq\delta_0^{-1}$. Using then a non standard pseudodifferential calculus (see~\cite{AN2} for a brief reminder of the procedure from~\cite{SZ}), one can quantize these cutoffs into pseudodifferential operators. We will denote $\Op(\chi^{(n)})$ the quantization of $\chi^{(n)}$. The main properties of this quantization are recalled in section~\ref{properties}. In particular, the quantization of these cutoffs preserves the eigenfunctions of the Laplacian, i.e.
$$\|\psi_{\hbar}-\Op(\chi^{(n)})\psi_{\hbar}\|=\mathcal{O}(\hbar^{\infty})\|\psi_{\hbar}\|.$$

\subsubsection{Applying the entropic uncertainty principle}
\label{several}
 Let $\|\psi_{\hbar}\|=1$ be a fixed element of the sequence of eigenfunctions of the Laplacian defined earlier, associated to the eigenvalue $-\frac{1}{\hbar^2}$.\\
To get bound on the entropy of the suspension measure, the entropic uncertainty principle should not be applied to the eigenvectors $\psi_{\hbar}$ directly but it will be applied several times. Precisely, we will apply it to each $P_{\gamma}\psi_{\hbar}:=P_{\gamma_1}P_{\gamma_0}(-\eta)\psi_{\hbar}$ where $\gamma=(\gamma_0,\gamma_1)$ varies in $\{1,\cdots, K\}^2$. In order to apply the entropic uncertainty principle to $P_{\gamma}\psi_{\hbar}$, we introduce new families of quantum partitions corresponding to each $\gamma$.\\
Let $\gamma=(\gamma_0,\gamma_1)$ be an element of $\{1,\cdots,K\}^2$. Introduce the following families of indices:
$$I_{\hbar}(\gamma):=\left\{(\alpha'):\gamma.\alpha'\in I^{\eta}(\hbar)\right\},$$
$$K_{\hbar}(\gamma):=\left\{(\beta'):\beta'.\gamma\in K^{\eta}(\hbar)\right\}.$$
Recall that we have defined $\gamma.\alpha'=(\gamma_0,\gamma_1,\alpha')$ in section~\ref{Anosov}. We underline that each sequence $\alpha$ of $I^{\eta}(\hbar)$ can be written under the form $\gamma.\alpha'$ where $\alpha'\in I_{\hbar}(\gamma)$. The same works for $K^{\eta}(\hbar)$. The following partitions of identity can be associated to these new families, for $\alpha'\in I_{\hbar}(\gamma)$ and $\beta'\in K_{\hbar}(\gamma)$,
$$\tilde{\tau}_{\alpha'}=P_{\alpha_n'}(n\eta)\cdots P_{\alpha_2'}(2\eta),$$
$$\tilde{\pi}_{\beta'}=P_{\beta_{-n}'}(-n\eta)\cdots P_{\beta_{-2}'}(-2\eta).$$
For analogous reasons as the case of $I^{\eta}(\hbar)$, the families $(\tilde{\tau}_{\alpha'})_{\alpha'\in I_{\hbar}(\gamma)}$ and $(\tilde{\pi}_{\beta'})_{\beta'\in I_{\hbar}(\gamma)}$ form quantum partitions of identity.\\
Given these new quantum partitions of identity, the entropic principle should be applied for given initial conditions $\gamma=(\gamma_0,\gamma_1)$ in times $0$ and $1$. We underline that for $\alpha'\in I_{\hbar}(\gamma)$ and $\beta'\in K_{\hbar}(\gamma)$,
\begin{equation}\label{switch}\tilde{\tau}_{\alpha'}U^{-\eta}P_{\gamma}=\tau_{\gamma.\alpha'}U^{-\eta}\ \text{and}\ \tilde{\pi}_{\beta'}P_{\gamma}=\pi_{\beta'.\gamma},\end{equation}
where $\gamma.\alpha'\in I^{\eta}(\hbar)$ and $\beta'.\gamma\in K^{\eta}(\hbar)$ by definition. In equality~(\ref{switch}) appears the fact that the definitions of $\tau$ and $\pi$ are slightly different (see~(\ref{tau}) and~(\ref{pi})). It is due to the fact that we want to compose $\tilde{\tau}$ and $\tilde{\pi}$ with the same operator $P_{\gamma}$.\\
Suppose now that $\|P_{\gamma}\psi_{\hbar}\|$ is not equal to $0$. We apply the quantum uncertainty principle~(\ref{uncertainty}) using that
\begin{itemize}
\item $(\tilde{\tau}_{\alpha'})_{\alpha'\in I_{\hbar}(\gamma)}$ and $(\tilde{\pi}_{\beta'})_{\beta'\in K_{\hbar}(\gamma)}$ are partitions of identity;
\item the cardinal of $I_{\hbar}(\gamma)$ and $K_{\hbar}(\gamma)$ is bounded by $\mathcal{N}\simeq \hbar^{-K_0}$ where $K_0$ is some fixed positive number (depending on the cardinality of the partition $K$, on $a_0$, on $b_0$ and $\eta$);
\item $\Op(\chi^{(k')})$ is a family of bounded bounded operators $O_{\beta'}$ (where $k'$ is the length of $\beta'$);
\item the parameter $\delta'$ can be taken equal to $\|P_{\gamma}\psi_{\hbar}\|^{-1}\hbar^L$ where $L$ is such that $\hbar^{L-K_0}\ll \hbar^{1/2(1-\epsilon')(1-\epsilon)-c\delta_0}$ for a given constant $c$ (see corollary~\ref{hyp});
\item $U^{-\eta}$ is an isometry;
\item $\tilde{\psi_{\hbar}}:=\frac{P_{\gamma}\psi_{\hbar}}{\|P_{\gamma}\psi_{\hbar}\|}$ is a normalized vector.
\end{itemize}
Applying the entropic uncertainty principle~(\ref{uncertainty}), one gets:
\begin{coro}\label{UP} Suppose that $\|P_{\gamma}\psi_{\hbar}\|$ is not equal to $0$. Then, one has
$$h_{\tilde{\tau}}(U^{-\eta}\tilde{\psi}_{\hbar})+h_{\tilde{\pi}}(\tilde{\psi}_{\hbar})\geq-2\log \left(c_{\chi}^{\gamma}(U^{-\eta})+\hbar^{L-K_0}\|P_{\gamma}\psi_{\hbar}\|^{-1}\right),$$
where  $\displaystyle c_{\chi}^{\gamma}(U^{-\eta})=\max_{\alpha'\in I_{\hbar}(\gamma),\beta'\in K_{\hbar}(\gamma)}\left(\|\tilde{\tau}_{\alpha'}U^{-\eta}\tilde{\pi}_{\beta'}^* \Op(\chi^{(k')})\|\right).$
\end{coro}
Under this form, the quantity $\|P_{\gamma}\psi_{\hbar}\|^{-1}$ appears several times and we would like to get rid of it. First, remark that the quantity $c_{\chi}^{\gamma}(U^{-\eta})$ can be easily replaced by \begin{equation}\label{norm}c_{\chi}(U^{-\eta}):=\max_{\gamma\in\{1,\cdots,K\}^2}\max_{\alpha'\in I_{\hbar}(\gamma),\beta'\in K_{\hbar}(\gamma)}\left(\|\tilde{\tau}_{\alpha'}U^{-\eta}\tilde{\pi}_{\beta'}^*\Op(\chi^{(k')})\|\right),\end{equation}
which is independent of $\gamma$. Then, one has the following lower bound:
\begin{equation}\label{bound}-2\log \left(c_{\chi}^{\gamma}(U^{-\eta})+\hbar^{L-K_0}\|P_{\gamma}\psi_{\hbar}\|^{-1}\right)\geq-2\log \left(c_{\chi}(U^{-\eta})+\hbar^{L-K_0}\right)+2\log\|P_{\gamma}\psi_{\hbar}\|^{2}.\end{equation}
as $\|P_{\gamma}\psi_{\hbar}\|\leq1$. Now that we have given an alternative lower bound, we rewrite $h_{\tilde{\tau}}(U^{-\eta}\tilde{\psi}_{\hbar})$ as follows:
$$h_{\tilde{\tau}}(U^{-\eta}\tilde{\psi}_{\hbar})=-\sum_{\alpha'\in I_{\hbar}(\gamma)}\|\tilde{\tau}_{\alpha'}U^{-\eta}\tilde{\psi}_{\hbar}\|^2\log\|\tilde{\tau}_{\alpha'}U^{-\eta}P_{\gamma}\psi_{\hbar}\|^2+\sum_{\alpha'\in I_{\hbar}(\gamma)}\|\tilde{\tau}_{\alpha'}U^{-\eta}\tilde{\psi}_{\hbar}\|^2\log\|P_{\gamma}\psi_{\hbar}\|^2.$$
Using the fact that $\psi_{\hbar}$ is an eigenvector of $U^{\eta}$ and that $(\tilde{\tau}_{\alpha'})_{\alpha'\in I_{\hbar}(\gamma)}$ is a partition of identity, one has
$$h_{\tilde{\tau}}(U^{-\eta}\tilde{\psi}_{\hbar})=-\frac{1}{\|P_{\gamma}\psi_{\hbar}\|^2}\sum_{\alpha'\in I_{\hbar}(\gamma)}\|\tau_{\gamma.\alpha'}\psi_{\hbar}\|^2\log\|\tau_{\gamma.\alpha'}\psi_{\hbar}\|^2+\log\|P_{\gamma}\psi_{\hbar}\|^2.$$
The same holds for $h_{\tilde{\pi}}(\tilde{\psi}_{\hbar})$ (using here equality~(\ref{switch})):
$$h_{\tilde{\pi}}(\tilde{\psi}_{\hbar})=-\frac{1}{\|P_{\gamma}\psi_{\hbar}\|^2}\sum_{\beta'\in K_{\hbar}(\gamma)}\|\pi_{\beta'.\gamma}\psi_{\hbar}\|^2\log\|\pi_{\beta'.\gamma}\psi_{\hbar}\|^2+\log\|P_{\gamma}\psi_{\hbar}\|^2.$$
Combining these last two equalities with~(\ref{bound}), we find that
\begin{equation}\label{inter}
-\sum_{\alpha'\in I_{\hbar}(\gamma)}\|\tau_{\gamma.\alpha'}\psi_{\hbar}\|^2\log\|\tau_{\gamma.\alpha'}\psi_{\hbar}\|^2-\sum_{\beta'\in K_{\hbar}(\gamma)}\|\pi_{\beta'.\gamma}\psi_{\hbar}\|^2\log\|\pi_{\beta'.\gamma}\psi_{\hbar}\|^2\geq-2\|P_{\gamma}\psi_{\hbar}\|^{2}\log \left(c_{\chi}(U^{-\eta})+\hbar^{L-K_0}\right).
\end{equation}
We underline that this lower bound is trivial in the case where $\|P_{\gamma}\psi_{\hbar}\|$ is equal to $0$. Using the following numbers:
\begin{equation}\label{cgamma}c_{\gamma.\alpha'}=c_{\beta'.\gamma}=c_{\gamma}=\frac{f(\gamma)}{\sum_{\gamma'\in\{1,\cdots,K\}^2}f(\gamma')\|P_{\gamma'}\psi_{\hbar}\|^2},\end{equation}
one easily checks that $\displaystyle\sum_{\gamma\in\{1,\cdots,K\}^2} c_{\gamma}\|P_{\gamma}\psi_{\hbar}\|^2=1$. If we multiply~(\ref{inter}) by $c_{\gamma}$ and make the sum over all $\gamma$ in $\{1,\cdots,K\}^2$, we find
$$-\sum_{\alpha\in I^{\eta}(\hbar)}c_{\alpha}\|\tau_{\alpha}\psi_{\hbar}\|^2\log\|\tau_{\alpha}\psi_{\hbar}\|^2-\sum_{\beta\in K^{\eta}(\hbar)}c_{\beta}\|\pi_{\beta}\psi_{\hbar}\|^2\log\|\pi_{\beta}\psi_{\hbar}\|^2\geq-2\log \left(c_{\chi}(U^{-\eta})+\hbar^{L-K_0}\right).$$
Finally, we use that $\displaystyle\sum_{\alpha\in I^{\eta}(\hbar)} c_{\alpha}\|\tau_{\alpha}\psi_{\hbar}\|^2=1$ and $\displaystyle\sum_{\beta\in K^{\eta}(\hbar)} c_{\beta}\|\pi_{\beta}\psi_{\hbar}\|^2=1$ and derive the following property:
\begin{coro}\label{corouncert}
One has:
\begin{equation}\label{uncerteigenfct}H\left(\overline{\mu}_{\hbar}^{\overline{\Sigma}_+},\overline{\mathcal{C}}_{\hbar}^+\right)+H\left(\overline{\mu}_{\hbar}^{\overline{\Sigma}_-},\overline{\mathcal{C}}^-_{\hbar}\right)\geq-2\log \left( c_{\chi}(U^{-\eta})+\hbar^{L-K_0}\right)-\log\left(\max_{\gamma}c_{\gamma}\right).\end{equation}
\end{coro}
As expected, by a careful use of the entropic uncertainty principle, we have been able to obtain a lower bound on the entropy of the measures $\overline{\mu}_{\hbar}^{\overline{\Sigma}_+}$ and $\overline{\mu}_{\hbar}^{\overline{\Sigma}_-}$.

\subsubsection{Exponential decrease of the atoms of the quantum partition}
\label{Normest}
Now that we have obtained the lower bound~(\ref{uncerteigenfct}), we give an estimate on the exponential decrease of the atoms of the quantum partition. As in~\cite{An},~\cite{AN2},~\cite{AKN}, one has\footnote{In the higher dimension case mentioned in the footnote of section~\ref{heuristic}, we should replace $\hbar^{-\frac{1}{2}}$ (where $d$ is the dimension of $M$) by $\hbar^{-\frac{d-1}{2}}$ in inequality~(\ref{ineqnorm}).}:
\begin{theo}\label{normest}~\cite{An}~\cite{AN2}~\cite{AKN} For every $\mathcal{K}>0$ ($\mathcal{K}\leq K_{\delta_0}$), there exists $\hbar_{\mathcal{K}}$ and $C_{\mathcal{K}}$ such that uniformly for all $\hbar\leq\hbar_{\mathcal{K}}$, for all $k+k'\leq\mathcal{K}|\log\hbar|$,
$$\|P_{\alpha_{k}}U^{\eta}P_{\alpha_{k-1}}\cdots U^{\eta}P_{\alpha_0}U^{3\eta}P_{\alpha_k'}U^{\eta}\cdots P_{\alpha_0'}\Op(\chi^{(k')})\|_{L^2(M)}\ \ \ \ \ \ \ \ \ \ \ \ \ \ \ \ \ \ \ \ $$
\begin{equation}\label{ineqnorm}\ \ \ \ \ \ \ \ \ \ \ \ \ \ \ \ \ \ \ \ \ \ \ \ \ \leq C_{\mathcal{K}}\hbar^{-\frac{1}{2}-c\delta_0}\exp\left(-\frac{1}{2}\left(\sum_{j=0}^{k-1}f(\sigma^j\alpha)+\sum_{j=0}^{k'-1}f(\sigma^j\alpha')\right)\right),\end{equation}
where $c$ depends only on the riemannian manifold $M$.\end{theo}
Outline that the crucial role of the sharp energy cutoff appears in particular to prove this theorem. In fact, without the cutoff, the previous norm operator could have only be bounded by $1$ and the entropic uncertainty principle would have been empty. The previous inequality~(\ref{ineqnorm}) allows to give an estimate on the quantity~(\ref{norm}) (as it allows us to bound $c_{\chi}(U^{-\eta})$). In fact, one has, for each $\gamma\in\{1,\cdots,K\}^2$:
$$\|\tilde{\tau}_{\alpha}U^{-\eta}\tilde{\pi}_{\beta}^*\Op(\chi^{(k')})\|=\|P_{\alpha_k}U^{\eta}P_{\alpha_{k-1}}\cdots U^{\eta}P_{\alpha_2}U^{3\eta}P_{\beta_{-2}}U^{\eta}\cdots P_{\beta_{-k'}}\Op(\chi^{(k')})\|,$$
where $(\alpha_2,\cdots,\alpha_k)\in I_{\hbar}(\gamma)$ and $(\beta_{-k'},\cdots,\beta_{-2})\in K_{\hbar}(\gamma)$. Using the definition of the sets $I^{\eta}(\hbar)$~(\ref{Ih}) and $K^{\eta}(\hbar)$~(\ref{Kh}), one has $k+k'\leq\frac{2}{a_0\eta}|\log\hbar|$. Using theorem~(\ref{normest}) with $\mathcal{K}=\frac{2}{a_0\eta}$, one has:
$$\|\tilde{\tau}_{\alpha}U^{-\eta}\tilde{\pi}_{\beta}^*\Op(\chi^{(k')})\|\leq C_{\mathcal{K}}\hbar^{-\frac{1}{2}-c\delta_0}\exp\left(-\frac{1}{2}\left(\sum_{j=2}^{k-1}f_+(\sigma^j_+\alpha)+\sum_{j=2}^{k'-1}f_-(\sigma_-^j\beta)\right)\right),$$
where $C_{\mathcal{K}}$ does not depend on $\hbar$ and $c$ is some universal constant. Using again the definition of the sets $I^{\eta}(\hbar)$~(\ref{Ih}) and $K^{\eta}(\hbar)$~(\ref{Kh}), one has
$$c_{\chi}(U^{-\eta})=\max_{\gamma\in\{1,\cdots,K\}^2}\max_{\alpha\in I_{\hbar}(\gamma),\beta\in K_{\hbar}(\gamma)}\left(\|\tilde{\tau}_{\alpha}U^{-\eta}\tilde{\pi}_{\beta}^*\Op(\chi^{(k')})\|\right)\leq \tilde{C}_{\mathcal{K}}\hbar^{\frac{1}{2}(1-\epsilon')(1-\epsilon)}\hbar^{-c\delta_0},$$
where $\tilde{C}_{\mathcal{K}}$ does not depend on $\hbar$. The main inequality~(\ref{uncerteigenfct}) for the quantum entropy can be rewritten using this last bound and it concludes the proof of proposition~\ref{lowerbound}.$\square$

\section{Subadditivity of the quantum entropy}

\label{commutativity}

As was mentioned in section~\ref{proof} and proved in section~\ref{partition}, the uncertainty principle gives an explicit lower bound on$$\frac{1}{n_E(\hbar)}\left(H_{n_E(\hbar)}\left(\overline{\mu}_{\hbar}^{\overline{\Sigma}_+},\overline{\sigma}_+,\overline{\mathcal{C}}_+\right)+H_{n_E(\hbar)}\left(\overline{\mu}_{\hbar}^{\overline{\Sigma}_-},\overline{\sigma}_-,\overline{\mathcal{C}}_-\right)\right).$$
To prove our main theorem~\ref{maintheo}, we need to show that this lower bound holds also for a fixed $n_0$ on the quantity $$\frac{1}{n_0}\left(H_{n_0}\left(\overline{\mu}_{\hbar}^{\overline{\Sigma}_+},\overline{\sigma}_+,\overline{\mathcal{C}}_+\right)+H_{n_0}\left(\overline{\mu}_{\hbar}^{\overline{\Sigma}_-},\overline{\sigma}_-,\overline{\mathcal{C}}_-\right)\right).$$
(as we need to let $\hbar$ tend to $0$ independently of $n$ to recover the semiclassical measure $\overline{\mu}^{\Sigma}$: see section~\ref{subad}). To do this we want to reproduce the classical argument for the existence of the metric entropy (see~(\ref{classicalsubad})), i.e. we need to prove a subadditivity property for logarithmic time as was given by theorem~\ref{subadditivity}. A key point to prove the subadditivity property in the case of the metric entropy is that the measure is invariant under the dynamics (see~(\ref{classicalsubad})). In our case, invariance of the semiclassical measure under the geodesic flow is a consequence of the Egorov property~(\ref{Egorov}): to prove that subadditivity almost holds (in the sense of the previous theorem), we will have to prove an Egorov property for logarithmic times. We will see that with our choice of 'local' Ehrenfest time, this will be possible and the theorem~\ref{subadditivity} will then hold.\\
The proof of theorem~\ref{subadditivity} is the subject of this section (and it also uses results from section~\ref{bigpdotheo}).
\begin{rema} In this section, only the case of $\{1,\cdots,K\}^{\mathbb{N}}$ is treated. As was mentioned, the proof of the backward case $\{1,\cdots,K\}^{-\mathbb{N}}$ works in the same way.\end{rema}
Let $n_0$ and $m$ be two positive integers such that que $m+n_0\leq T_E(\hbar)$. One has
$$H\left(\vee_{i=0}^{n+n_0-1}\overline{\sigma}^{-i}\overline{\mathcal{C}},\overline{\mu}_{\hbar}^{\overline{\Sigma}}\right)=H\left(\vee_{i=0}^{n-1}\overline{\sigma}^{-i}\overline{\mathcal{C}}\vee\vee_{i=n}^{n_0+n-1}\overline{\sigma}^{-i}\overline{\mathcal{C}},\overline{\mu}_{\hbar}^{\overline{\Sigma}}\right).$$
Using classical properties of the metric entropy, one has (see section~\ref{KSentropy})
$$H_{n+n_0}\left(\overline{\sigma},\overline{\mu}_{\hbar}^{\overline{\Sigma}},\overline{\mathcal{C}}\right)\leq H_n\left(\overline{\sigma},\overline{\mu}_{\hbar}^{\overline{\Sigma}},\overline{\mathcal{C}}\right)+H_{n_0}\left(\overline{\sigma},\overline{\sigma}^{n}\sharp\overline{\mu}_{\hbar}^{\overline{\Sigma}},\overline{\mathcal{C}}\right).$$
Using proposition~\ref{generalcase} and the continuity of the function $x\log x$ on $[0,1]$, there exists a function $R(n_0,\hbar)$ with the properties of theorem~\ref{subadditivity} such that $\displaystyle H_{n_0}\left(\overline{\sigma},\overline{\sigma}^{n}\sharp\overline{\mu}_{\hbar}^{\overline{\Sigma}},\overline{\mathcal{C}}\right)=H_{n_0}\left(\overline{\sigma},\overline{\mu}_{\hbar}^{\overline{\Sigma}},\overline{\mathcal{C}}\right)+R(n_0,\hbar)$ and thus:
\begin{equation}\label{pseudosubad}H_{n+n_0}\left(\overline{\sigma},\overline{\mu}_{\hbar}^{\overline{\Sigma}},\overline{\mathcal{C}}\right)\leq H_n\left(\overline{\sigma},\overline{\mu}_{\hbar}^{\overline{\Sigma}},\overline{\mathcal{C}}\right)+H_{n_0}\left(\overline{\sigma},\overline{\mu}_{\hbar}^{\overline{\Sigma}},\overline{\mathcal{C}}\right)+R(n_0,\hbar).\square\end{equation}
So the crucial point to prove this theorem is to show that the measure of the atoms of the refined partition is almost invariant under $\overline{\sigma}$ (proposition~\ref{generalcase}). In the following of this section, $A$ is defined as: $$A=\overline{\mathcal{C}}_{\gamma_0,p_0}\cap\cdots\cap\overline{\sigma}^{-(n_0-1)}\overline{\mathcal{C}}_{\gamma_{n_0-1},p_{n_0-1}}.$$

\subsection{Pseudo-invariance of the measure of the atoms of the partitions}

From this point, our main goal is to show the pseudo invariance of the atoms of the refined partition. More precisely:
\begin{prop}\label{generalcase} Let $m,n_0$ be two positive integers such that $m+n_0\leq T_E(\hbar)$. Consider an atom of the refined partition $A=\overline{\mathcal{C}}_{\gamma_0,p_0}\cap\cdots\cap\overline{\sigma}^{-(n_0-1)}\overline{\mathcal{C}}_{\gamma_{n_0-1},p_{n_0-1}}$. One has
$$\overline{\mu}_{\hbar}^{\overline{\Sigma}}\left(\overline{\sigma}^{-m}A\right)=\overline{\mu}_{\hbar}^{\overline{\Sigma}}\left(A\right)+\mathcal{O}(\hbar^{(1-2\nu)/6}),$$
with a uniform constant in $n_0$ and $m$ in the allowed interval. The constant $\nu<1/2$ is the one defined by theorem~\ref{pdotheo}.
\end{prop}
This result says that the measure $\overline{\mu}_{\hbar}^{\overline{\Sigma}}$ is almost $\overline{\sigma}$ invariant for logarithmic times. As a consequence, the classical argument (see~(\ref{classicalsubad})) for subadditivity of the entropy can be applied as long as we consider times where the pseudo invariance holds (see~(\ref{pseudosubad})).\\
Let $A$ be as in the proposition. From lemma~\ref{adapted}, there exists $(\alpha_0,\cdots,\alpha_k)$ and $B(\gamma)$ such that
$$A=\left(\mathcal{C}_{\alpha_0}\cap\cdots\sigma^{-k}\mathcal{C}_{\alpha_k}\right)\times B(\gamma).$$
Still from lemma~\ref{adapted}, one knows that $B(\gamma)$ is a subinterval of $[0,f(\gamma_0)[$. Moreover, from the proof of lemma~\ref{adapted}, the following property on $\alpha$ holds:
\begin{equation}\label{n0}n_0(1-\epsilon)\leq\sum_{j=0}^{k-1}f(\sigma^j\alpha)\leq n_0(1+\epsilon).\end{equation}
The plan of the proof of proposition~\ref{generalcase} is the following. First, we will give an exact expression in terms of $\alpha$ and $B(\gamma)$ of $\overline{\mu}_{\hbar}^{\overline{\Sigma}}\left(\overline{\sigma}^{-m}A\right)$. Then, we will see how to prove the proposition making the simplifying assumption that all operators $(P_i(k\eta))_{i,k}$ commute. Finally, we will estimate the error term due to the fact that operators do not exactly commute.

\subsubsection{Computation of $\overline{\mu}_{\hbar}^{\overline{\Sigma}}(\overline{\sigma}^{-m}A)$}
\label{case1}

We choose a positive integer $m$. As a first step of the proof, we want to give a precise formula for the measure of $\overline{\sigma}^{-m}A$. To do this, we have to determine the shape of the set $\overline{\sigma}^{-m}A$. Let us then define:
$$\overline{\Sigma}_p^m:=\left\{(x,t)\in\overline{\Sigma}:\sum_{j=0}^{p-2}f(\sigma^jx)\leq m+t<\sum_{j=0}^{p-1}f(\sigma^jx)\right\}.$$
We underline that because $m\geq 1$, we have that $\overline{\Sigma}_p^m$ is empty for $p\leq 3$. One has then $\displaystyle\overline{\Sigma}=\bigsqcup_{p\geq3}\overline{\Sigma}_p^m$ and as a consequence
\begin{eqnarray*}\overline{\sigma}^{-m}A &=&\bigsqcup_{p\geq3}\left(\overline{\Sigma}_p^m\cap\overline{\sigma}^{-m}A\right)\\&=&\bigsqcup_{p\geq3}\left\{(x,t)\in\overline{\Sigma}_p^m:m+t-\sum_{j=0}^{p-2}f(\sigma^jx)\in B(\gamma),(x_{p-1},\cdots,x_{p+k-1})=\alpha\right\}.\end{eqnarray*}
Note that $t\in B(\gamma)-m+\sum_{j=0}^{p-2}f(\sigma^jx)$ together with $(x_{p-1},\cdots,x_{p+k-1})=\alpha$ imply that $\sum_{j=0}^{p-2}f(\sigma^jx)\leq m+t<\sum_{j=0}^{p-1}f(\sigma^jx)$. It allows to rewrite
$$\overline{\sigma}^{-m}A=\bigsqcup_{p\geq3}\left\{(x,t)\in\Sigma\times\mathbb{R}_{+}:0\leq t<f(x),t\in B(\gamma)-m+\sum_{j=0}^{p-2}f(\sigma^jx),(x_{p-1},\cdots,x_{p+k-1})=\alpha\right\}.$$
Finally, one can write the measure of this suspension set
$$\overline{\mu}_{\hbar}^{\overline{\Sigma}}\left(\overline{\sigma}^{-m}A\right)=\sum_{p\geq1}\sum_{\tiny{\begin{array}{c}|\beta|=p+k\\(\beta_{p-1},\cdots,\beta_{p+k-1})=\alpha\end{array}}}c_{\beta,\alpha}(m)\|P_{\beta_{k+p-1}}((k+p-1)\eta)P_{\beta_{k+p-2}}((k+p-2)\eta)\cdots P_{\beta_{0}}\psi_{\hbar}\|^2,$$
where
$$c_{\beta,\alpha}(m)=\text{Leb}\left(B(\gamma)\cap [m-\sum_{j=0}^{p-2}f(\sigma^{j}\beta),m-\sum_{j=1}^{p-2}f(\sigma^{j}\beta)[\right)/\left(\sum_{\gamma'\in\{1,\cdots,K\}^2}f(\gamma')\mu_{\hbar}^{\Sigma}([\gamma'])\right).$$
For the sake of simplicity, we will denote $\lambda$ the normalization constant of the measure, i.e. 
$$\lambda^{-1}:=\sum_{\gamma'\in\{1,\cdots,K\}^2}f(\gamma')\mu_{\hbar}^{\Sigma}([\gamma']).$$  
Outline that the previous sum runs a finite number of $p$ with at most $2b_0/a_0$ non zeros terms in each string $\beta$ (as $c_{\bullet,\alpha}(m)$ is zero except a finite number of times). For simplicity of the following of the proof, we reindex the previous expressions
\begin{equation}\label{simple}\overline{\mu}_{\hbar}^{\overline{\Sigma}}\left(\overline{\sigma}^{-m}A\right)=\sum_{p\geq3}\sum_{\tiny{\begin{array}{c}|\beta|=p+k\\(\beta_{0},\cdots,\beta_{k})=\alpha\end{array}}}c_{\beta,\alpha}(m)\|P_{\beta_k}((k+p-1)\eta)P_{\beta_{k-1}}((k+p-2)\eta)\cdots P_{\beta_{-p+1}}\psi_{\hbar}\|^2,\end{equation}
where $c_{\beta,\alpha}(m)=\lambda\ \text{Leb}\left(B(\gamma)\cap [m-\sum_{j=0}^{p-2}f(\sigma^{j}\beta),m-\sum_{j=1}^{p-2}f(\sigma^{j}\beta)[\right)$ with $\lambda$ defined as previously. Then, to prove proposition~\ref{generalcase}, we have to show that the previous quantity~(\ref{simple}) is equal to
$$\lambda\ \text{Leb}\left(B(\gamma)\right)\|P_{\alpha_k}(k\eta)\cdots P_{\alpha_0}\psi_{\hbar}\|^2_{L^2}+\mathcal{O}_{L^2}(\hbar^{(1-2\nu)/6}).$$

\subsubsection{If everything would commute...}

We will now use our explicit expression for $\overline{\mu}_{\hbar}^{\Sigma}\left(\overline{\sigma}^{-m}A\right)$ (see~(\ref{simple})) and verify it is equal to $\overline{\mu}_{\hbar}^{\Sigma}\left(A\right)$ under the simplifying assumption that all the involved pseudodifferential operators commute. In the next section, we will then give an estimate of the error term due to the fact that the operators do not exactly commute. In order to prove the pseudo invariance, denote
$$K_{m}(\alpha):=\left\{\beta=\left(\beta_{-p+1},\cdots,\beta_{k}\right):(\beta_0,\cdots,\beta_k)=\alpha,c_{\beta,\alpha}(m)\neq 0\right\}$$
and
$$K^{(q)}_m(\alpha):=\left\{\left(\beta_{-q+1},\cdots,\beta_{k}\right):\exists \gamma=(\gamma_{-p+1},\cdots,\gamma_{-q})\ \text{s.t.}\ q< p,\ \gamma.\beta\in K_m(\alpha)\right\}.$$
With these notations, we can write~(\ref{simple}) as follows:
\begin{equation}\label{inductionserie}\overline{\mu}_{\hbar}^{\overline{\Sigma}}\left(\overline{\sigma}^{-m}A\right)=\sum_{\beta\in K_m(\alpha)}c_{\beta,\alpha}(m)\|\tau_{\beta}\psi_{\hbar}\|^2=\sum_{p=3}^N\sum_{\beta\in K_m(\alpha):|\beta|=k+p}c_{\beta,\alpha}(m)\|\tau_{\beta}\psi_{\hbar}\|^2.\end{equation}
Recall that by definition (see~(\ref{tau})) $\tau_{\beta}:=P_{\beta_k}((k+p-1)\eta)P_{\beta_{k-1}}((k+p-2)\eta)\cdots P_{\beta_{-p+1}}$. For simplicity of notations, let us denote $B(\gamma)=[a,b[$ (where $a$ and $b$ obviously depend on $\gamma$). A last notation we define is for $\beta$ such that $|\beta|=k+q$ and $\sigma^{q-1}\beta=\alpha$:
\begin{equation}\label{cbar}\overline{c}_{\beta,\alpha}(m):=\lambda\ \text{Leb}\left([a,b[\cap[a,m-\sum_{j=1}^{q-2}f(\sigma^{j}\beta)[\right),\end{equation}
where $\lambda$ is the normalization constant of the measure previously defined. We underline that the interval $B(\gamma)=[a,b[$ can be divided in smaller intervals (see the definition of $c_{\beta,\alpha}(m)$). The number $c_{\beta,\alpha}(m)$ corresponds to the length of one of this subinterval (weighted by $\lambda$) and $\overline{c}_{\beta,\alpha}(m)$ corresponds to the sum of the lengths of the first intervals. Suppose now that all the operators $(P_i(k\eta))_{i,k}$ commute. We have the following lemma:
\begin{lemm} If all the operators $(P_i(k\eta))_{i,k}$ commute, then one has, for $2\leq q\leq N$:
$$\sum_{\beta\in K_m(\alpha):|\beta|=k+q}c_{\beta,\alpha}(m)\|\tau_{\beta}\psi_{\hbar}\|^2+\sum_{\beta\in K_m^{(q)}(\alpha)}\overline{c}_{\beta,\alpha}(m)\|\tau_{\beta}\psi_{\hbar}\|^2=\sum_{\beta\in K_{m}^{(q-1)}(\alpha)}\overline{c}_{\beta,\alpha}(m)\|\tau_{\beta}\psi_{\hbar}\|^2.$$
\end{lemm}
\begin{proof} Let $2\leq q\leq N$. Consider $\beta$ an element of $K_{m}^{(q-1)}(\alpha)$. Using the property of partition of identity, we have
$$\sum_{\beta\in K_{m}^{(q-1)}(\alpha)}\overline{c}_{\beta,\alpha}(m)\|\tau_{\beta}\psi_{\hbar}\|^2=\sum_{j=1}^K\sum_{\beta\in K_{m}^{(q-1)}(\alpha)}\overline{c}_{\beta,\alpha}(m)\|P_{j}(-\eta)\tau_{\beta}\psi_{\hbar}\|^2.$$
For each $1\leq j\leq K$, we have the following property for $\overline{c}_{j.\beta,\alpha}(m)$ (as $f\geq 0$):
$$\overline{c}_{\beta,\alpha}(m)=\overline{c}_{j.\beta,\alpha}(m)+c_{j.\beta,\alpha}(m).$$
We can write then
$$\sum_{\beta\in K_{m}^{(q-1)}(\alpha)}\overline{c}_{\beta,\alpha}(m)\|\tau_{\beta}\psi_{\hbar}\|^2=\sum_{j=1}^K\sum_{\beta\in K_{m}^{(q-1)}(\alpha)}(\overline{c}_{j.\beta,\alpha}(m)+c_{j.\beta,\alpha}(m))\|P_{j}(-\eta)\tau_{\beta}\psi_{\hbar}\|^2.$$
Notice that, as we have assumed the operators commute, we have
\begin{equation}\label{commutation}P_{j}(-\eta)P_{\beta_k}((k+q-2)\eta)\cdots P_{\beta_{-q+2}} \psi_{\hbar}=P_{\beta_k}((k+q-1)\eta)\cdots P_{\beta_{-q+2}}P_{j}(-\eta)\psi_{\hbar}.\end{equation}
As a consequence, we have
$$\sum_{\beta\in K_{m}^{(q-1)}(\alpha)}\overline{c}_{\beta,\alpha}(m)\|\tau_{\beta}\psi_{\hbar}\|^2=\sum_{j=1}^K\sum_{\beta\in K_{m}^{(q-1)}(\alpha)}(\overline{c}_{j.\beta,\alpha}(m)+c_{j.\beta,\alpha}(m))\|\tau_{\beta}P_{j}(-\eta)\psi_{\hbar}\|^2.$$
By definition of the different sets $K_m$ and as $\psi_{\hbar}$ is an eigenfunction of the Laplacian, this last equality allows to conclude the proof of the lemma.\end{proof}
Proceeding then by induction from $N$ to $1$ (see equality~(\ref{inductionserie})) and using the previous lemma at each step, we can conclude that if all the operators commute,
$$\overline{\mu}_{\hbar}^{\overline{\Sigma}}\left(\overline{\sigma}^{-m}A\right)=\overline{\mu}_{\hbar}^{\overline{\Sigma}}\left(A\right).$$

\subsubsection{Estimates of the error terms}

Regarding to the previous section, we have to see what is exactly the error term we forgot at each step of the recurrence and we have to verify that it is bounded by some positive power of $\hbar$. Precisely, we have to understand what is the error term in equation~(\ref{commutation}) if we do not suppose anymore that all the operators commute. Precisely, the error term we have to take into account in~(\ref{commutation}) is
$$R_{\beta,\gamma,\hbar}=\sum_{j=-q+2}^k P_{\beta_k}((k+q-2)\eta)\cdots P_{\beta_{j+1}}((j+q-1)\eta) R^{j}(\beta,\gamma)P_{\beta_{j-1}}((j+q-3)\eta)\cdots P_{\beta_{-q+2}} \psi_{\hbar},$$
where $R^{j}(\beta,\gamma)=[P_{\gamma}(-\eta),P_{\beta_{j}}((j+q-2)\eta)]$ is the bracket of the two operators. We denote $R_{\beta,\gamma,\hbar}^j$ each term of the previous sum. The error term we forgot at each step $q$ of the induction in the previous section is then
\begin{equation}\label{qerror}E(\hbar,q):=\sum_{\gamma=1}^K\sum_{\beta\in K_{m}^{(q-1)}(\alpha)}\left(\left\langle R_{\beta,\gamma,\hbar},P_{\gamma}(-\eta)\tau_{\beta}\psi_{\hbar}\right\rangle+\left\langle \tau_{\beta}P_{\gamma}(-\eta)\psi_{\hbar},R_{\beta,\gamma,\hbar}\right\rangle\right).\end{equation}
So, for each step $q$ of the induction, if we want to prove the pseudo invariance of the measure, a first error term we have to estimate is of the form
\begin{equation}\label{Error}\sum_{j=-q+2}^k\sum_{\gamma=1}^K\sum_{\beta\in K_{m}^{(q-1)}(\alpha)}\overline{c}_{\beta,\alpha}(m)\left\langle  R_{\beta,\gamma,\hbar}^j,P_{\gamma}(-\eta)\tau_{\beta}\psi_{\hbar}\right\rangle.\end{equation}
Using Cauchy Schwarz inequality twice and the fact that $0\leq\overline{c}_{\beta,\alpha}(m)\leq\lambda\text{Leb}(B(\gamma))\leq \lambda b_0\eta$, this last quantity is bounded by
\begin{equation}\label{error}\lambda b_0\eta\left(\sum_{j=-q+2}^k\sum_{\gamma=1}^K\sum_{\beta\in K_{m}^{(q-1)}(\alpha)}\|  R_{\beta,\gamma,\hbar}^j\|^2\right)^{\frac{1}{2}}\left(\sum_{j=-q+2}^k\sum_{\gamma=1}^K\sum_{\beta\in K_{m}^{(q-1)}(\alpha)}\|P_{\gamma}(-\eta)\tau_{\beta}\psi_{\hbar}\|^2\right)^{\frac{1}{2}}.\end{equation}
The last term of the product is bounded as
$$\sum_{j=-q+2}^k\sum_{\gamma=1}^K\sum_{\beta\in K_{m}^{(q-1)}(\alpha)}\|P_{\gamma}(-\eta)\tau_{\beta}\psi_{\hbar}\|^2
\leq (k+q)K\sum_{|\beta|=k+q-1}\|\tau_{\beta}\psi_{\hbar}\|^2=(k+q)K=\mathcal{O}(|\log\hbar|).$$
We also underline that $\lambda b_0\eta$ is bounded by $b_0/a_0$. As a consequence, the error term~(\ref{error}) is bounded by
$$C|\log\hbar|\left(\sum_{j=-q+2}^k\sum_{\gamma=1}^K\sum_{\beta\in K_{m}^{(q-1)}(\alpha)}\|  R_{\beta,\gamma,\hbar}^j\|^2\right)^{\frac{1}{2}},$$
where $C$ is some positive uniform constant (depending only on the partition and on $\eta$). We extend now the definition of $R^j(\beta,\gamma)$ (previously defined as $[P_{\gamma}(-\eta),P_{\beta_{j}}((j+q-2)\eta)]$ for $\beta$ in $K_{m}^{(q-1)}(\alpha)$) to any word $\beta$ of length $k+q-1$. If $j+q-1$ letters of $\beta$ are also the $j+q-1$ first letters of a word $\beta'$ in $K_{m}^{(q-1)}(\alpha)$, we take $R^j(\beta,\gamma):=[P_{\gamma}(-\eta),P_{\beta_{j}}((j+q-2)\eta)]$. Otherwise, we take $R^j(\beta,\gamma):=\hbar\ \text{Id}_{L^{2}(M)}.$ We define then for any sequence of length $k+q-1$
$$R_{\beta,\gamma,\hbar}^j=P_{\beta_k}((k+q-2)\eta)\cdots P_{\beta_{j+1}}((j+q-1)\eta) R^{j}(\beta,\gamma)P_{\beta_{j-1}}((j+q-3)\eta)\cdots P_{\beta_{-q+2}} \psi_{\hbar}.$$
In theorem~\ref{comm2} from the section~\ref{mainproof}, we will prove in particular that for every $\beta$ of size $q+k-1$ and for each $-q+2\leq j\leq k$,
\begin{equation}\label{estimaterror}\|R^{j}(\beta,\gamma)P_{\beta_{j-1}}((j+q-3)\eta)\cdots P_{\beta_{-q+2}} \psi_{\hbar}\|_{L^2(M)}\leq C\hbar^{1-2\nu}\|P_{\beta_{j-1}}((j+q-3)\eta)\cdots P_{\beta_{-q+2}} \psi_{\hbar}\|_{L^2(M)},\end{equation}
where $C$ is a uniform constant for $n_0$ and $m$ positive integers such that $n_0+m\leq T_E(\hbar)$ and $\nu<1/2$ (defined in section~\ref{bigpdotheo}). We underline that the bracket $R^{j}(\beta,\gamma)$ of the two operators can commute (modulo $\hbar^{1-2\nu}$) because we have made a phase space localization thanks to the operator $P_{\beta_{j-1}}((j+q-3)\eta)\cdots P_{\beta_{-q+2}}$. Theorem~\ref{comm2} can be applied as $\sum_{j=0}^{k+q-2}f(\sigma^j\beta)\leq (n_0+m)(1+\epsilon)\leq n_E(\hbar)$ (see~(\ref{n0}) and~(\ref{cbar})). Using bound~(\ref{estimaterror}) and the property of partition of identity, we have
$$\sum_{|\beta|=k+q-1}\|  R_{\beta,\gamma,\hbar}^j\|^2=\mathcal{O}(\hbar^{2(1-2\nu)}).$$
The error term~(\ref{error}) (and as a consequence~(\ref{Error})) is then bounded by
$$\tilde{C}|\log\hbar|\left(\sum_{j=-q+2}^k\sum_{\gamma=1}^K\sum_{|\beta|=k+q-1}\|  R_{\beta,\gamma,\hbar}^j\|^2\right)^{\frac{1}{2}}=\mathcal{O}(\hbar^{\frac{1-2\nu}{4}}).$$
Looking at equation~(\ref{qerror}), we see that the other error term for the step $q$ of the induction can be estimated with the same method and is also a $\mathcal{O}(\hbar^{\frac{1-2\nu}{4}}).$ As the number $N$ of steps in the induction is a $\mathcal{O}(|\log\hbar|)$, the error term we forgot in the previous section (due to the fact that the operators do not commute) is a $\mathcal{O}(\hbar^{\frac{1-2\nu}{6}}).$ This concludes the proof of proposition~\ref{generalcase}.$\square$

\subsection{Commutation of pseudodifferential operators}
\label{mainproof}

In order to complete the proof of the pseudo invariance of the measure~(proposition~\ref{generalcase}), we need to prove inequality~(\ref{estimaterror}). It will be a consequence of~(\ref{pseudocomm}) below. Once we have proved this inequality, the subadditivity property will be completely proved. The exact property we need is stated by the following theorem:
\begin{theo}\label{comm2}
Let $(\gamma_0,\cdots,\gamma_k)$ be such that
\begin{equation}\label{allowedset}\sum_{j=0}^{k-1}f(\sigma^j\gamma)\leq n_E(\hbar).\end{equation} One has:
\begin{equation}\label{pseudocomm}\left\|\left[P_{\gamma_{k}}(k\eta),P_{\gamma_{0}}\right]P_{\gamma_{k-1}}((k-1)\eta)\cdots P_{\gamma_{1}}(\eta)\psi_{\hbar}\right\|_{L^2}\leq C \hbar^{1-2\nu} \left\|P_{\gamma_{k-1}}((k-1)\eta)\cdots P_{\gamma_{1}}(\eta)\psi_{\hbar}\right\|_{L^2},\end{equation}
where $\nu<1/2$ is defined in section~\ref{bigpdotheo}, $C$ is a constant depending on the partition and uniform in all $\gamma$ satisfying~(\ref{allowedset}).
\end{theo}
In this theorem, we underline that there are no particular reasons for the bracket $[P_{\gamma_{k}}(k\eta),P_{\gamma_{0}}]$ to be small: it will be in fact small thanks to the phase space localization induced by the operator $P_{\gamma_{k-1}}((k-1)\eta)\cdots P_{\gamma_{1}}(\eta)$.\\
Let $\gamma$ be a finite sequence as in the previous theorem. Denote $\displaystyle t(\gamma)=\sum_{j=0}^{k(\gamma)-1}f(\sigma^j\gamma)$. This quantity is less than $n_E(\hbar)$ in the setting of theorem~\ref{comm2}. There exists a unique integer $l(\gamma)<k(\gamma)$ such that:
$$\sum_{j=0}^{l(\gamma)-2}f(\sigma^j\gamma)\leq \frac{t(\gamma)}{2}<\sum_{j=0}^{l(\gamma)-1}f(\sigma^j\gamma).$$
In the following, the dependence of $l$ and $k$ in $\gamma$ will be often omitted for simplicity of notations and will be recalled only when it is necessary. This definition allows to write the quantity we want to bound
$$\left\|\left[P_{\gamma_{k}}(k\eta),P_{\gamma_{0}}\right]P_{\gamma_{k-1}}((k-1)\eta)\cdots P_{\gamma_{1}}(\eta)\psi_{\hbar}\right\|_{L^2}$$
in the following way:
\begin{equation}\label{pseudocomm10}\left\|\left[P_{\gamma_{k}}((k-l+1)\eta),P_{\gamma_{0}}((-l+1)\eta)\right]P_{\gamma_{k-1}}((k-l)\eta)\cdots P_{\gamma_{l}}(\eta) P_{\gamma_{l-1}}\cdots P_{\gamma_{1}}((-l+2)\eta)\psi_{\hbar}\right\|_{L^2}.
\end{equation}
The reason why we choose to write the quantity we want to bound in~(\ref{pseudocomm}) in the previous form instead of its original form is to have a more symmetric situation for our semiclassical analysis. To prove the bound in theorem~\ref{comm2}, a class of symbols taken from~\cite{DS} will be used (see~(\ref{symbol}) for a definition) and results about them are recalled in appendix~\ref{appendix}. Before starting the proof, using proposition~\ref{localization2}, we can restrict ourselves to observables carried on a thin energy strip around the energy layer $\mathcal{E}^{\theta}$. It means that the quantity we want to bound is the following norm:
\begin{equation}\label{pseudocomm2}\left\|\left[\hat{P}_{\gamma_{k}}((k-l+1)\eta),\hat{P}_{\gamma_{0}}((-l+1)\eta)\right]\hat{P}_{\gamma_{k-1}}((k-l)\eta)\cdots \hat{P}_{\gamma_{l}}(\eta)\hat{P}_{\gamma_{l-1}}\cdots \hat{P}_{\gamma_{k-1}}((-l+2)\eta)\psi_{\hbar}\right\|_{L^2},
\end{equation}
where $\hat{P}_i$ is now equal to $\Op_{\hbar}(P_i^f)$, where $P_i^f$ is compactly supported in $T^*\Omega_i\cap\mathcal{E}^{\theta}$ (see proposition~\ref{localization2}).

\subsubsection{Defining cutoffs}

If we consider quantity~(\ref{pseudocomm2}), we can see that because we consider large times $k\eta$, we can not estimate directly the norm of the bracket $\displaystyle\left[\hat{P}_{\gamma_{k}}((k-l+1)\eta),\hat{P}_{\gamma_{0}}((-l+1)\eta)\right]$ as there is no particular reason for $\hat{P}_{\gamma_{k}}((k-l+1)\eta)$ and $\hat{P}_{\gamma_{0}}((-l+1)\eta)$ to be pseudodifferential operators to which we can apply the classical rules from semiclassical analysis. However, the quantity we are really interested in is the norm of this bracket on the image of $\hat{P}_{\gamma_{k-1}}((k-l)\eta)\cdots \hat{P}_{\gamma_{l}}(\eta)\hat{P}_{\gamma_{l-1}}\cdots \hat{P}_{\gamma_{k-1}}((-l+2)\eta)$. So we will introduce some cutoff operators to localize the bracket we want to estimate on the image of $\hat{P}_{\gamma_{k-1}}((k-l)\eta)\cdots \hat{P}_{\gamma_{l}}(\eta)\hat{P}_{\gamma_{l-1}}\cdots \hat{P}_{\gamma_{k-1}}((-l+2)\eta)$. Then, as was discussed in section~\ref{heuristic}, we will have to verify that it defines a particular family of operators for which the Egorov theorem can be applied for large times.\\
First, we introduce a new family of functions $(Q_i)_{i=1}^K$ such that such that for each $1\leq i\leq K$, $Q_i$ belongs to $\mathcal{C}^{\infty}(T^*\Omega_i\cap\mathcal{E}^{\theta})$, $0\leq Q_i\leq 1$ and $Q_i\equiv 1$ on $\text{supp}P_i^f$. We then define two cutoffs associated to the strings $(\gamma_1,\cdots,\gamma_{l-1})$ and $(\gamma_{l},\cdots,\gamma_{k-1})$:
\begin{equation}\label{cutoff1}Q_{\gamma_{k-1},\cdots,\gamma_{l}}:=Q_{\gamma_{l}}\circ g^{-(k-l)\eta}\cdots Q_{\gamma_{k-1}}\circ g^{-\eta}\end{equation}
and
\begin{equation}\label{cutoff2}\tilde{Q}_{\gamma_{l-1},\cdots,\gamma_{1}}:=Q_{\gamma_{1}}\circ g^{\eta}\cdots Q_{\gamma_{l-1}}\circ g^{(l-1)\eta}.\end{equation}
The first point of our discussion will be to prove that Egorov theorem can be applied for large times to the pseudodifferential operators corresponding to these two symbols.\\
We prove the Egorov property for $Q_{\gamma_{k-1},\cdots,\gamma_{l}}$ for example (the proof works in the same way for the other one). Recall that one has the exact equality, for a symbol $a$:
\begin{equation}\label{egorov2}U^{-t}\Op_{\hbar}(a)U^t-\Op_{\hbar}(a(t))=\int_0^tU^{-s}(\text{Diff}a^{t-s})U^s ds,\end{equation}
where $a(t):=a\circ g^t$ and $\text{Diff}a^t:=\frac{\imath}{\hbar}[-\frac{\hbar^2\Delta}{2},\Op_{\hbar}(a(t))]-\Op_{\hbar}(\{H,a(t)\}).$ Here, we will consider $a:=Q_{\gamma_{k-1},\cdots,\gamma_{l}}$. One has, for $0\leq t\leq (k-l+1)\eta$:
$$Q_{\gamma_{k-1},\cdots,\gamma_{l}}(t):=Q_{\gamma_{k-1},\cdots,\gamma_{l}}\circ g^t=Q_{\gamma_{l}}\circ g^{-(k-l)\eta+t}\cdots Q_{\gamma_{k-1}}\circ g^{-\eta+t}.$$
There exists a unique integer $1\leq j\leq (k-l)$ such that $t-j\eta$ is negative and $t-(j-1)\eta$ is nonnegative. This allows us to rewrite:
$$Q_{\gamma_{k-1},\cdots,\gamma_{l}}(t)=\left(Q_{\gamma_{l}}\circ g^{-(k-l-j)\eta}\cdots Q_{\gamma_{k-j}}\right)\circ g^{-j\eta+t}\left(Q_{\gamma_{k-j+1}}\cdots Q_{\gamma_{k-1}}\circ g^{(j-2)\eta}\right)\circ g^{-(j-1)\eta+t}.$$
Using the last part of theorem~\ref{pdotheo} and its subsequent remark, we know that $Q_{\gamma_{l}}\circ g^{-(k-l-j)\eta}\cdots Q_{\gamma_j}$ and $Q_{\gamma_{j-1}}\cdots Q_{\gamma_{k-1}}\circ g^{(j-2)\eta}$ are symbols of the class $S_{\nu}^{-\infty,0}$ (see the appendix for a definition of this class of symbols), where $\nu:=\frac{1-\epsilon'+4\epsilon}{2}$. Moreover the constants in the bounds of the derivatives are uniform for the words $\gamma$ in the allowed set (see theorem~\ref{pdotheo} and proposition~\ref{symbolclass}). As $-\eta\leq t-j\eta<0\leq t-(j-1)\eta\leq \eta$ and as the class $S_{\nu}^{-\infty,0}$ is stable by product, we have then that $Q_{\gamma_1,\cdots,\gamma_{k-l}}(t)$ is in the class $S_{\nu}^{-\infty,0}$, for $0\leq t\leq (k-l+1)\eta$, with uniform bounds in $t$ and $\gamma$ in the allowed set. As, in~\cite{AN2}, we can verify that $\text{Diff}Q_{\gamma_{k-1},\cdots,\gamma_{l}}^t$ is in $\Psi^{-\infty,2\nu-1}_{\nu}$ and then apply the Calder\'on-Vaillancourt theorem for $\Psi^{-\infty,2\nu-1}_{\nu}$. As a consequence, there exists a constant $C$ depending only on the family $Q_i$ and on the derivatives of $g^s$ (for $-\eta\leq s\leq \eta$) such that
\begin{equation}\label{mainegorov}\forall0\leq t\leq (k-l+1)\eta,\ \ \ \|\Op_{\hbar}(Q_{\gamma_{k-1},\cdots,\gamma_{l}})(t)-\Op_{\hbar}(Q_{\gamma_{k-1},\cdots,\gamma_{l}}(t))\|_{\mathcal{L}(L^2(M))}\leq C\hbar^{1-2\nu}.\end{equation}
As we mentioned it in the heuristic of the proof (section~\ref{heuristic}), taking into account the support of the symbol, we have proved a 'local' Egorov property for a range of time that depends on the support of our symbol. Precisely, we have shown that the Egorov property holds until the stopping time defined in section~\ref{stoppingtime}.

\subsubsection{Proof of theorem~\ref{comm2}}

Before proving theorem~\ref{comm2}, we define (in order to have simpler expressions):
$$\psi^{\gamma}_{\hbar}:=\hat{P}_{\gamma_{k-1}}((k-l)\eta)\cdots \hat{P}_{\gamma_{1}}((-l+2)\eta)\psi_{\hbar}.$$
To prove theorem~\ref{comm2}, we need to bound quantity~(\ref{pseudocomm2}) and precisely to estimate~(\ref{pseudocomm2}), we have to estimate:
\begin{equation}\label{pseudocomm5}(\ref{pseudocomm2})=\left\|\left[\hat{P}_{\gamma_{0}}((-l+1)\eta),\hat{P}_{\gamma_{k}}((k-l+1)\eta)\right]\psi^{\gamma}_{\hbar}\right\|_{L^2(M)}.\end{equation}
Now we want to introduce our cutoff operators $\Op_{\hbar}(Q_{\bullet})$ in the previous expression:
$$\hat{P}_{\gamma_{0}}((-l+1)\eta)\hat{P}_{\gamma_{k}}(k-l+1)\eta)=\hat{P}_{\gamma_{0}}((-l+1)\eta)\left(Id-\left(\hat{P}_{\gamma_{k}}\Op_{\hbar}(Q_{\gamma_{k-1},\cdots,\gamma_{l}})\right)\left((k-l+1)\eta\right)\right)$$
$$+\left(\hat{P}_{\gamma_{k}}\Op_{\hbar}(Q_{\gamma_{k-1},\cdots,\gamma_{l}})\right)\left((k-l+1)\eta\right).$$
We will first estimate the norm $$\left\|\hat{P}_{\gamma_{0}}((-l+1)\eta)\left(Id-\left(\hat{P}_{\gamma_{k}}\Op_{\hbar}(Q_{\gamma_{k-1},\cdots,\gamma_{l}})\right)\left((k-l+1)\eta\right)\right)\psi^{\gamma}_{\hbar}\right\|_{L^2(M)}.$$ To do this, we first outline that $\hat{P}_{\gamma_k}$ is in $\Psi^{-\infty,0}(M)$ and $\Op_{\hbar}(Q_{\gamma_{k-1},\cdots,\gamma_{l}})$ is in $\Psi^{-\infty,0}_{\nu}(M)$. Using the standard rules for a product, we know that the previous expression can be transformed as follows:
$$\left\|\hat{P}_{\gamma_{0}}((-l+1)\eta)\left(Id-\Op_{\hbar}(P_{\gamma_k}^f Q_{\gamma_{k-1},\cdots,\gamma_{l}})((k-l+1)\eta)\right)\psi^{\gamma}_{\hbar}\right\|_{L^2(M)}+R_{\gamma}^{1}(\hbar),$$
where $\|R_{\gamma}^{1}(\hbar)\|_{L^2}\leq C\hbar^{1-2\nu}\|\psi^{\gamma}_{\hbar}\|_{L^2}$ (where $C$ is independent of $k-l$ as the bounds implied in the derivatives in theorem~\ref{pdotheo} are uniform for words $\gamma$ in the allowed set: see proposition~\ref{symbolclass}). We can apply the strategy of the previous section to prove an Egorov property for the operator $\Op_{\hbar}(P_{\gamma_{k}}^fQ_{\gamma_{k-1},\cdots,\gamma_{l}})$. So, up to a $\mathcal{O}_{L^2}(\hbar^{1-2\nu})$, $\Op_{\hbar}(P_{\gamma_{k}}^f Q_{\gamma_{k-1},\cdots,\gamma_{l}})((k-l+1)\eta)$ is equal to the pseudodifferential operator in $\Psi_{\nu}^{-\infty,0}$
$$\Op_{\hbar}\left((P_{\gamma_k}^fQ_{\gamma_{k-1},\cdots,\gamma_{l}})\circ g^{(k-l+1)\eta}\right)$$ supported in $g^{-\eta}T^*\Omega_{\gamma_{l}}\cap\cdots\cap g^{-(k-l+1)\eta}T^*\Omega_{\gamma_{k}}\cap\mathcal{E}^{\theta}$. Using then theorem~\ref{pdotheo}, the following holds:
$$\left(Id-\Op_{\hbar}\left((P_{\gamma_k}^fQ_{\gamma_{k-1},\cdots,\gamma_{l}})\circ g^{(k-l+1)\eta}\right)\right)\hat{P}_{\gamma_{k-1}}((k-l)\eta)\cdots \hat{P}_{\gamma_{1}}((-l+2)\eta)\psi_{\hbar}=\mathcal{O}_{L^2}(\hbar^{\infty}).$$
Even if the proof of this fact is rather technical, it is intuitively quite clear. In fact, if we suppose that the standard pseudodifferential rules (Egorov, composition) apply, $\hat{P}_{\gamma_{k-1}}((k-l)\eta)\cdots \hat{P}_{\gamma_{1}}((-l+2)\eta)$ is a pseudodifferential operator compactly supported in $g^{(l-2)\eta}T^*\Omega_{\gamma_1}\cap\cdots\cap g^{(l-k)\eta}T^*\Omega_{\gamma_{k-1}}\cap\mathcal{E}^{\theta}$. On this set, by definition of the cutoff operators ($Q_i\equiv 1$ on $\text{supp}(P_i)$), $(1-(P_{\gamma_k}^fQ_{\gamma_{k-1},\cdots,\gamma_{l}})\circ g^{(k-l+1)\eta})$ is equal to $0$. As a consequence, we consider the product of two pseudodifferential operators of disjoint supports: it is $\mathcal{O}_{L^2}(\hbar^{\infty})$. The statement of theorem~\ref{pdotheo} makes this argument work. To conclude the previous lines of the proof, we have
\begin{equation}\label{remainder5}\left\|\hat{P}_{\gamma_{0}}((-l+1)\eta)\left(Id-\left(\hat{P}_{\gamma_{k}}\Op_{\hbar}(Q_{\gamma_{k-1},\cdots,\gamma_{l}})\right)((k-l+1)\eta)\right)\psi^{\gamma}_{\hbar}\right\|_{L^2(M)}\leq \tilde{C}\hbar^{1-2\nu}\|\psi^{\gamma}_{\hbar}\|_{L^2(M)}.\end{equation}
Performing this procedure for the other operators, we finally obtain that the only quantity we need to bound to prove theorem~\ref{comm2} is the following quantity:
\begin{equation}\label{comm14}\left\|\left[(\hat{P}_{\gamma_{k}}\Op_{\hbar}(Q_{\gamma_{k-1},\cdots,\gamma_{l}}))((k-l+1)\eta),(\hat{P}_{\gamma_{0}}\Op_{\hbar}(\tilde{Q}_{\gamma_{l-1},\cdots,\gamma_{1}}))((-l+1)\eta)\right]\right\|_{\mathcal{L}(L^2(M))}.\end{equation}
Using the property of the product on $\Psi^{-\infty,0}_{\nu}$, we know that, up to a $\mathcal{O}_{L^2}(\hbar^{1-2\nu})$, the previous quantity is equal to
$$\left\|\left[\Op_{\hbar}(P_{\gamma_k}^fQ_{\gamma_{k-1},\cdots,\gamma_{l}})((k-l+1)\eta),\Op_{\hbar}(P_{\gamma_0}^f\tilde{Q}_{\gamma_{l-1},\cdots,\gamma_{1}})((-l+1)\eta)\right]\right\|_{\mathcal{L}(L^2(M))}.$$
Using the same method that in the previous section (which uses theorem~\ref{pdotheo}), we can prove an Egorov property for the two pseudodifferential operators that are in the previous bracket and show that, up to a $\mathcal{O}_{L^2}(\hbar^{1-2\nu})$, the quantity~(\ref{comm14}) is equal to
$$\left\|\left[\Op_{\hbar}((P_{\gamma_k}^fQ_{\gamma_{k-1},\cdots,\gamma_{l}})\circ g^{(k-l+1)\eta}),\Op_{\hbar}((P_{\gamma_0}^f\tilde{Q}_{\gamma_{l-1},\cdots,\gamma_{1}})\circ g^{(-l+1)\eta})\right]\right\|_{\mathcal{L}(L^2(M))}.$$
Using the pseudodifferential rules in $\Psi_{\nu}^{-\infty,0}(M)$ (proceeding as in the previous section, the two symbols stay in the good class of symbol using theorem~\ref{pdotheo}), we know that the previous bracket is in $\Psi_{\nu}^{-\infty,2\nu-1}$. Using the Calder\'on-Vaillancourt theorem, we know that quantity~(\ref{comm14}) is a $\mathcal{O}_{L^2}(\hbar^{1-2\nu}),$ where the constant depends only on the partition. This concludes the proof of theorem~\ref{comm2}.$\square$

\section{Products of many evolved pseudodifferential operators}
\label{bigpdotheo}

The goal of this section is to prove a property used in the proof of theorem~\ref{comm2}. Precisely, the following theorem states that the product of a large number of evolved pseudodifferential operators remains in a good class of pseudodifferential operators provided the range of times is smaller than the `local' Ehrenfest time. First, recall that using proposition~\ref{localization2}, we can restrict ourselves to observables carried on a thin energy strip around the energy layer $\mathcal{E}^{\theta}$. We underline that we do not suppose anymore that this thin energy strip is of size $\hbar^{1-\delta}$: we only need to have a small macroscopic neighborhood of the unit energy layer. Moreover, the class of symbols we will consider will be the class $S^{-\infty,0}_{\nu}$ (see~(\ref{symbol}) for a precise definition) with $\nu:=\frac{1-\epsilon'+4\epsilon}{2}$ ($<1/2$, see section~\ref{proof}).
\begin{theo}\label{pdotheo} Let $(Q_i)_{i=1}^K$ be a family of smooth functions on $T^*M$ such that for each $1\leq i\leq K$, $Q_i$ belongs to $\mathcal{C}^{\infty}(T^*\Omega_i\cap\mathcal{E}^{\theta})$ and $0\leq Q_i\leq 1$. Consider a family of indices $(\gamma_1,\cdots,\gamma_l)$ such that
$$\sum_{j=1}^{l-1}f(\gamma_{j+1},\gamma_{j})\leq\frac{n_E(\hbar)}{2}.$$
Then, for any $1\leq j\leq l$, one has
$$\Op_{\hbar}(Q_{\gamma_1})\Op_{\hbar}(Q_{\gamma_2})(-\eta)\cdots \Op_{\hbar}(Q_{\gamma_j})(-(j-1)\eta)=\Op_{\hbar}\left(A^{\gamma_1,\cdots,\gamma_j}\right)\left(-j\eta\right)+\mathcal{O}_{L^2}(\hbar^{\infty}),$$
where $A^{\gamma_1,\cdots,\gamma_j}$ is in the class $S^{-\infty,0}_{\nu}$. Precisely, one has the following asymptotic expansion:
$$A^{\gamma_1,\cdots,\gamma_j}\sim\sum_{p\geq 0}\hbar^p A^{\gamma_1,\cdots,\gamma_j}_p,$$
where $A^{\gamma_1,\cdots,\gamma_j}_p$ is in the class $S^{-\infty,2p\nu}_{\nu}$ (with the symbols semi norm uniform for $\gamma$ in the allowed set of sequences and $1\leq j\leq l$: see proposition~\ref{symbolclass}) and compactly supported in $g^{-\eta}T^{*}\Omega_{\gamma_{j}}\cap\cdots g^{-j\eta}T^{*}\Omega_{\gamma_{1}}\cap\mathcal{E}^{\theta}$. Finally the principal symbol $A^{\gamma_1,\cdots,\gamma_j}_0$ is given by the following formula:
$$A^{\gamma_1,\cdots,\gamma_j}_0=Q_{\gamma_j}\circ g^{\eta}\cdots Q_{\gamma_2}\circ g^{(j-1)\eta}Q_{\gamma_1}\circ g^{j\eta}.$$
\end{theo}
\begin{rema} We underline that the asymptotic expansion (except for the order $0$ term) is not intrinsically defined as it depends on the choice of coordinates on $M$. We also remark that this theorem holds in particular for the smooth partition of identity we considered previously on the paper. Note also that the the result can be rephrased by saying that $\Op_{\hbar}(Q_{\gamma_1})(j\eta)\Op_{\hbar}(Q_{\gamma_2})((j-1)\eta)\cdots \Op_{\hbar}(Q_{\gamma_j})(\eta)$ is, up to a $\mathcal{O}_{L^2}(\hbar^{\infty})$, a pseudodifferential operator of the class $\Psi_{\nu}^{-\infty,0}$ and of well determined support.
As we also have to consider `past' evolution, we mention that we can also suppose $\sum_{j=1}^{l-1}f(\gamma_{j},\gamma_{j+1})\leq\frac{n_E(\hbar)}{2}$. Under this assumption, we would have proved that $\Op_{\hbar}(Q_{\gamma_1})(-j\eta)\Op_{\hbar}(Q_{\gamma_2})(-(j-1)\eta)\cdots \Op_{\hbar}(Q_{\gamma_j})(-\eta)$ is, up to a $\mathcal{O}_{L^2}(\hbar^{\infty})$, a pseudodifferential operator of the class $\Psi_{\nu}^{-\infty,0}$ and of well determined support. These are exactly the properties we used in section~\ref{mainproof}.\end{rema}
The plan of the proof is the following. First, we will construct formally $A^{\gamma_1,\cdots,\gamma_j}$ and its asymptotic expansion in powers of $\hbar$. Then, we will check that these different symbols are in a good class. Finally, we will check that these operators approximate the product we considered. For simplicity of notations, we will forget (for a time) the dependence on $\gamma$ and denote the previous symbol $A^j$ for $l\geq j\geq1$.

\subsection{Definition of $A^{\gamma_1,\cdots,\gamma_l}$}

In this section, we construct formally the symbol $A^j$. The way to do it is by induction on $j$. First, we will see how to define formally $A^j$ from $A^{j-1}$. Then, using the formulas of the previous section, we will construct the formal order $N$ expansion associated to this $A^j$. We only construct what the order $N$ expansion should be regarding to the formal formulas.

\subsubsection{Definition at each step}

To construct $A^j$, we proceed by induction and at the first step, we consider $\Op_{\hbar}(Q_{\gamma_{1}})$ and we write it into the form $Op_{\hbar}(A^1)(-\eta)$. This means that we have defined formally for $0\leq t\leq\eta$:
$$\Op_{\hbar}(A^1(t)):=U^{-t}Op_{\hbar}(Q_{\gamma_{1}})U^{t}.$$
Using Egorov theorem for fixed time $\eta$ and the corresponding asymptotic expansion (see section~\ref{fixedegorov} for explicit formulas of the asymptotic expansion), we prove that, up to a $\mathcal{O}_{L^2}(\hbar^{\infty})$, $\Op_{\hbar}(Q_{\gamma_{1}})$ is equal to $\Op_{\hbar}(A^1(\eta))(-\eta)$, where $A^1(\eta)$ is in $S^{-\infty,0}$, given by the asymptotic expansion of the Egorov theorem and supported in $g^{-\eta}T^{*}\Omega_{\gamma_{1}}\cap\mathcal{E}^{\theta}$. We can continue this procedure formally. At the second step, we have $$\Op_{\hbar}(Q_{\gamma_{1}})\Op_{\hbar}(Q_{\gamma_{2}})(-\eta)=U^{\eta}\Op_{\hbar}(A^1(\eta))\Op_{\hbar}(Q_{\gamma_{2}})U^{-\eta}.$$
We want this quantity to be of the form $\Op_{\hbar}(A^2(\eta))(-2\eta)$. This means that we have defined formally for $0\leq t\leq\eta$:
$$\Op_{\hbar}(A^2(t)):=U^{-t}\Op_{\hbar}(A^1(\eta))\Op_{\hbar}(Q_{\gamma_{2}})U^{t}.$$
Using rules of pseudodifferential operators (see section~\ref{compositionpdo} and~\ref{fixedegorov}), we can obtain a formal asymptotic expansion for $A^2(\eta)$ (see next section) starting from the expansion of $A^1(\eta)$. One can easily check that this formal expansion is supported in $g^{-\eta}T^{*}\Omega_{\gamma_{2}}\cap g^{-2\eta}T^{*}\Omega_{\gamma_{1}}\cap\mathcal{E}^{\theta}$. Following the previous method, we will construct a formal expansion of $A^{j}(t)$ (for $0\leq t\leq\eta$) starting from the expansion of $A^{j-1}(\eta)$ (see next section). To do this, we will write at each step $1\leq j\leq l$, \begin{equation}\label{inductionformula}\Op_{\hbar}(A^j(t)):=U^{-t}\Op_{\hbar}(A^{j-1}(\eta))\Op_{\hbar}(Q_{\gamma_{j}})U^{t}.\end{equation}
We also introduce the intermediate operator
\begin{equation}
\Op_{\hbar}(\overline{A}^j):=\Op_{\hbar}(A^{j-1}(\eta))\Op_{\hbar}(Q_{\gamma_{j}}).
\end{equation}
With this definition, we will have
$$\Op_{\hbar}(A^j(\eta))(-j\eta):=\left(\Op_{\hbar}(A^{j-1}(\eta))\Op_{\hbar}(Q_{\gamma_{j}})\right)\left(-(j-1)\eta\right).$$
Using again rules of pseudodifferential calculus (see section~\ref{compositionpdo} and~\ref{fixedegorov}), we can obtain a formal asymptotic expansion for $A^j(t)$ (see next section) starting from the expansion of $A^{j-1}(\eta)$. One can easily check that this formal expansion is supported in $\displaystyle g^{-t}\left(T^{*}\Omega_{\gamma_{j}}\cap\cdots\cap g^{-j\eta}T^{*}\Omega_{\gamma_{1}}\right)\cap\mathcal{E}^{\theta}$.\\
In the next section, we will use the induction formula~(\ref{inductionformula}) to deduce the $\hbar$-expansion of $A^{j}(t)$ from the expansion for the composition of $\Op_{\hbar}(A^{j-1}(\eta))$ and $\Op_{\hbar}(Q_{\gamma_{j}})$ and from the expansion for the Egorov theorem for times $0\leq t\leq\eta$. At each step $1\leq j\leq l$ of the induction, we will have to prove that $A^j$ stays in a good class of symbols to be able to continue the induction.

\subsubsection{Definition of the order $N$ expansion}

We fix a large integer $N$ (to be determined). We study the previous construction by induction up to $\mathcal{O}(\hbar^N)$. From this point, we truncate $A^j(t)$ at the order $N$ of its expansion. First, we see how we construct the symbols $A^j(t)$ by induction. To do this, we use the formulas for the asymptotic expansions for the composition of pseudodifferential operators and for the Egorov theorem (see section~\ref{compositionpdo} and~\ref{fixedegorov}). Suppose that $$A^{j-1}(\eta)=\sum_{p=0}^{N}\hbar^pA^{j-1}_p(\eta)$$ is well defined, we have to define the expansion of $A^j(t)$ from the asymptotic expansion of $A^{j-1}(\eta)$, for $0\leq t\leq\eta$. First, we define:
\begin{equation}\label{composedsymbol}\overline{A}^j:=\sum_{p=0}^{N}\hbar^p\overline{A}_p^j,\ \ \ \ \
\text{where}\ \overline{A}_p^j:=\sum_{r=0}^{p}\left(A^{j-1}_{p-r}(\eta)\sharp_M Q_{\gamma_j}\right)_{r}.\end{equation}
The symbol $\sharp_M$ represents an analogue on a manifold of the Moyal product (see appendix~\ref{compositionpdo}): $(a\sharp_M b)_p$ is the order $p$ term in the expansion of the symbol of $\Op_{\hbar}(a)\Op_{\hbar}(b)$. Recall from the appendix that $(A^{j-1}_{p-q}\sharp_M Q_{\gamma_j})_{q}$ is a linear combination (that depends on the local coordinates and on the $(Q_i)_i$) of the derivatives of order less than $q$ of $A^{j-1}_{p-q}(\eta)$. Using proposition~\ref{exactegorov} in appendix~\ref{fixedegorov}, one has the following order $N-p$ expansion, for the symbol of the operator $U^{-t}\Op_{\hbar}(\overline{A}^j_p)U^t$,
$$\overline{A}^j_p:=\sum_{k=0}^{N-p}\hbar^k\overline{A}^j_{p,k}(t),$$
where $\overline{A}^j_{p,0}=\overline{A}^j_p\circ g^t$ and $\overline{A}^j_{p,k}(t):=\sum_{l=0}^{k-1}\int_0^t\left\{H,\overline{A}^j_{p,l}(t-s)\right\}_{M}^{(k,l)}\left(g^{s}(\rho)\right)ds.$ Then, we can define $A^j(t)$ using these different expansions. Precisley, we define
$$A^j(t):=\sum_{p=0}^N\hbar^pA_p^j(t)\ \text{where, for}\ 0\leq p\leq N,\ A_p^j(t):=\sum_{q=0}^p\overline{A}^j_{p-q,q}(t).$$
This construction is the precise way we want to define the asymptotic expansion of the symbol $A^j(t)$ in theorem~\ref{pdotheo}. If we want the theorem to be valid, we have to check now that the remainders we forget at each step are negligible (with an arbitrary high order in $\hbar$). To do this, we will first have to control at each step $j$ the derivatives of $A^j(t)$ (see next section).
\begin{rema} The support of $A^j_p(t)$ is included in $g^{-t}\left(T^*\Omega_{\gamma_j}\cap\cdots\cap g^{-(j-1)\eta}T^*\Omega_{\gamma_1}\right)\cap\mathcal{E}^{\theta}$ as the support of every $\overline{A}^j_{p,k}(t)$ is.\end{rema}
Finally, we underline that, according to our construction, $A^j_p(t)$ can be written as follows:
\begin{equation}\label{dep-lineaire}
A^j_p(t):=\left(A^{j-1}_p(\eta)Q_{\gamma_j}\right)\circ g^t+\sum_{r=1}^p(A_{p-r}^{j-1}\sharp_M Q_{\gamma_j})_r\circ g^t+\sum_{q=1}^p\sum_{l=0}^{q-1}\int_0^t\left\{H,\overline{A}^j_{p-q,l}(t-s)\right\}_{M}^{(q,l)}\left(g^{s}(\rho)\right)ds.
\end{equation}
For the following, we need to know precisely on how many derivatives of $A^{j-1}$ depends $A^j$. We analyse the three terms of the previous sum separately:
\begin{itemize}
\item the first term is explicit and it depends linearly on $A^{j-1}_p$;
\item according to appendix~\ref{compositionpdo}, the second term depends linearly on $(\partial^{\alpha}A^{j-1}_{p-r})_{1\leq r\leq p, |\alpha|\leq r}$;
\item according to corollary~\ref{number-derivatives}, the third term depends linearly on $(\partial^{\alpha}\overline{A}^{j}_{p-q})_{1\leq q\leq p, |\alpha|\leq 2q}$ and consequently, according to appendix~\ref{compositionpdo}, it depends linearly on $(\partial^{\alpha}A^{j-1}_{p-r})_{1\leq r\leq p, |\alpha|\leq 2r}$.
\end{itemize}

\subsection{Estimates of the derivatives}

The goal of the first part of this section is to prove the following lemma.
\begin{lemm}\label{symbolclass2} Let $N$ be a fixed integer. Fix also two integers $0\leq p\leq N$ and $m\leq 2(N-p+1)$. Then, there exists a constant $C(m,p)$ such that for all $j\geq 1$ and for all $\rho$ in the set $g^{-t}\left(T^*\Omega_{\gamma_j}\cap\cdots\cap g^{-(j-1)\eta}T^*\Omega_{\gamma_1}\right)\cap\mathcal{E}^{\theta}$,
$$\forall 0\leq t\leq\eta,\ |d^mA^j_p(t,\rho)|\leq C(m,p)j^{m+2p^2+1}|d_{\rho}g^{t+(j-1)\eta}|^{m+2p}.$$
If $\rho$ is not in this set, the bound is trivially $0$ by construction. Here the constant $C(m,p)$ depends only on $m$, $p$, the atlas we chose for the manifold and the size of the $(\Omega_{\gamma})_{\gamma}$.
\end{lemm}
Once this lemma will be proved, we will check that it also tells us that the $A_p^j$'s are in a nice class of symbols.

\subsubsection{Proof of lemma~\ref{symbolclass2}}

To make all the previous pseudodifferential arguments work, we will have to obtain estimates on the $m$-differential forms $d^mA_p^j$, for each $m\leq2(N+1-p)$. If we have estimates on these derivatives, we will then check that all the asymptotic expansions given by the pseudodifferential theory are valid. To do these estimates, we will have to understand the number of derivatives that appear when we repeat the induction formula~(\ref{inductionformula}). The spirit of this proof is the same as in~\cite{AN2} (section $3.4$) when they iterate the WKB expansion $\mathcal{K}|\log\hbar|$ times. We define a vector $\mathbf{A}^j$ with entries $\mathbf{A}^j_{(p,m)}(t,\rho):=d^m_{\rho}A_p^j(t)$ (where $0\leq p\leq N$ and $0\leq m\leq 2(N-p+1)$). Precisely, we order it by the following way, for $0\leq t\leq \eta$ and $\rho\in T^*M$,
$$\mathbf{A}^j=\mathbf{A}^j(t,\rho):=\begin{cases}(A_0^j,dA_0^j,\cdots,d^{2(N+1)}A_0^j,\\
A_1^j,dA_1^j,\cdots, d^{2N}A_1^j,\\
\cdots,\\
A_{N}^j,dA_{N}^j,d^{2}A_{N}^j).
\end{cases}$$
The induction formula~(\ref{dep-lineaire}) of the previous section can be rewritten under the following form
\begin{equation}\label{ind-formula-lemmePDO}A^j_p(t,\rho)=\left(A_p^{j-1}(\eta)Q_{\gamma_j}\right)\circ g^t(\rho)+L^j(t)(\mathbf{A}^{j-1}(\eta))(\rho),\end{equation}
where $L^j(t)$ acts linearly on $\mathbf{A}^{j-1}_{(p-q,m)}(\eta)$, where $q\geq 1$ and $m\leq 2q$. We underline that this linear application depends on derivatives of $g^s$ for $0\leq s\leq \eta$, on the choice of the coordinates and on the maps $Q_j$. We would also like to have an expression for $d^m_{\rho}A_p^j(t)$ for $m\leq 2(N+1-p)$. To do this, we start by writing that for an observable $a$, one has
$$d^m_{\rho}(a\circ g^t):=\sum_{l\leq m}d^l_{g^t\rho}a.\theta_{m,l}(t,\rho),$$
where $\theta_{m,l}(t,\rho)$ sends $(T_{\rho}T^*M)^m$ on $(T_{g^t\rho}T^*M)^l$. We can write the explicit form of $\theta_{m,m}$:
$$\theta_{m,m}(t,\rho):=\left(d_{\rho}g^t\right)^{\otimes m}.$$
Using these relations, we can rewrite the induction formula~(\ref{ind-formula-lemmePDO}) as follows:
$$\mathbf{A}^j(t)=(\mathbf{M}^j_{0}(t)+\mathbf{M}^j_{1}(t)+\mathbf{M}^j_{2}(t))\mathbf{A}^{j-1}(\eta),$$
where an exact expression of $\mathbf{M}_{0}^j$ is given by
$$\left(\mathbf{M}_{0}^j\mathbf{A}^{j-1}\right)_{(p,m)}(t,\rho):=Q_{\gamma_j}(g^t\rho)\times\mathbf{A}^{j-1}_{(p,m)}(\eta,g^t\rho).\theta_{m,m}(t,\rho).$$
In particular, $\mathbf{M}_{0}^j$ is a diagonal matrix. We will not give explicit expression for the two other matrices. We only need to know that the matrix $\mathbf{M}^j_{1}(t)$ relates $\mathbf{A}^j_{p,m}(t)$ to $(\mathbf{A}^{j-1}_{p,l}(\eta))_{l<m}$ and that the matrix $\mathbf{M}^j_{2}(t)$ relates $\mathbf{A}^j_{p}(t)$ to $(\mathbf{A}^{j-1}_{q}(\eta))_{q<p}$. Iterating the induction formula, one then has:
$$\mathbf{A}^{j}(t):=\sum_{\epsilon_2,\cdots,\epsilon_j=0}^{2}\mathbf{M}_{\epsilon_j}^{j}(t)\mathbf{M}_{\epsilon_{j-1}}^{j-1}(\eta)\cdots\mathbf{M}_{\epsilon_2}^{2}(\eta)\mathbf{A}^1(\eta).$$
From this expression, one can estimate how many terms contributes to the definition of $\mathbf{A}^j_{(p,m)}$. For instance, suppose that $\displaystyle\left|\{j':\epsilon_{j'}=2\}\right|> p$, the contribution of such a string of matrices to $\mathbf{A}^j_{(p,m)}$ is $0$ (using the nilpotence property). We can also give an upper bound on the number of terms of type $\mathbf{M}^*_{1}$ in string of matrices that contributes to $\mathbf{A}^j_{(p,m)}$. To do this, we underline that the action of a matrix of the type $\mathbf{M}^*_{1}$ add a block equal to $0$ at the begining of every $\mathbf{A}_p^*$ (as it is nilpotent). In particular, consider a given block $\mathbf{A}_p^*$ of the form $(0,\cdots,0,*)$ (where the $(l_0-2p'p)$ first terms are equal to $0$). After the action of a series of $\mathbf{M}^*_{1}$ (say $l_1$) and of $\mathbf{M}^*_{0}$ (in any order), we get a $p$-block of the form $(0,\cdots,0,*)$, where the $(l_0+l_1-2p'p)$ first terms are equal to $0$. On the other hand, we know that, if $\mathbf{A}^j:=\mathbf{M}_2^{j}(\eta)\mathbf{A}^{j-1}$, then the term of order $(p,m)$ depends only on $(\mathbf{A}^{j-1}_{q,r})_{q\leq p-1, r\leq2(p-q)+m}$. So after the action of a matrix $\mathbf{M}^*_{2}$, the $p$-block is still of the form $(0,\cdots,0,*)$, where now only the $(l_0+l_1-2(p'+1)p)$ first terms are equal to $0$. By an immediate induction, we find that the contribution of a string of matrices to $\mathbf{A}^j_{p,m}$ is $0$ if $\displaystyle\left|\{j':\epsilon_{j'}=1\}\right|-2p\left|\{j':\epsilon_{j'}=2\}\right|> m$.\\
As a conclusion, the product of matrices that contributes to the expression of $\mathbf{A}^j_{(p,m)}$ can only be non zero if $\displaystyle\left|\{j':\epsilon_{j'}=2\}\right|\leq p$ and $\displaystyle\left|\{j':\epsilon_{j'}=1\}\right|\leq m+2p\left|\{j':\epsilon_{j'}=2\}\right|$. As a consequence, for large $j$, to be non zero, a string of matrices need to be made of at most $(N+1)^2$ matrices of the form $\mathbf{M}^j_{\epsilon}$ (for $\epsilon\in\{1,2\}$). Finally, we need to compute the number of string of matrices that contributes to a given $\mathbf{A}^j_{(p,m)}$. To do this, we consider the set of symbols $\{(\epsilon_1,\cdots,\epsilon_k):k\leq m+2p^2,\ \epsilon_j\in\{1,2\}\}.$ For a given $(\epsilon_1,\cdots,\epsilon_k)$ in this set, the number of ways of putting these symbols in a string of length $j$ is bounded by $j^k$. Moreover, we know that there are at most $2^k$ sequences of length $k$. These two remarks implies that the number of string of matrices that contributes to a given $\mathbf{A}^j_{(p,m)}$ is bounded by $\sum_{k=0}^{m+2p^2}(2j)^k$, which is a $\mathcal{O}((2j)^{m+2p^2+1})$.\\
Then, to estimate the norm of the derivatives of $A^j$, we should look how the different matrices act. First we study the action of the diagonal matrix. As $0\leq Q_{\gamma_j}\leq 1$, one has that, for $0\leq t\leq \eta$ and for any $\rho\in g^{-t}\left(T^*\Omega_{\gamma_j}\cap\cdots\cap g^{-(j-1)\eta}T^*\Omega_{\gamma_1}\right)\cap\mathcal{E}^{\theta}$ (otherwise the following quantity is clearly equal to $0$),
$$|\mathbf{M}_{0}^j\mathbf{A}^{j-1}_{(p,m)}(t,\rho)|\leq |d_{\rho}g^{t}|^m|\mathbf{A}^{j-1}_{(p,m)}(\eta,g^t(\rho))|.$$
We note that we can iterate this bound and find, for any $j$ and $j'$ in $\mathbb{N}$, we have, for any $0\leq t\leq\eta$,
$$|\mathbf{M}_{0}^{j+j'}\cdots\mathbf{M}_{0}^j\mathbf{A}^{j-1}_{(p,m)}(t,\rho)|\leq |d_{\rho}g^{t+j'\eta}|^m|\mathbf{A}^{j-1}_{(p,m)}(\eta,g^{t+j'\eta}(\rho))|.$$
Now, using the fact that for every iteration, we consider a fixed interval of time $[0,\eta]$ and the fact that the set of observables $(Q_l)_{l=1}^K$ is fixed, we get that there exists a constant $C(m,p)$ such that, for $\epsilon\in\{1,2\}$,
$$\sup_{0\leq t\leq \eta}\|\mathbf{M}_{\epsilon}^j\mathbf{A}^{j-1}_{(p,m)}(t)\|_{L^{\infty}}\leq C(m,p)\max_{m'\leq m}\max_{q\leq p}\|\mathbf{A}^{j-1}_{(q,m')}\|_{L^{\infty}}.$$
The only thing we need to know is that the constant depends only on $m$, $p$, the manifold, $\eta$, the coordinate maps and the partition. The difference with the action of the diagonal matrix is that we have constant prefactor that can accumulate and become large (without any precise control on it).\\
These different observations allow us to prove lemma~\ref{symbolclass2}. In fact, by construction, the total number of derivatives of $g^t$ that appears in the definition of $\mathbf{A}^{j}_{(p,m)}(t)$ is bounded by $m+2p$. Moreover, a given string $\mathbf{M}_{\epsilon_j}^{j}(t)\mathbf{M}_{\epsilon_{j-1}}^{j-1}(\eta)\cdots\mathbf{M}_{\epsilon_2}^{2}(\eta)$ is made of long string only made of matrices of the form $\mathbf{M}_{0}^{*}(\eta)$ and of short strings of matrices of the form  $\mathbf{M}_{\epsilon}^{*}(\eta)$ (where $\epsilon\in\{1,2\}$). We know that only the long strings made of $\mathbf{M}_{0}^{*}(\eta)$ will contribute to a given $\mathbf{A}^{j}_{(p,m)}(t)$ and as we know that the number of derivatives involved is bounded by $m+2p$, we have, for any $\rho\in g^{-t}\left(T^*\Omega_{\gamma_j}\cap\cdots\cap g^{-(j-1)\eta}T^*\Omega_{\gamma_1}\right)\cap\mathcal{E}^{\theta}$,
$$\left|\left(\mathbf{M}_{\epsilon_j}^{j}(t)\mathbf{M}_{\epsilon_{j-1}}^{j-1}(\eta)\cdots\mathbf{M}_{\epsilon_2}^{2}(\eta)\mathbf{A}^1(\eta)\right)_{(p,m)}(\rho)\right|\leq C'(p,m)|d_{\rho}g^{t+(j-1)\eta}|^{m+2p}\|\mathbf{A}^1(\eta)\|.$$
Finally, the number of matrices that contributes to the $(p,m)$-term of the vector $\mathbf{A}^j$ is bounded by $\mathcal{O}((2j)^{m+2p^2+1})$. It gives that, for any $\rho\in g^{-t}\left(T^*\Omega_{\gamma_j}\cap\cdots\cap g^{-(j-1)\eta}T^*\Omega_{\gamma_1}\right)\cap\mathcal{E}^{\theta}$,
$$|\mathbf{A}^{j}_{(p,m)}(t,\rho)|\leq \tilde{C}(p,m)j^{m+2p^2+1}|d_{\rho}g^{t+(j-1)\eta}|^{m+2p}\|\mathbf{A}^1(\eta)\|.\square$$

\subsubsection{Class of symbol of each term of the expansion}
\label{symbolclasspart}
Using the previous lemma, we want to show that $A^j_p(t)$ is an element of $S^{-\infty,2p\nu}_{\nu}$. Let $\rho$ be an element of $\displaystyle g^{-t}\left(T^*\Omega_{\gamma_j}\cap\cdots\cap g^{-(j-1)\eta}T^*\Omega_{\gamma_1}\right)\cap\mathcal{E}^{\theta}$. Using the fact that $E^u$ is of dimension $1$, we get that for any positive $t$, $|d_{\rho}g^{t}|\leq J^{u,t}(\rho)^{-1}$, where $\displaystyle J^{u,t}(\rho):=\det\left(dg^{-t}_{|E^{u}(g^t\rho)}\right).$ Then we can write the multiplicativity of the determinant and get
$$J^{u,t+(j-1)\eta}(\rho)=J^{u,t}(\rho)J^{u,\eta}(g^{t}\rho)J^{u,\eta}(g^{t+\eta}\rho)\cdots J^{u,\eta}(g^{t+(j-2)\eta}\rho).$$
\begin{rema} Before continuing the estimate, let us underline some property of the Jacobian. Suppose $S$ is a positive integer and $1/\eta$ also (large enough to be in our setting). We have, for all $0\leq k\leq 1/\eta-1$,
$$J^u(g^{k\eta}\rho)J^u(g^{1+k\eta}\rho)\cdots J^u(g^{S-1+k\eta}\rho)=J^{u,\eta}(g^{k\eta}\rho)J^{u,\eta}(g^{(k+1)\eta}\rho)\cdots J^{u,\eta}(g^{S+(k-1)\eta}\rho),$$
where $J^u(\rho)$ is the unstable Jacobian in time $1$ that appears in the main theorem~\ref{maintheo}. We make the product over $k$ of all these equalities and we get
$$J^{u}(\rho)^{\eta}J^u(g^{\eta}\rho)^{\eta}\cdots J^{u}(g^{S-\eta}\rho)^{\eta}\leq C(\eta) J^{u,\eta}(g\rho)J^{u,\eta}(g^{1+\eta}\rho)\cdots J^{u,\eta}(g^{S-\eta}\rho),$$
where $C(\eta)$ only depends on $\eta$ and does not depend on $S$.\end{rema}
Finally, using previous remark and inequality~(\ref{continuity}), the following estimate holds, for $\rho$ in  $g^{-t}\left(T^*\Omega_{\gamma_j}\cap\cdots\cap g^{-(j-1)\eta}T^*\Omega_{\gamma_1}\right)\cap\mathcal{E}^{\theta}$:
$$|d_{\rho}g^{t+(j-1)\eta}|\leq C(\eta) e^{j\epsilon\eta a_0} J^{u}_{\eta}(\gamma_{j},\gamma_{j-1})^{-\eta}J^{u}_{\eta}(\gamma_{j-1},\gamma_{j-2})^{-\eta}\cdots J^{u}_{\eta}(\gamma_{2},\gamma_{1})^{-\eta}$$
with $C(\eta)$ independent of $j$. Then, one has
$$|d_{\rho}g^{t+(j-1)\eta}|\leq C(\eta) e^{l(\gamma)\epsilon\eta a_0}e^{t(\gamma)},$$
where $t(\gamma)=\sum_{j=0}^{l-1}f(\gamma_{j+1},\gamma_j)$. As $t(\gamma)\leq n_{E}(\hbar)/2$, this last quantity is bounded by $\hbar^{\frac{\epsilon'-1}{2}-\epsilon}$ (as $l(\gamma)a_0\eta\leq n_E(\hbar)/2$). Using lemma~\ref{symbolclass2}, we want to estimate the $m$ derivatives of the symbol $A^j_p$. According to the previous paragraph, they can be estimated up to order $2(N+1-p)$. To get a control on an arbitrary order $m$, we can fix a large $\tilde{N}$ such that $m\leq2(\tilde{N}-N)$ and use the result of the previous section for this $\tilde{N}$. Finally, we have, for $p<N$, $m\in\mathbb{N}$ and $0\leq t\leq\eta$,
\begin{equation}\label{estimder2}
|d^mA_p^j(t,x)|\leq \tilde{C}(m,p)\hbar^{(m+2p)(\frac{\epsilon'-1}{2}-2\epsilon)}.\end{equation}
Here appears the fact that we only apply the backward quantum evolution for times $l$ (we also used the fact that $j=\mathcal{O}(|\log\hbar|)$). In fact, as we want our symbols to be in the class $S^{-\infty,.}_{\nu}$, we need derivatives to lose at most a factor $\hbar^{-1/2}$ (this would have not been the case if we had considered times of size $n_{E}(\hbar)$ instead of size $n_E(\hbar)/2$). The previous estimate~(\ref{estimder2}) is uniform for all the $\gamma$ in the allowed set of theorem~\ref{pdotheo}.\\
Finally, to summarize this section, we can write the following proposition:
\begin{prop}\label{symbolclass} Let $p$ and $m$ be elements of $\mathbb{N}$. There exists $C(m,p,(Q_i)_i,\eta)$ (depending on $m$, $p$, $\eta$, $(Q_i)_{i=1}^K$ and the coordinate charts) such that for all $\gamma=(\gamma_0,\cdots,\gamma_l)$ such that
$$\sum_{j=0}^{l-1}f(\gamma_{j+1},\gamma_j)\leq\frac{n_E(\hbar)}{2},$$
for all $0\leq j\leq l$ and for all $0\leq t\leq \eta$,
$$|d^mA_p^{\gamma_1,\cdots,\gamma_j}(t,x)|\leq C(m,p,(Q_i)_i,\eta)\hbar^{(m+2p)(\frac{\epsilon'-1}{2}-2\epsilon)}.$$
Then, as the $A_p^j$ are compactly supported, $A_p^j$ is in class $S^{-\infty,2p\nu}_{\nu}$, where $\nu=\frac{1-\epsilon'+4\epsilon}{2}$.
\end{prop}
So, our formal construction allows us to define a family of symbol $A^j_p$ and each of them belongs to $S^{-\infty,p\epsilon}_{\nu}$. Moreover the constants implied in the bounds of the derivatives are uniform with respect to the allowed sequences. We underline that the same proof would show that the intermediate symbols $\overline{A}^j_p$~(\ref{composedsymbol}) are also in the same class of symbols.

\subsection{Estimate of the remainder terms}

We are now able to conclude the proof of theorem~\ref{pdotheo} starting from the family we have just constructed. We have to verify that the remainder is of small order in $\hbar$. Fix a large integer $N$ and denote $\displaystyle A^j(\eta):=\sum_{p=0}^{N}\hbar^pA_p^j(\eta)$. We want to estimate
$$R_N^j=\|\Op_{\hbar}(Q_{\gamma_1})\cdots \Op_{\hbar}(Q_{\gamma_j})(-(j-1)\eta)-\Op_{\hbar}(A^j(\eta))(-j\eta)\|_{\mathcal{L}(L^2(M))}.$$
Using the induction formula~(\ref{inductionformula}), we write
$$R_N^j\leq\|U^{-\eta}\Op_{\hbar}(A^{j-1}(\eta))\Op_{\hbar}(Q_{\gamma_j})U^{\eta}-\Op_{\hbar}(A^j(\eta))\|_{\mathcal{L}(L^2(M))}+R_N^{j-1},$$
where $R_N^{j-1}=\|\Op_{\hbar}(Q_{\gamma_1})\cdots \Op_{\hbar}(Q_{\gamma_{j-1}})(-(j-2)\eta)-\Op_{\hbar}(A^{j-1}(\eta))\|_{\mathcal{L}(L^2(M))}$.  We start by giving an estimate on the first term of the previous upper bound. To do this, we first give a bound on
$$R^{\text{comp},j}_N:=\|\Op_{\hbar}(A^{j-1}(\eta))\Op_{\hbar}(Q_{\gamma_j})-\Op_{\hbar}(\overline{A}^{j})\|_{\mathcal{L}(L^2(M))}.$$
Using the expansion of $A^{j-1}(\eta)$ and $\overline{A}^j$, this can rewritten
$$R^{\text{comp},j}_N\leq\sum_{p=0}^N\hbar^p\left\|\Op_{\hbar}(A^{j-1}_p(\eta))\Op_{\hbar}(Q_{\gamma_j})-\sum_{r=0}^{N-p}\hbar^{r}\Op_{\hbar}((A^{j-1}_p\sharp_M Q_{\gamma_j})_r)\right\|_{\mathcal{L}(L^2(M))}.$$
Then, we can use section~\ref{compositionpdo} and the estimates~(\ref{remmoyal}), to bound each term of the previous sum as follows:
$$\left\|\Op_{\hbar}(A^{j-1}_p(\eta))\Op_{\hbar}(Q_{\gamma_j})-\sum_{r=0}^{N-p}\hbar^{r}\Op_{\hbar}((A^{j-1}_p\sharp_M Q_{\gamma_j})_r)\right\|_{\mathcal{L}(L^2(M))}\leq C_{N,p}\hbar^{(N+1-p)(1-\nu)-2p\nu-(C+C')\nu}.$$
In particular, we find that $R^{\text{comp},j}_N=\mathcal{O}_N(\hbar^{(N+1)(1-2\nu)-(C+C')\nu})$ (as $\nu<1/2$). We have now to give a bound on $R^{\text{Egorov},j}_N:=\|\Op_{\hbar}(A^{j}(\eta))-U^{-\eta}\Op_{\hbar}(\overline{A}^{j})U^{\eta}\|_{\mathcal{L}(L^2(M))}.$
We will now use results on Egorov theorem from section~\ref{fixedegorov} to get this bound. First, we write the expansion of $\overline{A}^j$ to get
$$R^{\text{Egorov},j}_N\leq \sum_{p=0}^N\hbar^p\left\|U^{-\eta}\Op_{\hbar}(\overline{A}^{j}_p)U^{\eta}-\sum_{r=0}^{N-p}\hbar^r\Op_{\hbar}(\overline{A}^j_{p,r}(\eta))\right\|_{\mathcal{L}(L^2(M))}.$$
According to the rules for Egorov expansion from section~\ref{fixedegorov} (see estimates~(\ref{remainderestimates})) and as we know the class $\overline{A}_p^j$ from the last remark of the previous section, we find that each term of the previous sum can be bounded as follows:
$$\left\|U^{-\eta}\Op_{\hbar}(\overline{A}^{j}_p)U^{\eta}-\sum_{r=0}^{N-p}\hbar^r\Op_{\hbar}(\overline{A}^j_{p,r}(\eta))\right\|_{\mathcal{L}(L^2(M))}\leq C_{N,p}\hbar^{(N+1)(1-\nu)-3p\nu-D\nu}.$$
This implies that $R^{\text{Egorov},j}_N=\mathcal{O}_N(\hbar^{(N+1)(1-2\nu)-D\nu})$ (as $\nu<1/2$). Finally, it tells us that $R_N^j=R_N^{j-1}+\mathcal{O}_N(\hbar^{(N+1)(1-2\nu)-D'\nu})$, for some fixed integer $D'$. By induction on $j$, we find that
$$\left\|\Op_{\hbar}(Q_{\gamma_1})\cdots \Op_{\hbar}(Q_{\gamma_j})(-(j-1)\eta)-\Op_{\hbar}(A^j(\eta))(-j\eta)\right\|_{\mathcal{L}(L^2(M))}=\mathcal{O}_N(j\hbar^{(N+1)(1-2\nu)-D'\nu}).$$
As $j=\mathcal{O}(|\log\hbar|)$ and as $\nu<1/2$, we find that, for large $N$, the remainder tends to $0$ as $\hbar$ tends to $0$. This concludes the proof of theorem~\ref{pdotheo}.$\square$

\appendix

\section{Pseudodifferential calculus on a manifold}

\label{appendix}

In this appendix, a few facts about pseudodifferential calculus on a manifold and the sharp energy cutoff used in this paper are recalled. Even if most of this setting can be found in~\cite{AN2}, it is recalled because it is extensively used in section~\ref{mainproof} and~\ref{bigpdotheo}. The results from the two first sections of this appendix can be found in more details in~\cite{SZ} or~\cite{AN2}. The results of the last section of this appendix are the extension to the case of a manifold of standard results from semiclassical analysis that can be found either in~\cite{BR},~\cite{DS} or~\cite{EZ}.

\subsection{Pseudodifferential calculus on a manifold}

\label{pdo}

We start this appendix by recalling some facts of $\hbar$-pseudodifferential calculus that can be found in~\cite{DS} (or in~\cite{EZ}). Recall that we define on $\mathbb{R}^{2d}$ the following class of symbols:
$$S^{m,k}(\mathbb{R}^{2d}):=\left\{a_{\hbar}(x,\xi)\in C^{\infty}(\mathbb{R}^{2d}\times(0,1]):|\partial^{\alpha}_x\partial^{\beta}_{\xi}a_{\hbar}|\leq C_{\alpha,\beta}\hbar^{-k}\langle\xi\rangle^{m-|\beta|}\right\}.$$
Let $M$ be a smooth Riemannian $d$-manifold without boundary. Consider a smooth atlas $(f_l,V_l)$ of $M$, where each $f_l$ is a smooth diffeomorphism from $V_l\subset M$ to a bounded open set $W_l\subset\mathbb{R}^{d}$. To each $f_l$ correspond a pull back $f_l^*:C^{\infty}(W_l)\rightarrow C^{\infty}(V_l)$ and a canonical map $\tilde{f}_l$ from $T^*V_l$ to $T^*W_l$:
$$\tilde{f}_l:(x,\xi)\mapsto\left(f_l(x),(Df_l(x)^{-1})^T\xi\right).$$
Consider now a smooth locally finite partition of identity $(\phi_l)$ adapted to the previous atlas $(f_l,V_l)$. That means $\sum_l\phi_l=1$ and $\phi_l\in C^{\infty}(V_l)$. Then, any observable $a$ in $C^{\infty}(T^*M)$ can be decomposed as follows: $a=\sum_l a_l$, where $a_l=a\phi_l$. Each $a_l$ belongs to $C^{\infty}(T^*V_l)$ and can be pushed to a function $\tilde{a}_l=(\tilde{f}_l^{-1})^*a_l\in C^{\infty}(T^*W_l)$. As in~\cite{DS}, define the class of symbols of order $m$ and index $k$
\begin{equation}\label{defpdo}S^{m,k}(T^{*}M):=\left\{a_{\hbar}\in C^{\infty}(T^*M\times(0,1]):|\partial^{\alpha}_x\partial^{\beta}_{\xi}a_{\hbar}|\leq C_{\alpha,\beta}\hbar^{-k}\langle\xi\rangle^{m-|\beta|}\right\}.\end{equation}
Then, for $a\in S^{m,k}(T^{*}M)$ and for each $l$, one can associate to the symbol $\tilde{a}_l\in S^{m,k}(\mathbb{R}^{2d})$ the standard Weyl quantization
$$\Op_{\hbar}^{w}(\tilde{a}_l)u(x):=\frac{1}{(2\pi\hbar)^d}\int_{R^{2d}}e^{\frac{\imath}{\hbar}\langle x-y,\xi\rangle}\tilde{a}_l\left(\frac{x+y}{2},\xi;\hbar\right)u(y)dyd\xi,$$
where $u\in\mathcal{S}(\mathbb{R}^d)$, the Schwartz class. Consider now a smooth cutoff $\psi_l\in C_c^{\infty}(V_l)$ such that $\psi_l=1$ close to the support of $\phi_l$. A quantization of $a\in S^{m,k}$ is then defined in the following way:
\begin{equation}\label{pdomanifold}\Op_{\hbar}(a)(u):=\sum_l \psi_l\times\left(f_l^*\Op_{\hbar}^w(\tilde{a}_l)(f_l^{-1})^*\right)\left(\psi_l\times u\right),\end{equation}
where $u\in C^{\infty}(M)$. This quantization procedure $\Op_{\hbar}$ sends (modulo $\mathcal{O}(\hbar^{\infty})$) $S^{m,k}(T^{*}M)$ onto the space of pseudodifferential operators of order $m$ and of index $k$, denoted $\Psi^{m,k}(M)$~\cite{DS}. It can be shown that the dependence in the cutoffs $\phi_l$ and $\psi_l$ only appears at order $2$ in $\hbar$ (using for instance theorem $18.1.17$ in~\cite{Ho}) and the principal symbol map $\sigma_0:\Psi^{m,k}(M)\rightarrow S^{m,k}/S^{m,k-1}(T^{*}M)$ is then intrinsically defined. Most of the rules (for example the composition of operators, the Egorov and Calder\'on-Vaillancourt theorems) that holds in the case of $\mathbb{R}^{2d}$ still holds in the case of $\Psi^{m,k}(M)$. Because our study concerns behavior of quantum evolution for logarithmic times in $\hbar$, a larger class of symbols should be introduced as in~\cite{DS}, for $0\leq\nu<1/2$,
\begin{equation}\label{symbol}S^{m,k}_{\nu}(T^{*}M):=\left\{a_{\hbar}\in C^{\infty}(T^*M\times(0,1]):|\partial^{\alpha}_x\partial^{\beta}_{\xi}a_{\hbar}|\leq C_{\alpha,\beta}\hbar^{-k-\nu|\alpha+\beta|}\langle\xi\rangle^{m-|\beta|}\right\}.\end{equation}
Results of~\cite{DS} can be applied to this new class of symbols. For example, a symbol of $S^{0,0}_{\nu}$ gives a bounded operator on $L^2(M)$ (with norm uniformly bounded with respect to $\hbar$).\\
As was explained, one needs to quantize the sharp energy cutoff $\chi^{(.)}$ (see section~\ref{cut}) to get sharp bounds in~\ref{normest}. As $\chi^{(0)}$ localize in a strip of size $\hbar^{1-\delta_0}$ with $\delta_0$ close to $0$, the $m$-th derivatives transversally to $\mathcal{E}$ grows like $\hbar^{m(\delta_0-1)}$. As $\delta_0$ is close to $0$, $\chi^{(0)}$ does not belongs to the previous class of symbols that allows $\nu< 1/2$. However, as the variations only appears in one direction, it is possible to define a new pseudodifferential calculus for these symbols. The procedure taken from~\cite{SZ} is briefly recalled in~\cite{AN2} (section $5$) and introduces a class of anisotropic symbols $S^{-\infty,0}_{\mathcal{E},\nu'}$ (where $\mathcal{E}:=S^*M$ and $\nu'<1$) for which a quantization procedure $\Op_{\mathcal{E},\nu'}$ can be defined. In the next section, we recall briefly a few results about the quantization $\Op_{\mathcal{E},\nu'}(\chi^{(n)})$ of the symbol $\chi^{(.)}$.

\subsection{Energy cutoff}

\label{properties}

Let $\chi^{(.)}$ be as in section~\ref{cut}. Consider $\delta_0>0$ and $K_{\delta_0}$ associated to it (see section~\ref{cut}). Taking $\nu'=1-\delta_0$, it can be checked that the cutoffs defined in section~\ref{cut} belongs to the class $S^{-\infty,0}_{\mathcal{E},\nu'}$ defined in~\cite{AN2}. A pseudodifferential operator corresponding to it can be defined following the nonstandard procedure mentioned above. Using results from~\cite{AN2} (section $5$), one has $\|\Op_{\mathcal{E},\nu'}(\chi^{(n)})\|=1+\mathcal{O}(\hbar^{\nu'/2})$ for all $n\leq K_{\delta_0}|\log\hbar|$. For simplicity of notations, in the paper $\Op(\chi^{(n)}):=\Op_{\mathcal{E},\nu'}(\chi^{(n)})$. In~\cite{AN2}, it is also proved that
\begin{prop}\label{localization}~\cite{AN2} For $\hbar$ small enough and any $n\in\mathbb{N}$ such that $0\leq n\leq K_{\delta_0}|\log\hbar|$ and for any $\psi_{\hbar}=-\hbar^2\Delta\psi_{\hbar}$ eigenstate, one has
$$\|\psi_{\hbar}-\Op(\chi^{(n)})\psi_{\hbar}\|=\mathcal{O}(\hbar^{\infty})\|\psi_{\hbar}\|.$$
Moreover for any sequence $\alpha$ and $\beta$ of length $n$ less than $K_{\delta_0}|\log\hbar|$, one has
$$\left\|\left(1-\Op\left(\chi^{(n)}\right)\right)\tau_{\alpha}\Op\left(\chi^{(0)}\right)\right\|=\mathcal{O}(\hbar^{\infty})\ \ \ \ \ \ \left\|\left(1-\Op\left(\chi^{(n)}\right)\right)\pi_{\beta}\Op\left(\chi^{(0)}\right)\right\|=\mathcal{O}(\hbar^{\infty})$$
where $\tau$ and $\pi$ are given by~(\ref{tau}) and~(\ref{pi}).
\end{prop}
This proposition tells that the quantization of this energy cutoff exactly have the expected property, meaning that it preserves the eigenfunction of the Laplacian. So, in the paper, introducing the energy cutoff $\Op(\chi^{(n)})$ does not change the semiclassical limit. Moreover this proposition implies the following corollary that allows to apply theorem~\ref{uncertainty} in section~\ref{several}:
\begin{coro}\label{hyp}~\cite{AN2} For any fixed $L>0$, there exists $\hbar_L$ such that for any $\hbar\leq\hbar_L$, any $n\leq K_{\delta_0}|\log\hbar|$ and any sequence $\beta$ of length $n$, the Laplacian eigenstate verify
$$\left\|\left(1-\Op\left(\chi^{(n)}\right)\right)\pi_{\beta}\psi_{\hbar}\right\|\leq \hbar^L\|\psi_{\hbar}\|.$$
\end{coro}
A last property of the quantization of this cutoff that we can quote from~\cite{AN2} (remark $2.4$) is that we can restrict ourselves to study observables carried in a thin neighborhood around $S^*M=H^{-1}(1/2)$:
\begin{prop}\label{localization2}~\cite{AN2} For $\hbar$ small enough and any $n\in\mathbb{N}$ such that $0\leq n\leq K_{\delta_0}|\log\hbar|/2$, one has:
$$\forall|\gamma|=n,\ \|\tau_{\gamma}\Op(\chi^{(n)})-\tau_{\gamma}^f\Op(\chi^{(n)})\|=\mathcal{O}(\hbar^{\infty}),$$
where $P_{\gamma_j}^f=\Op_{\hbar}(P_{\gamma_j}f)$, $f$ is a smooth compactly supported function in a thin neighborhood of $\mathcal{E}$ and $\tau_{\gamma}^f=P_{\gamma_{n-1}}^f((n-1)\eta)\cdots P_{\gamma_{0}}^f.$
\end{prop}

\subsection{$\hbar$-expansion for pseudodifferential operators on a manifold}
The goal of this last section is to explain how the usual $\hbar$-expansion of order $N$ for composition of pseudodifferential operators and Egorov theorem can be extended in the case of pseudodifferential calculus on a manifold. The $\hbar$-expansion will depend on the partition of identity in section~\ref{pdo}. In fact, on a manifold, the formulas for the terms of order larger than $1$ on the $\hbar$-expansion will depend on the local coordinates. For simplicity and as it is the case of all the symbols we consider (thanks to the energy cutoff: for example, see proposition~\ref{localization2}), we now restrict ourselves to symbols supported in $\mathcal{E}^{\theta}=H^{-1}([1/2-\theta,1/2+\theta])$. The symbols are now elements of $S^{-\infty,0}_{\nu}(T^*M)$.

\subsubsection{Composition of pseudodifferential operators on a manifold}

\label{compositionpdo}

First, recall that the usual semiclassical theory on $\mathbb{R}^{d}$ (see~\cite{DS} or appendix of~\cite{BR}) tells that the composition of two elements $\Op_{\hbar}^w(a)$ and $\Op_{\hbar}^w(b)$ in $\Psi^{-\infty,k}_{\nu}(\mathbb{R}^{d})$ is still in $\Psi^{-\infty,k}_{\nu}(\mathbb{R}^{d})$ and that the essential support of its symbol is included in $\text{supp}(a)\cap\text{supp}(b)$. More precisely, it says that $\Op_{\hbar}^w(a)\circ \Op_{\hbar}^w(b)=\Op_{\hbar}^w(a\sharp b)$, where $a\sharp b$ is in $S^{-\infty,k}_{\nu}$ and its asymptotic expansion in power of $\hbar$ is given by the Moyal product
\begin{equation}\label{moyal2}a\sharp b(x,\xi)\sim\sum_{k}\frac{1}{k!}\left(\frac{\imath\hbar}{2}\omega(D_x,D_{\xi},D_y,D_{\eta})\right)^ka(x,\xi)b(y,\eta)|_{x=y,\xi=\eta},\end{equation}
where $\omega$ is the standard symplectic form. Outline that it is clear that each element of the sum is supported in $\text{supp}(a)\cap\text{supp}(b)$. As quantization on a manifold is constructed from quantization on $\mathbb{R}^{2d}$ (see definition~(\ref{pdomanifold})), one can prove an analogue of this asymptotic expansion in the case of a manifold $M$ (except that it will not be intrinsically defined). Precisely, let $a$ and $b$ be two symbols in $S^{-\infty,0}_{\nu}(T^{*}M)$. For a choice of quantization $\Op_{\hbar}$ (that depends on the coordinates maps), one has $\Op_{\hbar}(a)\circ\Op_{\hbar}(b)$ is a pseudodifferential operator in $\Psi^{-\infty,0}_{\nu}(M)$. Its symbol (mod $\mathcal{O}(\hbar^{\infty})$) is denoted $a\sharp_M b$ and its asymptotic expansion is of the following form:
$$a\sharp_M b\sim\sum_{p\geq 0}\hbar^{p} (a\sharp_M b)_{p}.$$
In the previous asymptotic expansion, $(a\sharp_M b)_{p}$ is a linear combination (that depends on the cutoffs and the local coordinates) of elements of the form $\partial^{\gamma}a\partial^{\gamma'}b$ with $|\gamma|\leq p$ and $|\gamma'|\leq p$. As a consequence, $(a\sharp_M b)_{p}$ is an element of $S^{-\infty,2p\nu}_{\nu}(T^*M)$.
\begin{rema} We know that we have an asymptotic expansion so by definition and using Calder\'on-Vaillancourt theorem, we know that each remainder is bounded in norm by a constant which is a small power of $\hbar$ (in fact $C\hbar^{(N+1)(1-2\nu)}$ for the remainder of order $N$). In our analysis, we need to know precisely how these bounds depend on $a$ and $b$ as we  have to make large product of pseudodifferential operators (see section~\ref{bigpdotheo}) and to use the composition formula to get Egorov theorem (see next section). The following lines explain how the remainder in the asymptotic expansion in powers of $\hbar$ is bounded by the derivatives of $a$ and $b$.\end{rema}
In the appendix of~\cite{BR}, they defined the remainder of the order $N$ expansion, in the case of $\mathbb{R}^{2d}$,
$$\hbar^{N+1}R_{N+1}(a,b,\hbar):=a\sharp b-\sum_{p=0}^N\hbar^p(a\sharp b)_p$$
and, using a stationary phase argument, they get the following estimates on the remainder, for all $\gamma$ and all $N$,
$$|\partial^{\gamma}_{z}R_{N+1}(a,b,z,\hbar)|\leq\rho_{d}K_{d}^{N+|\gamma|}(N!)^{-1}\sup_{(*)}|\partial_u^{(\alpha,\beta)+\mu}a(u+z)||\partial_v^{(\beta,\alpha)+\nu}b(v+z)|,$$
where $(*)$ means
$$u,v\in\mathbb{R}^{2d}\times\mathbb{R}^{2d},\ |\mu|+|\nu|\leq 4d+|\gamma|,\ |(\alpha,\beta)|=N+1,\ \alpha,\beta\in\mathbb{N}^d.$$
Applying Calder\'on-Vailancourt theorem (see~\cite{DS}-theorem $7.11$), one knows that there exist a constant $C$ and a constant $D$ (depending only on $d$), such that for a symbol $a$ in $S^{0,0}_{\mathbb{R}^{2d}}(1)$:
$$\|\Op_{\hbar}^{w}(a)\|_{L^2}\leq C\sup_{|\alpha|\leq D}\hbar^{\frac{|\alpha|}{2}}\|\partial^{\alpha}a\|_{\infty}.$$
Combining this result with the previous estimates on the $R^{(N+1)}$, one finds that
\begin{equation}\label{remmoyal}\|\Op_{\hbar}^{w}(R_{N+1}(a,b,z,\hbar))\|_{L^2}\leq C(d,N)\sup_{(*)}\hbar^{\frac{|\alpha|}{2}}\|\partial^{\beta+\beta'}a\|_{\infty}\|\partial^{\gamma+\gamma'}b\|_{\infty},\end{equation}
where $(*)$ means
$$|\alpha|\leq C',\ |\beta|\leq N+1,\ |\gamma|\leq N+1\ \text{and}\ |\beta'|+|\gamma'|\leq C+|\alpha|.$$
The constants $C$ and $C'$ depend only on the dimension $d$. The same kind of estimates holds on the remainder in the asymptotic expansion for change of variables. As the asymptotic expansion for composition of pseudodifferential operators is obtained from the composition and variable change rules on $\mathbb{R}^{2d}$~\cite{Ho} (theorem $18.1.17$; see also~\cite{EZ}-chapter $8$), the previous estimates~(\ref{remmoyal}) hold for semiclassical analysis on a manifold.

\subsubsection{Egorov expansion on a manifold}

\label{fixedegorov}

In this section, we want to recall how we prove an Egorov property with an expansion of any order. We follow the proof from~\cite{BR}. First, for the order $0$ term, we write the following exact expression for a symbol $a$ in $S^{-\infty,0}_{\nu}(T^*M)$,
\begin{equation}\label{step1}U^{-t}\Op_{\hbar}(a)U^t-\Op_{\hbar}(a(t))=\hbar\int_0^tU^{-s}(R^{(1)}(t-s))U^s ds,\end{equation}
where $a(t):=a\circ g^t$, $H(\rho)=\frac{\|\xi\|_x^2}{2}$ is the Hamiltonian and $$R^{(1)}(t-s):=\frac{1}{\hbar}\left(\frac{\imath}{\hbar}\left[-\frac{\hbar^2\Delta}{2},\Op_{\hbar}(a(t))\right]-\Op_{\hbar}(\{H,a(t)\})\right).$$ According to the rules of pseudodifferential calculus described in the previous section, we know that there exists some constants such that
$$\|R^{(1)}(t-s)\|_{\mathcal{L}(L^2(M))}\leq C(M,1)\sup_{0\leq s\leq t,|\alpha|\leq D,|\beta|\leq 1+D+|\alpha|}\hbar^{\frac{|\alpha|}{2}}\|\partial^{\beta}(a(s))\|_{\infty},$$
where $D$ depends only on the dimension of the manifold and $C(M,1)$ depnds on the choice of coordinates on the manifold. We proceed then by induction to recover the terms of higher order. For these higher order terms, we will see terms depending on the local coordinates appear in the expansion and we will obtain expressions as in~\cite{BR} for the higher order terms of the expansion that will be different from the case of $\mathbb{R}^{d}$~\cite{BR}. However, we do not need to have an exact expression for each term of the expansion: we only need to know on how many derivatives the order $p$ term depends and how the remainder can be bounded at each step. To obtain, the $\hbar$ formal term of the Egorov expansion, we first outline that $R^{(1)}(t-s)$ is a pseudodifferential operator whose asymptotic expansion is given by the composition rules on a manifold (see previous section). One can compute its principal symbol and verify that it is a linear combination (depending on the manifold and on the choice of coordinates) of derivatives of $a\circ g^{t-s}:=a_0(t-s)$ of order at most $2$. We denote $\left\{H,a_0(t-s)\right\}_{M}^{(1,0)}$ its principal symbol. Then, we can apply the same procedure as in equation~(\ref{step1}) to get the exact expression
$$\Op_{\hbar}(a)(t)=\Op_{\hbar}(a^{(1)}(t))+\hbar^2\int_0^tU^{-s}R^{(2)}(t-s)U^sds.$$
where
$$a^{(1)}(t):=a\circ g^t+\hbar\int_0^t(\left\{H,a_0(t-s)\right\}_{M}^{(1,0)})\circ g^sds.$$
We denote the previous formula in a more compact way
$$a^{(1)}(t):=a_0(t)+\hbar a_1(t),$$
where $\displaystyle a_1(t,\rho):=\int_0^t\left\{H,a_0(t-s)\right\}_{M}^{(1,0)}\left(g^{s}(\rho)\right)ds$. As was mentioned, this generalized `bracket' is a linear combination depending on the devivatives of order at most $2$ of $a_{t-s}$ (it also depends on $H$, $M$ and the choice of the quantization procedure). The operator norm of the remainder $R^{(2)}$ is, once more, controlled by the derivatives of $a_0(t)$ and $a_1(t)$. Precisely, one has
$$\|R^{(2)}(t)\|_{\mathcal{L}(L^2(M))}\\
\leq C(M,2)\sup_{(*)}\hbar^{\frac{|\alpha|}{2}}\|\partial^{\beta}\left(a_j(s)\right)\|_{\infty},$$
where $C(M,2)$ depends on the manifold $M$ (and on the choice of the quantization procedure) and $(*)$ means
$$j\leq 1,\ 0\leq s\leq t,\ |\alpha|\leq D,\ |\beta|\leq 2-j+D+|\alpha|.$$
Suppose the terms of order less than $p$, i.e. $a_0(t)$, ..., $a_{p-1}(t)$, are constructed. Then, we want to construct the term of order $p$. There will be several contributions. First, we write that the symbol (up to $\mathcal{O}(\hbar^{\infty}$) of $R^{(1)}(t-s)$ has an asymptotic expansion where the term of order $p-1$ depends on at most $p+1$ derivatives of $a_0(t-s)$. We can apply~(\ref{step1}) to this term of order $p-1$ and it will provide a symbol in $S^{-\infty,-p+(p+1)\nu}_{\nu}(T^*M)$ that we denote $\hbar^p\{H,a_0(t-s)\}^{(p,0)}$. Using the same procedure for every $a_j$ (where $j\leq p-1$), we can show finally that for any order $N$,
$$\Op_{\hbar}(a)(t)=\Op_{\hbar}(a^{(N)}(t))+\hbar^{N+1}\int_0^tU^{-s}R^{(N+1)}(t-s)U^sds.$$
In the previous formula, $a^{(N)}(t)$  is defined as follows:
$$a^{(N)}(t):=\sum_{p=0}^{N}\hbar^{p}a_{p}(t)\
\text{where}\ \ a_0(t):=a\circ g^t$$
and for $1\leq p\leq N$,
$$a_p(t,\rho):=\sum_{j=0}^{p-1}\int_0^t\left\{H,a_j(t-s)\right\}_{M}^{(p,j)}\left(g^{s}(\rho)\right)ds,$$
where $\{.,.\}_{M}^{(p,j)}$ is a generalized 'bracket' of order $(p,j)$ depending on the local coordinates on the manifold (it is the analogue of formula given by theorem $1.2$ in~\cite{BR}). We do not need to have an exact expression for these brackets: we only need to know on how many derivatives it depends. From the previous section, we know how the order $p$ term in the expansion of $a\sharp_M b$ depends linearly on products of the $p$ derivatives of $a$ and $b$. The term $\{H,a_0(t-s)\}^{(p,0)}$ comes from the order $p-1$ term of the asymptotic expansion of the symbol of $R^{(1)}(t-s)$. According to the rules of composition of pseudodifferential operators on a manifold, it is a linear combination (depending on $H$ and the choice of coordinates) of derivatives of $a$ of order at most $p+1$. More generally, $\{H,a_j(t-s)\}_{M}^{(p,j)}$ is a linear combination of derivatives of $a_j(t)$ of order at most $p+1-j$. For the remainder term $R^{(N+1)}(s)$ of order $N$, using the formulas for the composition of pseudodifferential operators, one can control it by the derivatives of the lower terms of the expansion. The previous discussion can be summarized in the following proposition:
\begin{prop}[Egorov expansion on a manifold]\label{exactegorov} Let $a$ be a symbol in $S^{-\infty,0}_{\nu}(T^*M)$. One has the exact expression for every $N\geq 0$,
\begin{equation}\label{exact1}\Op_{\hbar}(a)(t)=\Op_{\hbar}(a^{(N)}(t))+\hbar^{N+1}\int_0^tU^{-s}R^{(N+1)}(t-s)U^sds.\end{equation}
In the previous formula, one has
$$a^{(N)}(t):=\sum_{p=0}^{N}\hbar^{p}a_{p}(t)\
\text{where}\ \ a_0(t):=a\circ g^t$$
and for $1\leq p\leq N$,
$$a_p(t,\rho):=\sum_{j=0}^{p-1}\int_0^t\left\{H,a_j(t-s)\right\}_{M}^{(p,j)}\left(g^{s}(\rho)\right)ds.$$
For each $0\leq j\leq p-1$, $\{H,a_j(t-s)\}_{M}^{(p,j)}$ is a linear combination of derivatives of $a_j(t-s)$ of order at most $p+1-j$ that depends on the choice of coordinates on the manifold. Finally, the norm of $R^{(N+1)}(t)$ satisfies the following bound:
\begin{equation}\label{remainderestimates}\|R^{(N+1)}(t))\|_{L^2}\\
\leq C(M,N)\sup_{(*)}\hbar^{\frac{|\alpha|}{2}}\|\partial^{\beta}\left(a_p(s)\right)\|_{\infty},\end{equation}
where $C(M,N)$ depends on $N$ and on the manifold $M$ (also on the choice of coordinates) and where $(*)$ means:
$$p\leq N,\ 0\leq s\leq t,\ |\alpha|\leq D,\ |\beta|\leq N+1-p+D+|\alpha|.$$
The constant $D$ depends only on the dimension of the manifold.
\end{prop}
\begin{rema} Theorem $1.2$ in~\cite{BR} gives an exact expression of each term of this exact expansion in the case of $\mathbb{R}^{2d}$. We also mention that if $a$ is in the class $S^{-\infty,0}_{\nu}(T^*M)$, then each term of the expansion $a_p$ is in the class $S^{-\infty,2p\nu}_{\nu}$.\end{rema}
Finally, we underline that, by an induction argument, one can derive the following corollary:
\begin{coro}\label{number-derivatives} Using the notations of proposition~\ref{exactegorov}, one has that every $a_p(t)$ depends linearly on the derivatives of order at most $2p$ of $a$.
\end{coro}

\end{document}